\documentclass[journal,onecolumn]{IEEEtran}\usepackage{color}
\usepackage{tikz}

\usepackage{etoolbox}
\usepackage{etoolbox}

\usepackage{tikz,amsthm,amsmath,amssymb, colortbl,algorithm,algorithmic}
\usepackage{pdfpages,graphicx,epsfig,caption,subcaption,cite}

\usetikzlibrary{patterns}
\usepackage{psfrag,verbatim,relsize,ulem}
\usetikzlibrary{shapes,arrows}
\usetikzlibrary{backgrounds}

\usepackage{color}

\newcommand{\rank}{\textrm{rank}}
\newcommand{\card}[1]{|{#1}|}

\DeclareMathAlphabet\mathbfcal{OMS}{cmsy}{b}{n}

\newcommand{\e}{\mathcal{E}}

\newtheorem{theorem}{Theorem}

\newtheorem{lemma}{Lemma}
\newtheorem{proposition}{Proposition}
\newtheorem{corollary}{Corollary}
\newtheorem{definition}{Definition}
\newtheorem{remark}{Remark}

\usepackage{color}

  \newtheorem{examplex}{Example}

\newenvironment{example}
 {\pushQED{\qed}\examplex }
  {\popQED\endexamplex}
  {}
  \AtBeginEnvironment{example}{\renewcommand{\itshape}{\upshape}}

\floatname{algorithm}{Procedure}

\ifCLASSINFOpdf
\else
\fi

\begin{document}
%
\title{Capacity Results for Multicasting Nested Message Sets over Combination Networks}
%
%
%

\author{Shirin~Saeedi~Bidokhti,
         Vinod~M.~Prabhakaran, Suhas~N.~Diggavi
\thanks{S. Saeedi Bidokhti was with the School of Computer and Communication Sciences, Ecole Polytechnique F\'ed\'eral de Lausanne. She is now with the Department of Electrical Engineering department for Communications Engineering, Technische Universit\"{a}t M\"{u}nchen
 (e-mail: shirin.saeedi@tum.de).}
\thanks{V. Prabhakaran is with the School of Technology and Computer Science,
Tata Institute of Fundamental Research
(e-mail: vinodmp@tifr.res.in).}
\thanks{S. Diggavi is with the Department of Electrical Engineering, UCLA, USA
e-mail: (suhas@ee.ucla.edu).}
\thanks{
The work of S. {Saeedi Bidokhti} was partially supported by the SNSF fellowship no.146617. The work of Vinod M. Prabhakaran was partially supported by
a Ramanujan Fellowship from the Department of Science and Technology,
Government of India. The work of S. Diggavi was supported in part by NSF awards 1314937 and 
1321120. }
\thanks{The material in this paper
was presented in parts at the 2009 IEEE Information Theory Workshop, Taormina, Italy,  2012 IEEE Information Theory Workshop, Lausanne, Switzerland, and the 2013 IEEE International Symposium on Information Theory, Istanbul, Turkey.}
}

\maketitle

\begin{abstract}
The problem of multicasting two nested messages is studied over a class of networks known as combination networks.  A source multicasts two messages, a common and a private message, to several receivers. A subset of the receivers (called the public receivers) only demand the common message and the rest of the receivers (called the private receivers) demand both the common and the private message. Three encoding schemes are discussed that employ linear superposition coding and their optimality is proved in special cases. 
The standard linear superposition scheme is shown to be optimal for networks with two public receivers and any number of private receivers. When the number of public receivers increases, this scheme stops being optimal. 
Two improvements are discussed: one using pre-encoding at the source, and one using a block Markov encoding scheme. 
The rate-regions that are achieved by the two schemes are characterized in terms of feasibility problems.   Both inner-bounds are shown to be the capacity region for networks with three (or fewer) public and any number of private receivers. 
Although the inner bounds are not comparable in general, it is shown through an example that the  region achieved by the block Markov encoding scheme  may strictly include the region achieved by the pre-encoding/linear superposition scheme.
Optimality results are founded on the general framework of  Balister and Bollob\'{a}s (2012) for sub-modularity of the entropy function. An equivalent graphical representation is introduced and a lemma is proved that might be of independent interest.

{Motivated by the connections between combination networks and broadcast channels, a new block Markov encoding scheme is proposed for broadcast channels with two nested messages. The rate-region that is obtained includes the previously known rate-regions. It remains open whether this inclusion is strict.}

\end{abstract}

\begin{IEEEkeywords}
Network Coding, Combination Networks, Broadcast Channels, Superposition Coding, Linear Coding, Block Markov Encoding
\end{IEEEkeywords}

%
\IEEEpeerreviewmaketitle

\allowdisplaybreaks

\section{Introduction}
The problem of communicating common and individual messages over general networks has been unresolved over the past several decades, even in the two extreme cases of single-hop broadcast channels and multi-hop wireline networks.
Special cases have been studied where the capacity region is fully characterized (see \cite{Cover95} and the references therein, and \cite{AhlswedeCaiLiYeung00, WeingartenSteinbergShamai06, AvestimehrDiggaviTse11, GengNair14}). Inner and outer bounds on the capacity region are derived in \cite{Marton79, GohariElGamalAnantheram10, NairElGamal07, LiangKramerPoor11,  ChanGrant07,Yeung08, SongYeungCai03}. Moreover, new encoding and decoding techniques are developed such as superposition coding \cite{Cover72, Bergmans73}, Marton's coding \cite{Marton77, Pinsker78, KornerMarton77}, network coding \cite{AhlswedeCaiLiYeung00, LiYeungCai03, KoetterMedard03}, and joint unique and non-unique decoding \cite{HanKobayashi81, ChongMotaniGargElGamal08, NairElGamal09, SaeediPrabhakaranDiggavi12}.

Surprisingly, the problem of broadcasting nested (degraded) message sets has been completely resolved for two users. The problem was introduced and investigated for two-user broadcast channels in \cite{KornerMarton77} and it was shown that superposition coding was optimal. The problem has also been investigated for wired networks with two users \cite{RamamoorthyWesel09,NgaiYeung04a,ErezFeder03} where a scheme consisting of  routing and random network coding turns out to be rate-optimal. This might suggest that the nested structure of the messages makes the problem easier in nature. Unfortunately, however, the problem has remained open for networks with more than two receivers and only some special cases are understood \cite{NairElGamal09, DiggaviTse06, BoradeZhengTrott07, PrabhakaranDiggaviTse07, SaeediDiggaviFragouliPrabhakaran09}. The state of the art is not favourable for wired networks either.  Although extensions of the joint routing and network coding scheme in \cite{RamamoorthyWesel09}  are optimal for special classes of three-receiver networks (e.g., in \cite{GheorghiuSaeediFragouliToledo11}), they are suboptimal in general depending on the structure of the network \cite[Chapter 5]{Saeedi12}, \cite{GLT}.  It was  recently shown in \cite{ChanGrant12} that the problem of multicasting two nested message sets over general wired networks is as hard as the general network coding problem. 

In this paper, we study nested message set broadcasting over a class of wired networks known as combination networks \cite{NgaiYeung04}. These networks also form a resource-based model for broadcast channels and are  a special class of linear deterministic broadcast channels that were studied in \cite{PrabhakaranDiggaviTse07, SaeediDiggaviFragouliPrabhakaran09}. Lying at the intersection of  multi-hop wired networks and single-hop broadcast channels, combination networks are an interesting class of networks to build up intuition and understanding, and develop new encoding schemes applicable to both sets of problems. 

We study the problem of multicasting two nested message sets (a common message and a private message) to multiple users.  A subset of the users (public receivers) only demand the common message and the rest of the users (private receivers) demand both the common message and the private message. The term private does not imply any security criteria in 
this paper. 


Combination networks turn out to be an interesting class of networks that allow us to discuss new encoding schemes and obtain new capacity results. We discuss three encoding schemes and prove their optimality  in several cases (depending on the number of public receivers, irrespective of the number of private receivers). In particular, we propose a block Markov encoding scheme that outperforms schemes based on rate splitting and linear superposition coding. Our inner bounds are expressed in terms of feasibility problems and are easy to calculate.  To illustrate the implications of our approach over broadcast channels, we generlize our results and  propose a block Markov encoding scheme for  broadcasting two nested message sets over general broadcast channels (with multiple 
public and private receivers). The rate-region that is obtained includes  previously known rate-regions.

\subsection{Communication Setup}

A combination network is a three-layer directed network with one source and multiple destinations.
It consists of a source node in the first layer, $d$ intermediate nodes in the second layer and $K$ destination nodes in the third layer. 
The source is connected to all intermediate nodes, and each intermediate node is connected to a subset of the destinations. We refer to the outgoing edges of the source as the \textit{resources} of the combination network, see Fig.~\ref{blockdiagram}. We assume that each edge in this network carries one symbol from a finite field $\mathbb{F}_q$. We express all rates in symbols per channel use ($\log_2q$ bits per channel use) and thus all edges are of unit capacity.

The communication setup is shown in Fig.~\ref{blockdiagram}.
A source multicasts a common message $W_1$ of $nR_1$ bits and a private message $W_2$ of $nR_2$ bits. $W_1$ and $W_2$ are independent.
The common message is to be reliably decoded at all destinations, and the private message is to be reliably 
decoded at a subset of the destinations. We refer to those destinations who demand both messages as  \textit{private receivers} and to those  who demand only the common message as 
 \textit{public receivers}. 
We denote the number of  public receivers by $m$ and assume, without  loss of generality,   that they are indexed  $1,\ldots,m$.  The set of all public receivers is denoted by 
$\mathcal{I}_1=\{1,2,\ldots,m\}$ and the set of all private receivers is denoted by $\mathcal{I}_2=\{m+1,\ldots,K\}$.
\begin{figure}
\centering
\begin{tikzpicture}[scale=2]
\tikzstyle{every node}=[draw,thick,shape=circle,minimum size=.05cm,font=\small]; 
\path (1.3,2.9cm) node[draw=none] (w) {$W_1,W_2$}; 
\path (1,2.5cm) node[shape=rectangle,minimum height=.75cm] (v0) {Encoder}; 
\path (-1,1.75)  node (v2) {};
\path(.25,1.75)  node (v3) {}; 
\path(1.75,1.75)  node (v6) {}; 
\path (3,1.75) node (v4) {}; 
\path (-.5,0) node[shape=rectangle,minimum height=.75cm] (d1) {Decoder $1$};
\path (1,0) node[shape=rectangle,minimum height=.75cm] (d2) {Decoder $2$};
\path (2.5,0) node[shade,shape=rectangle,minimum height=.75cm] (d3) {Decoder $3$};
\path (1.25,-.4) node[draw=none] () {$\hat{W}^{(2)}_1$};
\path (-.25,-.4) node[draw=none] () {$\hat{W}^{(1)}_1$};
\path (3,-.4) node[draw=none] () {$\hat{W}^{(3)}_1\!,\hat{W}^{(4)}_2$};
\path (-1.25,.425cm) node[draw=none] () {$Y_1^n\!\!=\!\!\left[\!\!\!\!\begin{array}{c}X_1^n\\X_3^n\\X_2^n\end{array}\!\!\!\!\right]$}; 
\path (.5,.425cm) node[draw=none] () {$Y_2^n\!\!=\!\!\left[\!\!\!\!\begin{array}{c}X_3^n\\X_2^n\end{array}\!\!\!\!\right]$}; 
\path (3.25,.425cm) node[draw=none] () {$Y_3^n\!\!=\!\!\left[\!\!\!\!\begin{array}{c}X_1^n\\X_2^n\\X_4^n\end{array}\!\!\!\!\right]$}; 

\draw[->,thick] (v0) -- node[left,draw=none,fill=none] {$X_1^n$}  (v2) ;
\draw[->,thick] (v0) -- node[right,draw=none,fill=none] {$X_2^n$}   (v6) ;
\draw[->,thick] (v0) -- node[right,draw=none,fill=none] {$X_3^n$}  (v3) ;
\draw[->,thick] (v0) -- node[right,draw=none,fill=none] {$X_4^n$}  (v4) ;


\draw[->,thick] (1,3) -- (v0);
\draw[->,thick] (d2) -- (1,-.4);
\draw[->,thick] (d1) -- (-.5,-.4);
\draw[->,thick] (d3) -- (2.5,-.4);

\draw[->,thick] (v3) -- (d1);
\draw[->,thick] (v3) -- (d2);
\draw[->,thick] (v6) -- (d1);
\draw[->,thick] (v6) -- (d2);
\draw[->,thick] (v6) -- (d3);
\draw[->,thick] (v2) -- (d1);
\draw[->,thick] (v2) -- (d3);
\draw[->,thick] (v4) -- (d3);
\end{tikzpicture}
\caption{A combination network with two public receivers indexed by $I_1=\{1,2\}$ and one private receiver indexed by $I_2=\{3\}$. All edges are of  unit capacity.}
\label{blockdiagram}
\end{figure}

Encoding is done over blocks of length $n$. The encoder \textit{encodes} $W_1$ and $W_2$ 
into $d$ sequences $X_i^n$, $i=1,\ldots,d$, that are sent over the resources of the combination network (over $n$ uses of the network). 
Based on the structure of the network, each user $i$ receives a vector of sequences, $Y_i^n$, that is a certain collection of sequences that were sent by the source.
{Given $Y_i^n$, public 
receiver $i$, $i=1,\ldots,m$, \textit{decodes} $\hat{W}^{(i)}_1$ and given $Y_p^n$, private receiver $p$, $p=m+1,\ldots,K$, \textit{decodes}  $\hat{W}^{(p)}_1,\hat{W}^{(p)}_2$.
A rate pair $(R_1,R_2)$ is said to be achievable if there is an encoding/decoding scheme that allows the error probability $$P_e=\Pr\left(\hat{W}^{(i)}_1\neq W_1 \text{ for some }i=1,\ldots,K\text{ or }\hat{W}^{(i)}_2\neq W_2\text{ for some }i=m+1,\ldots,K\right)$$ approach zero (as $n\to \infty$).
We call the closure of all achievable rate-pairs  the capacity region. Although we allow $\epsilon-$error probability in the communication scheme (for example in proving converse theorems), the achievable schemes that we propose are zero-error schemes. Therefore, in all cases where we characterize the capacity region, the $\epsilon-$error capacity region and the zero-error capacity region coincide. This is not surprising considering the deterministic nature of our channels/networks.}
\subsection{Organization of the Paper}
The paper is organized in eight sections. In Section \ref{sec-main}, we give an 
overview of the underlying challenges and our main ideas through several 
examples. Our notation is introduced in Section \ref{notation}. We study  linear 
encoding schemes that are based on rate splitting, linear superposition coding, 
and pre-encoding in Section \ref{CombNet-ach-rate splitting}. 
Section \ref{CombNet-ach-blockMarkov} proposes a block Markov encoding scheme 
for multicasting nested message sets, Section \ref{CombNet-outerbound} discusses 
 optimality results, and Section \ref{lb-sec-generalbc} generalizes the block 
Markov encoding scheme of Section \ref{CombNet-ach-blockMarkov} to general 
broadcast channels. We conclude in Section \ref{conclusion}.

\section{Main Ideas at a Glance}
\label{sec-main}
{The problem of multicasting messages to receivers which have (two) different demands over a shared medium (such as the 
combination network) is, in a sense, finding an optimal solution to a trade-off. On the one 
hand,  public receivers (which presumably have access to fewer resources) need information about the common message 
\textit{only} so that each can decode the common message. On the other hand, private receivers require complete 
information of both messages. It is, therefore, desirable from private receivers' point of view to have these messages  
jointly encoded (especially when there are multiple private receivers). This may be in contrast with  public receivers' 
decodability requirement. This tension is best seen through an example. 
Example \ref{lb-CombNet-ex-joint} below shows that an optimal encoding scheme should allow joint encoding of the common and private messages but in a restricted and controlled manner, so that  only \textit{partial information} about the private message is ``revealed" to the public receivers and decodability of the common message is ensured.}
\begin{example}
\label{lb-CombNet-ex-joint}
\begin{figure}
\centering
\begin{tikzpicture}[scale=1.75]
\tikzstyle{every node}=[draw,shape=circle,font=\small\itshape]; 
\path (1.5,2.5) node (v0) {$S$}; 
\path  (1.5,3.25) node[rectangle,draw=none] () {$\hspace{-.75cm}W_1=[w_{1,1}]$}; 
\path  (1.5,2.85) node[rectangle,draw=none] () {$W_2=[w_{2,1},w_{2,2}]$}; 

\path (0,1.25) node (v6) {};
\path (1,1.25) node (v5) {};
\path (3,1.25) node (v8) {};
\path (2,1.25)  node (v7) {};
\path (0,0) node (d1) {$D_1$};
\path (0,-.45) node[rectangle,draw=none] () {$W_1$};
\path (1,0) node (d2) {$D_2$};
\path (1,-.45) node[rectangle,draw=none] () {$W_1,w_{2,1}$};

\path (2,0) node[shade]  (d3) {$D_3$};
\path (2,-.45) node[rectangle,draw=none] () {$W_1,W_2$};

\path (3,0) node (d4)[shade]  {$D_4$};

\path (3,-.45) node[rectangle,draw=none] () {$W_1,W_2$};

\draw[->] (v0) -- (v5);
\draw[->] (v0) -- node[left,draw=none,fill=none,xshift=.1cm] {${w_{1,1}}$}  (v5) ;

\draw[->] (v0) -- (v6);
\draw[->] (v0) -- node[draw=none,fill=none,left] {$w_{1,1}+w_{2,1}$}  (v6) ;

\draw[->] (v0) -- (v7);
\draw[->] (v0) -- node[right,draw=none,fill=none] {$w_{2,1}$}  (v7) ;

\draw[->] (v0) -- (v8);
\draw[->] (v0) -- node[right,draw=none,fill=none] {$w_{2,2}$}  (v8) ;

\draw[->] (v5) -- (d1);
\draw[->] (v5) -- (d3);
\draw[->] (v8) -- (d3);
\draw[->] (v8) -- (d4);
\draw[->] (v5) -- (d4);
\draw[->] (v6) -- (d2);
\draw[->] (v6) -- (d3);
\draw[->] (v7) -- (d2);
\draw[->] (v7) -- (d4);

\end{tikzpicture}
\caption{The source multicasts $W_1=[w_{1,1}]$ and $W_2=[w_{2,1},w_{2,2}]$ (of rates $R_1=1$ and $R_2=2$) in $n=1$ channel use.}
\label{CombNet-example-joint}
\end{figure}
{Consider the combination network shown in Fig.~\ref{CombNet-example-joint} in $n=1$ channel use. The source communicates a common message $W_1=[w_{1,1}]$ and a private message $W_2=[w_{2,1},w_{2,2}]$ to four receivers. Receivers $1$ and $2$ are public receivers and receivers $3$ and $4$ are private receivers. Since Receivers $1$ and $2$ each have min-cuts less than $3$, they are not able to decode both the common and private messages. For this reason, random linear network
coding does not ensure decodability of $W_1$ at the public receivers. We note that to communicate $W_1,W_2$,  it is necessary that some partial information about the private 
message is revealed to public receiver $2$.}

{Split the private message into $w_{2,1}$ of rate $\alpha_{\{2\}}=1$ and $w_{2,2}$ of rate $\alpha_{\phi}=1$. $w_{2,2}$ is the part of $W_2$ that is not revealed to the public receivers and  $w_{2,1}$ is the part that is revealed to Receiver $2$ (but not Receiver $1$). By splitting $W_2$ into independent pieces, we make sure that only a part of $W_2$, in this case $w_{2,1}$, is revealed to Receiver $2$; In other words, only $w_{2,1}$ is encoded into sequences $X^n_i$'s that are received by Receiver $2$. 
A linear scheme based on this idea is illustrated in Fig. \ref{CombNet-example-joint} and could be represented as follows:
\begin{align}
\left[\begin{array}{c}X_1\\X_2\\X_3\\X_4\end{array}\right]=\left[\begin{array}{ccc}1&1&0\\1&0&0\\0&1&0\\0&0&1\end{array}\right]\left[\begin{array}{c}w_{1,1}\\w_{2,1}\\w_{2,2}\end{array}\right].
\end{align}
Note that by this construction we have ensured that the private receivers get a full rank transformation of all information symbols, and the public receivers get a full rank transformation of a subset of the information symbols (including the information symbol of $W_1$). }
\end{example}


{Our first coding scheme (Proposition \ref{CombNetp1-innerbound} in Section \ref{lb-CombNet-ZES}) builds on Example \ref{lb-CombNet-ex-joint} by splitting the private message into $2^m$ independent pieces (of rates $\alpha_{\mathcal{S}}$, $\mathcal{S}\subseteq\mathcal{I}$, to be optimized) and using linear superposition coding. The rate split parameters $\alpha_{\mathcal{S}}$ should be designed such that they satisfy several rank constraints (for different decoders to decode their messages of interest). We characterize the achievable rate-region by a linear program with no objective function (a feasibility problem). The solution of this linear program gives the optimal choice of $\alpha_{\mathcal{S}}$ for the scheme. We show that our scheme is optimal for combination networks with two public and any number of private receivers.}

{For networks with three or more public and any number of private receivers, the above scheme may perform sub-optimally.  It turns out that one may, depending on the structure of  resources, gain by introducing  dependency among the partial (private) information that is revealed to different subsets of public receivers.}
\begin{example}
\label{lb-CombNet-example2to3}
\begin{figure}
\centering
\begin{tikzpicture}[scale=1.5]
\tikzstyle{every node}=[draw,shape=circle,minimum size=.01cm,font=\small\itshape]; 
\path (2.5,2.5) node (v0) {$S$}; 
\path (2.55,3.5) node[rectangle,draw=none] () {$\hspace{-1.65cm}W_1=[]$}; 
\path (2.5,3) node[rectangle,draw=none] () {$\hspace{-.25cm}W_2=[w_{2,1},w_{2,2}]$}; 

\path (1,1.25) node (v5) {};
\path (2.5,1.25) node (v6) {};
\path (4,1.25)  node (v7) {};

\path (0,0) node  (d1) {$D_1$};
\path (1,0) node  (d2) {$D_2$};
\path (2,0) node   (d3) {$D_3$};
\path (3,0) node (d4)[shade]   {$D_4$};

\path (4,0) node (d5)[shade]  {$D_5$};
\path (5,0) node (d6)[shade]   {$D_6$};

\draw[->] (v0) --node[left,draw=none]{$X_1$} (v5);
\draw[->] (v0) --node[draw=none,left]{$X_2$} (v6);
\draw[->] (v0) --node[right,draw=none,left]{$X_3$} (v7);

\draw[->] (v0) -- (v5);
\draw[->] (v0) -- (v6);
\draw[->] (v0) -- (v7);

\draw[->] (v5) -- (d1);
\draw[->] (v6) -- (d2);
\draw[->] (v7) -- (d3);
\draw[->] (v5) -- (d4);
\draw[->] (v6) -- (d4);
\draw[->] (v5) -- (d5);
\draw[->] (v7) -- (d5);
\draw[->] (v6) -- (d6);
\draw[->] (v7) -- (d6);

\end{tikzpicture}
\caption{Multicasting $W_1$ of rate $0$ and $W_2$ of rate $2$. A multicast code such as $X_1=w_{2,1}$, $X_2=w_{2,2}$, $X_3=w_{2,1}+w_{2,2}$ ensures  achievability of $(R_1=0,R_2=2)$.}
\label{CombNet-example2to3}
\end{figure}
{Consider the combination network of Fig.~\ref{CombNet-example2to3} where destinations $1$, $2$, $3$ are public receivers and destinations $4$, $5$, $6$ are private receivers. The source wants to communicate a common message of rate $R_1=0$ (i.e., $W_1=\emptyset$) and a private message $W_2=[w_{2,1},w_{2,2}]$ of rate $R_2=2$.  Clearly this rate pair is achievable using the multicast code shown in Fig.~\ref{CombNet-example2to3} (or simply through a random linear network code).
However, the scheme we outlined before is incapable of supporting this rate-pair. This may be seen by looking at the linear program characterization of the scheme: $(R_1=0,R_2=2)$ is achievable by our first encoding scheme if there is a solution to the following feasibility problem (see Section \ref{CombNet-ach-rate splitting} for details of derivation).
\begin{align}
&\alpha_{\mathcal{S}}\geq 0,\quad \mathcal{S}\subseteq\{1,2,3\}\label{exampleregion1}\\
&\sum_{\mathcal{S}\subseteq\{1,2,3\}}\alpha_{\mathcal{S}}=2\\
&\text{for all $\{i,j,k\}$ that is a permutation of $\{1,2,3\}$:}\\
&\alpha_{\{i\}}+\alpha_{\{i,j\}}+\alpha_{\{i,k\}}+\alpha_{\{i,j,k\}}\leq 1\\
&0\leq \alpha_{\{i,j,k\}}\\
&0\leq \alpha_{\{j,k\}}+\alpha_{\{i,j,k\}}\\
&0\leq \alpha_{\{i,k\}}+\alpha_{\{j,k\}}+\alpha_{\{i,j,k\}}\\
&0\leq \alpha_{\{i,j\}}+\alpha_{\{i,k\}}+\alpha_{\{j,k\}}+\alpha_{\{i,j,k\}}\\
&1\leq \alpha_{\{k\}}+\alpha_{\{i,k\}}+\alpha_{\{j,k\}}+\alpha_{\{i,j,k\}}\\
&1\leq \alpha_{\{k\}}+\alpha_{\{i,j\}}+\alpha_{\{i,k\}}+\alpha_{\{j,k\}}+\alpha_{\{i,j,k\}}\\
&2 \leq \alpha_{\{j\}}+\alpha_{\{k\}}+\alpha_{\{i,j\}}+\alpha_{\{i,k\}}+\alpha_{\{j,k\}}+\alpha_{\{i,j,k\}}\\
&2\leq \alpha_{\{i\}}+\alpha_{\{j\}}+\alpha_{\{k\}}+\alpha_{\{i,j\}}+\alpha_{\{i,k\}}+\alpha_{\{j,k\}}+\alpha_{\{i,j,k\}}
\label{exampleregionlast}
\end{align}
By testing the feasibility of the above linear program, it can be seen that the above problem is infeasible (by Fourier-Motzkin elimination or using MATLAB). On a higher level, although the optimal scheme in Fig. \ref{CombNet-example2to3} reveals 
partial (private) information to the different subsets of the public receivers, this is not done by splitting the private 
message into independent pieces and there is a certain dependency structure between the symbols that are revealed to 
 receivers~$1$, $2$, and $3$. This is why our first linear superposition scheme does not support the rate-pair $(0,2)$.}

{We use this observation to modify the encoding scheme and achieve the rate pair $(0,2)$. First, pre-encode message $W_2$, through a pre-encoding matrix $\mathbf{P}\in{\mathbb{F}_q}^{3\times 2}$, into a \textit{pseudo private message} $W^\prime_2$ of larger ``rate":
\begin{align}
\left[\begin{array}{c}w^\prime_{2,1}\\w^\prime_{2,2}\\w^\prime_{2,3}\end{array}\right]&=\mathbf{P}\left[\begin{array}{c}w\prime_{2,1}\\w\prime_{2,2}\end{array}\right].
\end{align}
Then, encode $W^\prime_2$ using rate splitting and linear superposition coding:
\begin{align}
\label{lb-ex-st}
\left[\begin{array}{c}X_1\\X_2\\X_3\end{array}\right]&=\left[\begin{array}{ccc}1&0&0\\0&1&0 \\0&0&1 \end{array}\right]\left[\begin{array}{c}w^\prime_{2,1}\\w^\prime_{2,2}\\w^\prime_{2,3}\end{array}\right].
\end{align}}

{Suppose $\mathbf{P}$ is a MDS (maximum distance separable) code. Each private receiver is able to decode  two symbols out of the three symbols of $W^\prime_2$ and can thus decode $W_2$. It turns out that the achievable rate-region of this scheme is a relaxation of \eqref{exampleregion1}-\eqref{exampleregionlast} when $\alpha_{\phi}$ may be negative.}
\end{example}

{Our second approach (Theorem \ref{lb-CombNet-Theoremk2} in Section \ref{lb-CombNet-MES}) builds on Example \ref{lb-CombNet-example2to3} and the underlying idea is to allow dependency among the pieces of information that are revealed to different sets of public receivers by an appropriate pre-encoder that encodes the private message into a pseudo private message of a larger rate, followed by a linear superposition encoding scheme.
We  prove that the rate-region achieved by our second scheme is  tight for $m=3$ (or fewer) public receivers and any number of private receivers. To prove the converse, 
we first write an outer-bound on the rate-region which looks \textit{similar} to the inner-bound feasibility problem and is in terms of some entropy functions.
{Next, we use \textit{sub-modularity of entropy} to write a converse for every inequality of the inner bound. In this process, we develop a visual tool in the framework of \cite{BalisterBollobas07} to deal with the sub-modularity of entropy and prove a lemma that allows us show the tightness of the inner bound without explicitly solving its corresponding feasibility problem. }}

{Generalizing the pre/encoding scheme to networks with more than three public receivers is difficult because of the more involved dependency structure that might be needed, in a good code, among the partial (private) information pieces that are to be revealed to the subsets of public receivers. Therefore, we propose an alternative encoding scheme that captures these dependencies over sequential blocks, rather than  the structure of the {one-time} (one-block) code. This is done by devising a simple \textit{block Markov encoding scheme}. 
Below, we illustrate the main idea of the block Markov scheme by revisiting the combination network of Fig.~\ref{CombNet-example2to3}. }
\begin{example}
\label{CombNet-BME-example12-3}
{Consider the combination network in Fig.~\ref{CombNet-example2to3} over which we want to achieve the rate pair $(R_1=0,R_2=2)$. Our first code design using rate splitting and linear superposition coding (with no pre-encoding) was not capable of achieving this rate pair. Let us add one resource to this combination network and connect it to all the private receivers. This gives  an extended combination network, as shown in Fig.~\ref{CombNet-example2to3bmecode-final}, that differs from the original network only in one edge. This ``virtual'' 
resource is shown in Fig.~\ref{CombNet-example2to3bmecode-final} by a bold edge. 
One can verify that our basic linear superposition scheme achieves $(R_1=0,R_2^\prime=3)$
over this extended network by writing the corresponding linear program and finding a solution:
\begin{align}
&\alpha_{\{1\}}=\alpha_{\{2\}}=\alpha_{\{3\}}=1,\\
&\alpha_{\phi}=\alpha_{\{1,2\}}=\alpha_{\{1,3\}}=\alpha_{\{2,3\}}=\alpha_{\{1,2,3\}}=0.
\end{align} }

{Let the message $W^\prime_2=[w^\prime_{2,1},w^\prime_{2,2},w^\prime_{2,3}]$ be a \textit{pseudo} 
private message of larger rate ($R^\prime_2=3$) that is  communicated over 
the extended combination network (in one channel use), and let $X_1$, $X_2$, $X_3$, $X_{\phi}$ be the symbols that are sent over the extended combination network. 
One code design is given below. We will use this code to achieve  rate pair $(0,2)$ over the original network. 
\begin{align}
\label{lb-CombNet-bm-example1}
\begin{array}{l}
X_1=w^\prime_{2,1}\\
X_2=w^\prime_{2,2}\\
X_3=w^\prime_{2,3}\\
X_{\phi}=w^\prime_{2,1}+w^\prime_{2,2}+w^\prime_{2,3}.
\end{array}
\end{align}}

{Since the resource edge that carries $X_\phi$ is a virtual resource, we aim to emulate it through a block Markov encoding 
scheme. 
Using the code design of \eqref{lb-CombNet-bm-example1}, all information symbols 
($w^\prime_{2,1}$, $w^\prime_{2,2}$,  $w^\prime_{2,3}$) are decodable at all 
private receivers. One way to emulate the bold virtual resource is 
to send its information (the symbol carried over it) in the next time slot 
using one of the information symbols $w^\prime_{2,1}$, $w^\prime_{2,2}$,  $w^\prime_{2,3}$ 
that are to be communicated in the next time slot.}

{More precisely,
consider communication over $n$ transmission blocks, and let $(W_1[t],W^{\prime}_2[t])$ be the message pair that is  encoded in block $t\in\{1,\ldots,n\}$. 
In the $t^{\text{th}}$ block, encoding is done as suggested by the code in \eqref{lb-CombNet-bm-example1}. 
Nevertheless, to provide  private receivers with the information of $X_{\phi}[t]$ (as promised by the virtual resource), 
we  use  $w^{\prime}_{2,3}[t+1]$ in the next block to convey $X_{\phi}[t]$. 
Since this symbol is ensured to be decoded at the private receivers, it indeed emulates the virtual resource.
In the $n^{\text{th}}$ block, we simply encode $X_{\phi}[n-1]$ and directly communicate it with the private receivers. 
Upon receiving all the $n$ blocks at the receivers, we perform backward decoding \cite{WillemsMeulen85}. 
So in $n$ transmissions, we send $n-1$ symbols of $W_1$ and $2(n-1)+1$ new symbols of $W_2$ over the original combination 
network;  i.e., for $n\to \infty$, we achieve the rate-pair $(0,2)$. }
\begin{figure}
\centering
\begin{tikzpicture}[scale=1.5]
\tikzstyle{every node}=[draw,shape=circle,font=\small\itshape]; 
\path (2.5,2.5) node (v0) {$S$}; 

\path  (2.5,3.5) node[rectangle,draw=none] () {$W_1[t+1]=[]$}; 
\path  (2.5,3) node[rectangle,draw=none] (message) {$W_2^{\prime}[t+1]=[w^{\prime}_{2,1}[t+1],w^{\prime}_{2,2}[t+1],\textcolor{blue}{w^{\prime}_{2,3}[t+1]}]$}; 

\path (1,1.25) node (v5) {};
\path (2.5,1.25) node (v6) {};
\path (4,1.25)  node (v7) {};
\path (5.5,1.25)  node[fill=blue,blue] (v8) {};

\path (0,0) node  (d1) {$D_1$};
\path (1,0) node  (d2) {$D_2$};
\path (2,0) node   (d3) {$D_3$};
\path (3,0) node (d4)[shade]   {$D_4$};

\path (4,0) node (d5)[shade]  {$D_5$};
\path (5,0) node (d6)[shade]   {$D_6$};

\draw[->] (v0) --node[left,draw=none]{$X_1[t]$} (v5);
\draw[->] (v0) --node[left,draw=none,xshift=.1cm]{$X_2[t]$} (v6);
\draw[->] (v0) --node[left,draw=none]{$X_3[t]$} (v7);
\draw[->,very thick,blue] (v0) --node[right,draw=none,xshift=.05cm](virtual){$\textcolor{blue}{X_{\phi}[t]}$} (v8);

\draw[->,gray!75] (v5) -- (d1);
\draw[->,gray!75] (v6) -- (d2);
\draw[->,gray!75] (v7) -- (d3);
\draw[->] (v5) -- (d4);
\draw[->] (v6) -- (d4);
\draw[->] (v5) -- (d5);
\draw[->] (v7) -- (d5);
\draw[->] (v6) -- (d6);
\draw[->] (v7) -- (d6);
\draw[->,very thick,blue] (v8) -- (d4);
\draw[->,very thick,blue] (v8) -- (d5);
\draw[->,very thick,blue] (v8) -- (d6);
\draw[->,blue,dashed,very thick ](4.5,2.2) to [out=0,in=-90] (6.2,2.85)to [out=90,in=60] (4.5,3.2);
\end{tikzpicture}
\caption{The extended combination network of Example \ref{CombNet-BME-example12-3}. A block Markov encoding scheme achieves $(0,2)$ over the original combination network. At time $t+1$, information symbol $w^{\prime}_{2,3}[t+1]$ contains the information of symbol $X_{\phi}[t]$. }
\label{CombNet-example2to3bmecode-final}

\end{figure}
\end{example}

{Out third coding scheme (Theorem \ref{lb-CombNetbme-Theorem} in Section \ref{CombNet-ach-blockMarkov}) builds on Example \ref{CombNet-BME-example12-3}. When there are $m=4$ or more public receivers, our block Markov scheme is  more powerful that the first two schemes and we are not aware of any example where this scheme is sub-optimal.
In Section \ref{CombNet-ach-blockMarkov}, we describe our block Markov encoding scheme and characterize  the rate region it achieves.
We show, for three (or fewer) public and any number of private receivers, that 
this rate-region is equal to the capacity region and, therefore, coincides with 
the rate-region of Theorem \ref{lb-CombNet-Theoremk2}. Furthermore,  we show 
through an example that  the block Markov encoding scheme could outperform the 
previously discussed linear encoding schemes when there are $4$ or more 
public receivers.}

{In Section \ref{lb-sec-generalbc}, we further adapt this scheme to general 
broadcast channels with two nested message sets and obtain a rate region that includes  previously known rate-regions. We do not know if this inclusion is strict.}


\section{notation}
\label{notation}
We denote the set of outgoing edges from the source by $\mathcal{E}$ with cardiality  $|\mathcal{E}|=d$, and we refer to those edges as the resources of the combination network. 
The resources are labeled according to the
public receivers they are connected to; i.e., we denote the set of all resources that are connected to every public receiver 
in  $\mathcal{S}$, $\mathcal{S}\subseteq \mathcal{I}_1$, and not connected to any other public 
receiver by $\mathcal{E}_\mathcal{S}\subseteq \mathcal{E}$. 
Note that the edges in $\mathcal{E}_\mathcal{S}$ may or may not be connected to the private receivers.
We identify the subset of edges in $\mathcal{E}_{S}$ that are also connected to a private receiver 
$p$ by $\mathcal{E}_{\mathcal{S},p}$. Fig.~\ref{blockdiagram} shows this notation over a 
combination network with four receivers. In this example, $d=5$, $\mathcal{E}=\{(s,v_1),(s,v_2),(s,v_3),(s,v_4),(s,v_5)\}$, $\mathcal{E}_{\phi}=\{(s,v_4),(s,v_5)\}$, 
$\mathcal{E}_{\{1\}}=\{(s,v_1)\}$, $\mathcal{E}_{\{2\}}=\{\}$,
 $\mathcal{E}_{\{1,2\}}=\{(s,v_2),(s,v_3)\}$, $\mathcal{E}_{\phi,3}=\{(s,v_4),(s,v_5)\}$, $\mathcal{E}_{\phi,4}=\{\}$, 
$\mathcal{E}_{\{1\},3}=\mathcal{E}_{\{1\},4}=\{(s,v_1)\}$, $\mathcal{E}_{\{2\},3}=\mathcal{E}_{\{2\},4}=\{\}$, 
$\mathcal{E}_{\{1,2\},3}=\{(s,v_3)\}$, 
 $\mathcal{E}_{\{1,2\},4}=\{(s,v_2)\}$.

\begin{figure}
\centering
\begin{tikzpicture}[scale=2]
\tikzstyle{every node}=[draw,shape=circle,minimum size=.01cm,font=\small\itshape]; 

\path (1.5,2.5cm) node (v0) {$S$}; 
\draw (2.6,1.5) arc (-90:-30:.35cm) ;
\draw (1.1,1.75) arc (-140:-80:.5cm) ;
\draw (1.45,2.1) arc (-140:-80:.2cm) ;

\path (-.5,1.25)  node (v2) {$v_1$};
\path(1,1.25)  node (v3) {$v_2$}; 
\path(1.5,1.25)  node (v6) {$v_3$}; 
\path (3,1.25) node (v4) {$v_4$}; 
\path (3.5,1.25) node (v5) {$v_5$};
\path (0,0) node (d1) {$D_1$};
\path (1,0) node (d2) {$D_2$};
\path (2,0) node[shade] (d3) {$D_3$};
\path (3,0) node[shade] (d4) {$D_4$};

\draw[->] (v0) -- node[left,draw=none,fill=none] {$\mathcal{E}_{\{1\}}$}  (v2) ;
\draw[->] (v0) -- node[left,draw=none,fill=none,right] {$\mathcal{E}_{\{1,2\},3}$} (v6) ;
\draw[->] (v0) --node[left,draw=none,fill=none,left,yshift=-.4cm]{$\mathcal{E}_{\{1,2\}}$}(v3) ;
\draw[->] (v0) --(v4) ;
\draw[->] (v0) -- node[left,draw=none,fill=none,xshift=1.8cm,yshift=-.5cm]{$\mathcal{E}_{\{\phi\}}$} (v5);
\draw[->] (v3) -- (d1);
\draw[->] (v3) -- (d2);
\draw[->] (v3) -- (d4);
\draw[->] (v6) -- (d1);
\draw[->] (v6) -- (d2);
\draw[->] (v6) -- (d3);
\draw[->] (v2) -- (d1);
\draw[->] (v2) -- (d4);
\draw[->] (v2) -- (d3);
\draw[->] (v4) -- (d3);
\draw[->] (v5) -- (d3);
\end{tikzpicture}
\caption{A combination network with two public receivers indexed by $I_1=\{1,2\}$ and two private receivers indexed by $I_2=\{3,4\}$.}
\label{FigCombinationNetworkNotation}
\end{figure}

Throughout this paper, we denote  random variables by capital letters (e.g., $X$, $Y$), the 
sets $\{1,\ldots,m\}$ and $\{m+1,\ldots,K\}$ by $\mathcal{I}_1$ and $\mathcal{I}_2$, respectively, and subsets of $\mathcal{I}_1$ by script capital letters (e.g., $\mathcal{S}=\{1,2,3\}$ and $\mathcal{T}=\{1,3,4\}$). We denote the set of all subsets of $\mathcal{I}_1$ by $2^{\mathcal{I}_1}$ and in addition denote its subsets by Greek capital letters (e.g., $\Gamma=\{\{1\},\{1,2\}\}$ and $\Lambda=\{\{1\},\{2\},\{1,2\}\}$).

All rates are expressed in $\log_2|\mathbb{F}_q|$ in this work.
The symbol carried over a resource of the combination network, $e\in\mathcal{E}$, is denoted by $x_e$ which is a scalar 
from $\mathbb{F}_q$. Similarly, 
its corresponding random variable is denoted by $X_e$. We denote by $X_\mathcal{S}$,  $\mathcal{S}\subseteq \mathcal{I}_1$, the vector of  
symbols carried over resource edges in $\e_\mathcal{S}$, and by $X_{\mathcal{S},p}$,  $\mathcal{S}\subseteq \mathcal{I}_1$, $p\in \mathcal{I}_2$, 
the vector of symbols carried over resource edges in $\e_{\mathcal{S},p}$. To simplify notation, we sometimes abbreviate 
the union sets $\bigcup_{\mathcal{S}\in\Lambda}\e_S$, $\bigcup_{\mathcal{S}\in\Lambda}\e_{\mathcal{S},p}$ and 
$\bigcup_{\mathcal{S}\in\Lambda}X_\mathcal{S}$, by $\e_\Lambda$, $\e_{\Lambda,p}$ and $X_\Lambda$, respectively, where $\Lambda$ is a subset of $2^{\mathcal{I}_1}$. 
The vector of all received symbols at receiver $i$ is denoted by $Y_i$, $i\in\{1,\ldots,K\}$. 
When communication takes place over blocks of length $n$, all vectors above take the superscript~$n$; e.g., $X^n_e$ $X_\mathcal{S}^n$, $X_{\mathcal{S},p}^n$, $X^n_{\Lambda}$, $X^n_{\Lambda,p}$, $Y_i^n$.

We define superset saturated subsets of $2^{\mathcal{I}_1}$ as follows. A subset $\Lambda\subseteq 2^{\mathcal{I}_1}$ is superset saturated if inclusion of any set $\mathcal{S}$ in $\Lambda$ implies the inclusion of all its supersets; e.g., over subsets of $2^{\{1,2,3\}}$, $\Lambda=\{\{1\},\{1,2\},\{1,3\},\{2,3\},\{1,2,3\}\}$ is superset saturated, but not $\Lambda=\{\{1\},\{1,3\},\{1,2,3\}\}$.

\begin{definition}[Superset saturated subsets]
The subset $\Lambda\subseteq 2^{\mathcal{I}_1}$ is superset saturated if and only if it has the following property.
\begin{itemize} 
\item $\mathcal{S}$ is an element of $\Lambda$ only if every $\mathcal{T}\in 2^{\mathcal{I}_1}$,  $\mathcal{T}\supseteq \mathcal{S}$, is an element of $\Lambda$. 
\end{itemize}
\end{definition}
For notational matters, we sometimes abbreviate a  superset saturated subset $\Lambda$ by the few sets that are not implied by the other sets in $\Lambda$. For example, $\{\{1\},\{1,2\},\{1,3\},\{1,2,3\}\}$ is denoted by $\{\{1\}\star\}$, and similarly $\{\{1\},\{1,2\},\{1,3\},\{2,3\},\{1,2,3\}\}$ is denoted by $\{\{1\}\star,\{2,3\}\star\}$.

\section{Rate Splitting and Linear Encoding Schemes}
\label{CombNet-ach-rate splitting}
Throughout this section, we confine ourselves to linear encoding at the source. 
For simplicity, we describe our encoding schemes for block length $n=1$, and  highlight cases where we need to code over longer blocks.

We assume rates $R_1$ and $R_2$ to be non-negative integer values\footnote{There is no loss of generality in this assumption. One can deal with rational values of $R_1$ and $R_2$ by coding over blocks of large enough length $n$ and working with integer rates $nR_1$ and $nR_2$. Also, one can attain real valued rates through sequences of rational numbers that approach them.}.
Let $w_{1,1},\ldots,w_{1,R_1}$ and $w_{2,1},\ldots,w_{2,R_2}$ be variables in ${\mathbb{F}_q}$ for messages $W_1$ 
and $W_2$, respectively. We call them the information symbols of the common and the private message.
Also, let vector $W\in{\mathbb{F}_q}^{R_1+R_2}$ be defined as the vector with coordinates in the standard basis
$ W=[w_{1,1}\ldots w_{1,R_2} w_{2,1}\ldots w_{2,R_2}]^T$. 
The symbol carried by each resource is a linear combination of the information symbols. After properly rearranging all vectors $X_\mathcal{S}$, $\mathcal{S}\subseteq \mathcal{I}_1$, we have
\begin{align*}
 \left[\begin{array}{l}X_{\{1,\ldots,m\}}\\\vdots\\X_{\{2\}}\\X_{\{1\}}\\X_{\phi}\end{array}\right]=\mathbf{A}\cdot W,
\end{align*}
where $\mathbf{A}\in {\mathbb{F}_q}^{d\times (R_1+R_2)}$ is the encoding matrix.
At each public receiver $i$, $i\in\mathcal{I}_1$, the received signal $Y_i$ is given by $Y_i=\mathbf{A}_iW$, 
where $\mathbf{A}_i$ is a submatrix of  $\mathbf{A}$ corresponding to 
$X_\mathcal{S}$, $\mathcal{S}\ni i$. Similarly, the received signal at each private receiver $p$, $p\in\mathcal{I}_2$, 
is given by $Y_p=\mathbf{A}_pW$, where  $\mathbf{A}_p$ is the submatrix of the rows of $\mathbf{A}$ corresponding to 
$X_{\mathcal{S},p}$, $\mathcal{S}\subseteq \mathcal{I}_1$.

We design $\mathbf{A}$ to allow the public receivers decode  $W_1$ and the private receivers  decode  $W_1,W_2$. We then characterize the rate pairs achievable by our code design. The challenge in the optimal code design stems from the fact that destinations receive different subsets of the symbols that are sent by the encoder and they have two different decodability requirements. On the one hand, private receivers require their received signal to bring information about all information symbols of the common and the private message. On the other hand, public receivers might not be able to decode the common message if their received symbols depend on ``too many" private message variables. We make this statement precise in Lemma \ref{lb-CombNet2-lemmaW1}. 
In the following, we find conditions for decodability of  the messages.
\subsection{Decodability Lemmas}
\begin{lemma}
\label{lb-CombNet2-lemmaW1}
Let vector $Y$ be given by \eqref{lb-CombNet-lemmaW1-yi}, below, where $\mathbf{B}\in{\mathbb{F}_q}^{r\times R_1}$, $\mathbf{T}\in{\mathbb{F}_q}^{r\times R_2}$, $W_1\in{\mathbb{F}_q}^{R_1\times1}$, and $W_2\in{\mathbb{F}_q}^{R_2\times1}$.
\begin{align}
\label{lb-CombNet-lemmaW1-yi}
Y=\left[\begin{array}{c|c}\mathbf{B}&\mathbf{T}\end{array}\right]\cdot \left[\begin{array}{c}W_1\\W_2\end{array}\right].
\end{align} 
Message $W_1$ is recoverable from $Y$ \textit{if and only if} 
 $\rank\left(\mathbf{B}\right)=R_1$
and the column space of $\ \mathbf{B}$ is disjoint from that of $\mathbf{T}$.
\end{lemma}
\begin{corollary}
\label{lb-CombNet-cor-W1}
Message $W_1$ is recoverable from $Y$ in equation \eqref{lb-CombNet-lemmaW1-yi} only if 
$$\rank(\mathbf{T})\leq r-R_1.$$
\end{corollary}

We defer the proof of Lemma \ref{lb-CombNet2-lemmaW1} to Appendix \ref{ap-CombNet2-lemmaW1} and instead discuss the high-level implication of the result.
Let $r_\mathbf{T}=\rank(\mathbf{T})$ where $r_\mathbf{T}\leq r-R_1$. Matrix $\mathbf{T}$ can be written as $\mathbf{L}_1\mathbf{L}_2$, 
where $\mathbf{L}_1$ is a full-rank matrix of dimension $r\times r_\mathbf{T}$ and $\mathbf{L}_2$ is a full-rank matrix of dimension $r_\mathbf{T}\times R_2$. $\mathbf{L}_1$ is essentially  just a set of linearly independent columns of $\mathbf{T}$ spanning its column space.  In other words, we can write
\begin{eqnarray}
\left[\begin{array}{c|c}\mathbf{B}&\mathbf{T}\end{array}\right]W&=&\left[\begin{array}{c|c}\mathbf{B}&\mathbf{L}_1\mathbf{L}_2\end{array}\right]W\\
&=&\left[\begin{array}{c|c}\mathbf{B}&\mathbf{L}_1\end{array}\right]\left[\begin{array}{c}W_{1}\\\mathbf{L}_2W_2\end{array}\right].
\end{eqnarray}
Since $r_\mathbf{T}+R_1\leq r$, $W_1,\mathbf{L}_2W_2$ are decodable if $\left[\begin{array}{c|c}\mathbf{B}&\mathbf{L}_1\end{array}\right]$ is full-rank.

Defining $\left[\begin{array}{c|c}\mathbf{B}&\mathbf{T}\end{array}\right]$ as a new $\mathbf{B}^\prime$ of dimension $r\times (R_1+R_2)$ and defining a null matrix $\mathbf{T}^\prime$ in Lemma \ref{lb-CombNet2-lemmaW1}, we reach at the  trivial result of the following corollary. 
\begin{corollary}
\label{lb-CombNet-cor-W1W2}
Messages $W_1,W_2$ are recoverable from $Y$ in equation \eqref{lb-CombNet-lemmaW1-yi} \textit{if and only if}
$$\rank\left(\left[\begin{array}{c|c}\mathbf{B}&\mathbf{T}\end{array}\right]\right)= R_1+R_2.$$
\end{corollary}

Since every receiver sees a different subset of the sent symbols,  it becomes clear from Corollary \ref{lb-CombNet-cor-W1} and \ref{lb-CombNet-cor-W1W2} that  an admissible linear code needs to satisfy many rank constraints on its different submatrices.
In this section, our primary approach to the design of such codes is through zero-structured matrices, discussed next.
\subsection{Zero-structured matrices}
\label{CombNet-zero-structure}
\begin{definition}
A zero-structured matrix $\mathbf{T}$ is an $r\times c$ matrix with entries either zero or indeterminate\footnote{Although zero-structured matrices are defined in Definition \ref{lb-Comb-def-zerostructure} with zero or indeterminate variables, we also refer to the assignments of such matrices as zero-structured matrices.} (from a finite field ${\mathbb{F}_q}$) in a specific structure, as follows.
This matrix consists of $2^t \times 2^t$ blocks, where each block is indexed on rows and columns by the subsets of $\{1,\cdots,t\}$. Block $b_{(\mathcal{S}_1,\mathcal{S}_2)}$, $\mathcal{S}_1,\mathcal{S}_2 \subseteq \{1,\cdots,t\}$, is an ${r_{\mathcal{S}_1}\times c_{\mathcal{S}_2}}$ matrix. Matrix $\mathbf{T}$ is structured so that all entries in block $b_{(\mathcal{S}_1,\mathcal{S}_2)}$ are set to zero if $\mathcal{S}_1\not\subseteq \mathcal{S}_2$, and remain indeterminate otherwise. Note that $c=\sum_{\mathcal{S}} c_\mathcal{S}$ and $r=\sum_\mathcal{S} r_\mathcal{S}$.
\label{lb-Comb-def-zerostructure}
\end{definition}
Equation \eqref{lb-CombNet-structureB}, below, demonstrates this definition for $t=2$.
\begin{align}
\label{lb-CombNet-structureB}
& \begin{array}{c}\\\mathbf{T}=\end{array}\begin{array}{l}   
\begin{array}{cccc}\hspace{-.1cm}\stackrel{c_{\{1,2\}}}{\longleftrightarrow}&\hspace{-.4cm}\stackrel{\, c_{\{1\}}}{\leftrightarrow}&\hspace{-.4cm}\stackrel{\,c_{\{2\}}}{\leftrightarrow}&\hspace{-0.25cm}\stackrel{c_{\phi}}{\leftrightarrow}\end{array}\\
\left[\begin{array}{c|c|c|c}
\ \hspace{1.5cm}&0&0&0\\\hline
\hspace{1.5cm}&\hspace{.2cm}&0&0\\	\hline	   
\hspace{1.5cm}&0&\hspace{.2cm}&0\\\hline
\hspace{1.5cm}&\hspace{.2cm}&\hspace{.2cm}&\hspace{.2cm}
                  \end{array}\right]
	    \begin{array}{c}\updownarrow\\\updownarrow\\\updownarrow\\\updownarrow\end{array}
	   \hspace{-.3cm} {\tiny{\begin{array}{l}r_{\{1,2\}}\\\\ r_{\{1\}}\\\\r_{\{2\}}\\\\r_{\phi}\end{array}}}
\end{array}
\end{align}
The idea behind using zero-structred encoding matrices is the following: the zeros are inserted in the encoding matrix such that the linear combinations that are formed for the public receivers do not depend on ``too many" private information symbols (see Corollary \ref{lb-CombNet-cor-W1}).

In the rest of this subsection, we find conditions on zero-structured matrices so that they can be made full column rank. 

\begin{lemma}
\label{lb-CombNet-fullrank}
There exists an assignment of the indeterminates in the zero-structured matrix $\mathbf{T}\in{\mathbb{F}_q}^{r\times c}$  (as specified in Definition~\ref{lb-Comb-def-zerostructure}) that makes it full column rank, provided that
\begin{align}
\label{lb-CombNet-fullrankzero}
c\leq \sum_{\mathcal{S}\in \Lambda}c_\mathcal{S}+\sum_{\mathcal{S}\in \Lambda^c}r_\mathcal{S},\quad\quad{\forall \Lambda \subseteq 2^{\{1,\ldots,t\}} \text{ superset saturated}}.
\end{align}
\end{lemma}
For $t=2$, \eqref{lb-CombNet-fullrankzero} is given by
\begin{align}
&c\leq r_{\{1,2\}}+r_{\{1\}}+r_{\{2\}}+r_\phi\\
&c\leq c_{\{1,2\}}+r_{\{1\}}+r_{\{2\}}+r_\phi\\
&c\leq c_{\{1\}}+ c_{\{1,2\}}+r_{\{2\}}+r_\phi\\
&c\leq c_{\{2\}}+c_{\{1,2\}}+r_{\{1\}}+r_\phi\\
&c\leq c_{\{1\}}+c_{\{2\}}+c_{\{1,2\}}+r_\phi\\
&c\leq c_\phi+c_{\{1\}}+c_{\{2\}}+c_{\{1,2\}}.
\end{align}


We briefly outline the proof of Lemma \ref{lb-CombNet-fullrank}, because this line of argument is used later in Section \ref{lb-CombNet-MES}. For simplicity of notation and clarity of the proof, we give details of the proof for $t = 2$. The same proof technique proves the general case. 

Let  $\mathbf{T}\in {\mathbb{F}_q}^{r\times c}$ be a zero-structured matrix given by equation \eqref{lb-CombNet-structureB}.
First, we reduce the problem of matrix $\mathbf{T}$ being full column rank to an information flow problem over the equivalent unicast network of Fig.~\ref{FigCombinationNetworkBasicLemma2}. Then we find conditions for feasibility of the equivalent unicast problem. 
The former is stated in Lemma \ref{lb-CombNet-declem-2} and the latter is formulated in Lemma \ref{lb-CombNet-mincut}, both to follow.
\begin{figure}
\centering
\begin{tikzpicture}[scale=2.2]
\tikzstyle{every node}=[draw,shape=circle,minimum size=1cm,font=\small\itshape]; 
\path (0:0cm) node (v0) {$A$}; 
\path (20:.5cm) node [rectangle,draw=none] {source}; 

\path(180+30:1.8cm)  node (v12) {$n_{\{1,2\}}$}; 
\path(180+70:1.0cm)  node (v2) {$n_{\{2\}}$}; 
\path (180+110:1.0cm) node (v1) {$n_{\{1\}}$}; 
\path (180+150:1.8cm) node (vemp) {$n_\phi$};

\path (180+45:2.5cm)  node (u12) {$n^\prime_{\{1,2\}}$};
\path(180+70:1.87cm)  node (u2) {$n^\prime_{\{2\}}$}; 
\path(180+110:1.87cm)  node (u1) {$n^\prime_{\{1\}}$}; 
\path (180+135:2.5cm) node (uemp) {$n^\prime_{\phi}$};

\path (180+90:2.6cm) node (d) {$B$};
\path (-80:2.8cm) node[rectangle,draw=none] {sink};

\draw[->] (v0) -- node[left,draw=none,fill=none] {$c_{\{1,2\}}$}    (v12) ;

\draw[->] (v0) -- node[left,draw=none,fill=none] {$c_{\{2\}}$}    (v2) ;

\draw[->] (v0) -- node[left,draw=none,fill=none,xshift=.2cm] {$c_{\{1\}}$}    (v1) ;

\draw[->] (v0) -- node[left,draw=none,fill=none] {$c_\phi$}    (vemp) ;


\draw[->,line width=.07cm] (v12) -- (u12);
\draw[->,line width=.07cm] (v12) -- (u2);
\draw[->,line width=.07cm] (v12) -- (u1);
\draw[->,line width=.07cm] (v12) -- (uemp);
\draw[->,line width=.07cm] (v2) -- (u2);
\draw[->,line width=.07cm] (v2) -- (uemp);
\draw[->,line width=.07cm] (v1) -- (u1);
\draw[->,line width=.07cm] (v1) -- (uemp);
\draw[->,line width=.07cm] (vemp) -- (uemp);

\draw[->] (u12) -- node[left,draw=none,fill=none] {$r_{\{1,2\}}$}    (d) ;

\draw[->] (u2) -- node[left,draw=none,fill=none] {$r_{\{2\}}$}    (d) ;

\draw[->] (u1) -- node[left,draw=none,fill=none] {$r_{\{1\}}$}    (d) ;

\draw[->] (uemp) -- node[left,draw=none,fill=none] {$r_\phi$}    (d) ;

\end{tikzpicture}
\caption{The source $A$ communicates a message of rate $c=\sum_{\mathcal{S}}c_\mathcal{S}$ to the sink  $B$ over a unicast network. The capacity of each edge is marked beside it. The bold edges are of infinite capacity.
The network is tailored so that the mixing of the information which happens at the third layer mimics the same mixing of the information that is present in the rows of matrix $\mathbf{T}$ in \eqref{lb-CombNet-structureB}. }
\label{FigCombinationNetworkBasicLemma2}
\end{figure}

The equivalent unicast network of Fig.~\ref{FigCombinationNetworkBasicLemma2} is formed as follows.
The network is a four-layer directed network with a source node $A$ in the first layer, four (in general $2^t$) nodes $n_\mathcal{S}$, $\mathcal{S}\subseteq \{1,\ldots,t\}$, in the second layer, another four (in general $2^t$) nodes $n^\prime_\mathcal{S}$, $\mathcal{S}\subseteq \{1,\ldots,t\}$, in the third layer, and finally a sink node $B$ in the fourth layer. 
The source wants to communicate a message of rate $c$ to the sink.
We have $c_\mathcal{S}$ (unit capacity) edges from the source $A$ to each node $n_\mathcal{S}$. Also, we have $r_\mathcal{S}$ (unit capacity) edges from  each node $n^\prime_\mathcal{S}$ to the sink  $B$. The edges from the second layer to the third layer are  of infinite capacity and they connect each node $n_\mathcal{S}$ to all nodes $n^\prime_{\mathcal{S}^\prime}$ where $\mathcal{S}^\prime\subseteq \mathcal{S}$. The equivalent unicast network is tailored so that the mixing of the information which happens at each node $n^\prime_\mathcal{S}$ (at the third layer) mimics the same mixing of the information that is present in the rows of matrix $\mathbf{T}$.
This equivalence is discussed formally in Lemma \ref{lb-CombNet-declem-2} and its proof is deferred to Appendix \ref{ap-CombNet-declem-2}.
\begin{lemma}
\label{lb-CombNet-declem-2}
Given the zero-structured matrix $\mathbf{T}\in{\mathbb{F}_q}^{r\times c}$ in \eqref{lb-CombNet-structureB}, the following two statements are equivalent. 
\begin{enumerate}
\item[$(i)$] There exists an assignment of variables in $\mathbf{T}$ that makes it full column rank.
\item[$(ii)$] A message of rate $c$ could be sent over its equivalent unicast network in Fig.~\ref{FigCombinationNetworkBasicLemma2}.
\end{enumerate}
\end{lemma}

Using Lemma \ref{lb-CombNet-declem-2},  $\mathbf{T}$ could be made full-rank if $c$ is less than or equal to the min-cut between nodes $A$ and $B$ over the equivalent unicast network. The min-cut between $A$ and $B$ is given by Lemma \ref{lb-CombNet-mincut} and the proof is delegated to Appendix \ref{ap-CombNet-mincut}. 
\begin{lemma}
\label{lb-CombNet-mincut}
The min-cut separating nodes $A$ and $B$ over the network of Fig.~\ref{FigCombinationNetworkBasicLemma2} is given by 
\begin{align}
\min_{\substack{\Lambda \subseteq 2^{\mathcal{I}_1}\\ \Lambda\text{ superset saturated}}}\sum_{\mathcal{S}\in \Lambda}c_\mathcal{S}+\sum_{\mathcal{S}\in \Lambda^c}r_\mathcal{S}.
\end{align}
\end{lemma}
\noindent Lemma \ref{lb-CombNet-declem-2} and Lemma \ref{lb-CombNet-mincut} give conditions for matrix $\mathbf{T}$ to become full rank. The proof for $t>2$ is along the same lines.

\subsection{Zero-structured linear codes: an achievable rate-region}
\label{lb-CombNet-ZES}
In our initial approach, we design the encoding matrix $\mathbf{A}$ to be  zero-structured. The idea is to have the resource symbols that are available to the public receivers  not depend on ``too many" private message variables.  
Equation \eqref{lb:CombNet:structure} below shows such an encoding matrix for two public and any number of private receivers.
\begin{align}
\label{lb:CombNet:structure}
 \begin{array}{c}\\\mathbf{A}=\end{array}\begin{array}{l}   
\begin{array}{ccccc}\hspace{.1cm}\stackrel{\ \hspace{.25cm}R_1\hspace{.25cm}\ }{\longleftrightarrow}&\hspace{-.3cm}\stackrel{\alpha_{\{1,2\}}}{\longleftrightarrow}&\hspace{-.3cm}\stackrel{\alpha_{\{2\}}}{\leftrightarrow}&\hspace{-.35cm}\stackrel{\alpha_{\{1\}}}{\leftrightarrow}&\hspace{-0.25cm}\stackrel{\alpha_{\phi}}{\leftrightarrow}\end{array}\\
\left[\begin{array}{c|c|c|c|c}
                   \hspace{.8cm}\ &\ \hspace{.4cm}\ &0&0&0\\\hline
		    \hspace{.8cm}\ &\ \hspace{.4cm}\ &\hspace{.2cm}&0&0\\	\hline	   
		    \hspace{.8cm}\ &\ \hspace{.4cm}\ &0&\hspace{.2cm}&0\\\hline
		    \hspace{.8cm}\ &\ \hspace{.4cm}\ &\hspace{.2cm}&\hspace{.2cm}&\hspace{.2cm}
                  \end{array}\right]
	    \begin{array}{c}\updownarrow\\\updownarrow\\\updownarrow\\\updownarrow\end{array}
	   \hspace{-.3cm} {\tiny{\begin{array}{l}|\e_{\{1,2\}}|\\\\ |\e_{\{2\}}|\\\\|\e_{\{1\}}|\\\\|\e_{\phi}|\end{array}}}.
\end{array}
\end{align}
Here, the non-zero entries are all indeterminate and to be designed appropriately. Also, parameters $\alpha_{\{1,2\}}$, $\alpha_{\{2\}}$, $\alpha_{\{1\}}$ and $\alpha_\phi$ are non-negative structural parameters, and they satisfy
\begin{align}
R_2=\alpha_{\{1,2\}}+\alpha_{\{2\}}+\alpha_{\{1\}}+\alpha_{\phi}.
\end{align}

In effect, matrix $\mathbf{A}$ splits message $W_2$ into four independent messages, $W_2^{\{1,2\}}$, $W_2^{\{2\}}$, $W_2^{\{1\}}$, $W_2^{\phi}$, of rates $\alpha_{\{1,2\}}$, $\alpha_{\{2\}}$, $\alpha_{\{1\}}$, $\alpha_{\phi}$, respectively. The zero structure of $\mathbf{A}$ ensures that only messages $W_1$, $W_2^{\{1,2\}}$ and $W_2^{\{1\}}$ are involved in the linear combinations that are received at public receiver $1$ and that only messages $W_1$, $W_2^{\{1,2\}}$ and $W_2^{\{2\}}$ are involved in the linear combinations that are received at public receiver $2$.
In general, matrix $\mathbf{A}$ splits message $W_2$ into independent messages $W_2^\mathcal{S}$ of rates $\alpha_\mathcal{S}$, $\mathcal{S}\subseteq \mathcal{I}_1$, such that
\begin{align}
\sum_{\mathcal{S}\subseteq \mathcal{I}_1} \alpha_{\mathcal{S}}=R_2.
\end{align}
Note that the zero-structure allows messages $W_2^\mathcal{S}$ to be involved (only) in the linear combinations that are sent over resources in $\mathcal{E}_{\mathcal{S}^\prime}$ where $\mathcal{S}^\prime\subseteq \mathcal{S}$. When referring to a zero-structured encoding matrix $\mathbf{A}$, we also specify the rate-split parameters $\alpha_\mathcal{S}$, $\mathcal{S}\subseteq \mathcal{I}_1$.
\begin{remark}
\label{lb-CombNet-alphapos}
As defined above, parameters $\alpha_\mathcal{S}$, $\mathcal{S}\subseteq \mathcal{I}_1$, are assumed to be integer-valued. Nonetheless, one can let these parameters be real and approximately attain them by encoding over blocks of large enough length.
\end{remark}

Conditions under which all receivers can decode their messages of interest are as follow:
\begin{itemize}
\item \textbf{Public receiver $i\in \mathcal{I}_1$:}  $Y_i$ is the vector of all the symbols  that are available to receiver $i$. Using the 
zero-structure of $\mathbf{A}$
in \eqref{lb:CombNet:structure}, we have
\begin{align}
Y_2&=\left[\begin{array}{l}X_{\{1,2\}}\\X_{\{2\}}\end{array}\right]\\&=\begin{array}{l}\begin{array}{ccccc}\hspace{0cm}\stackrel{\ \hspace{.25cm}R_1\hspace{.25cm}\ }{\longleftrightarrow}&\hspace{-.2cm}\stackrel{\alpha_{\{1,2\}}}{\longleftrightarrow}&\hspace{-.3cm}\stackrel{\alpha_{\{2\}}}{\leftrightarrow}&\hspace{-.35cm}\stackrel{\alpha_{\{1\}}}{\leftrightarrow}&\hspace{-0.25cm}\stackrel{\alpha_{\phi}}{\leftrightarrow}\end{array}\\
\left[\begin{array}{c|c|c|c|c}
                  \hspace{.8cm}\ &\ \hspace{.4cm}\ &0&0&0\\\hline
		   \hspace{.8cm}\ &\ \hspace{.4cm}\ &\hspace{.2cm}&0&0
                  \end{array}\right]\begin{array}{c}\updownarrow\\\updownarrow\end{array}
	   \hspace{-.3cm} {\tiny{\begin{array}{l}|\e_{\{1,2\}}|\vspace{.15cm}\\ |\e_{\{2\}}|\end{array}}}
\end{array} \begin{array}{c}\\\hspace{-.5cm}\cdot\hspace{-.1cm}\left[\begin{array}{c}W_1\\W_2\end{array}\right].\end{array}
\label{lb-CombNet-structureAi}
\end{align}
Generally, $Y_i$ is given by 
\begin{align}
Y_i=\mathbf{A}_iW,
\end{align} where $\mathbf{A}_i$ is a submatrix of  $\mathbf{A}$ corresponding to $X_\mathcal{S}$, $\mathcal{S}\ni i$. It has at most $R_1+\sum_{\substack{\mathcal{S}\subseteq \mathcal{I}_1\ \mathcal{S}\ni i}}\alpha_{\mathcal{S}}$ non-zero columns. 
We relate decodability of message $W_1$ at receiver $i$ to a particular submatrix of $\mathbf{A}_i$ being full column rank. 
Let 
\begin{align}
\mathbf{A}_i=\left[\mathbf{B}^{(i)}|\mathbf{T}^{(i)}\right].
\end{align} 
Choose $\mathbf{L}^{(i)}_{1}$ to be a largest submatrix of the columns of $\mathbf{T}^{(i)}$ that could be made full column rank (over all possible assignments).
Define
\begin{align}
\mathbf{G}^{(i)}=\left[\begin{array}{c|c}\mathbf{B}^{(i)}&\mathbf{L}^{(i)}_{1}\end{array}\right].
\end{align} 
We have the following two lemmas.
\begin{lemma}
\label{lem-Lbuilt}
Each public receiver $i$ can decode $W_1$ from \eqref{lb-CombNet-structureAi} 
if $\mathbf{G}^{(i)}$, as defined above, is full column rank.
\end{lemma}
\begin{IEEEproof}
Conditions for decodability of $W_1$ are given in Lemma \ref{lb-CombNet2-lemmaW1} (and its following argument). By the manner $\mathbf{L}^{(i)}_1$ and $\mathbf{G}^{(i)}$ are defined, whenever $\mathbf{G}^{(i)}$ is full column rank, not only $\mathbf{B}^{(i)}$ is full column rank, but also  columns of $\mathbf{L}^{(i)}_1$ span all the column space of $\mathbf{T}^{(i)}$ (for all possible assignments-- otherwise a larger $\mathbf{L}^{(i)}$ would have been chosen) and the span of the column space of $\mathbf{T}^{(i)}$ is thus disjoint from that of $\mathbf{B}^{(i)}$.
\end{IEEEproof}
\begin{lemma}
\label{lb-lemGifull}
For each public receiver $i$, there exists an assignment of $\mathbf{A}$ (specific to $i$) such that $\mathbf{G}^{(i)}$ is full column rank, provided that
\begin{eqnarray}
\label{lb-lemGifull-cond}
R_1+\sum_{\substack{\mathcal{S}\subseteq \mathcal{I}_1\\ \mathcal{S}\ni i}}\alpha_{\mathcal{S}}\leq \sum_{\substack{\mathcal{S}\subseteq \mathcal{I}_1\\ \mathcal{S}\ni i}}|\mathcal{E}_{\mathcal{S}}|.
\end{eqnarray}
\end{lemma}
\begin{IEEEproof}
The proof is deferred to Appendix \ref{ap-lemGifull}.
\end{IEEEproof}


 \item \textbf{Private receiver $p\in \mathcal{I}_2$:} $Y_p$ is the  vector of all the symbols that are carried by the resources in  
$\mathcal{E}_{\mathcal{S},p}$, $\mathcal{S}\subseteq \mathcal{I}_1$. 
We have
\begin{align}
\label{lb-CombNet-structureAp}
Y_p&=\left[\begin{array}{l}X^p_{\{1,2\}}\\X^p_{\{2\}}\\X^p_{\{1\}}\\X^p_{\phi}\end{array}\right]\\&=\begin{array}{l}\begin{array}{ccccc}\hspace{0cm}\stackrel{\ \hspace{.25cm}R_1\hspace{.25cm}\ }{\longleftrightarrow}&\hspace{-.2cm}\stackrel{\alpha_{\{1,2\}}}{\longleftrightarrow}&\hspace{-.3cm}\stackrel{\alpha_{\{2\}}}{\leftrightarrow}&\hspace{-.35cm}\stackrel{\alpha_{\{1\}}}{\leftrightarrow}&\hspace{-0.25cm}\stackrel{\alpha_{\phi}}{\leftrightarrow}\end{array}\vspace{-.2cm}\\
\left[\begin{array}{c|c|c|c|c}
                   \hspace{.8cm}\ &\ \hspace{.4cm}\ &0&0&0\\\hline
		    \hspace{.8cm}\ &\ \hspace{.4cm}\ &\hspace{.2cm}&0&0\\\hline
		    \hspace{.8cm}\ &\ \hspace{.4cm}\ &0&\hspace{.2cm}&0\\\hline
		    \hspace{.8cm}\ &\ \hspace{.4cm}\ &\hspace{.2cm}&\hspace{.2cm}&\hspace{.2cm}
                  \end{array}\right] \begin{array}{c}\updownarrow\\\updownarrow\\\updownarrow\\\updownarrow\end{array}
	   \hspace{-.3cm} {\tiny{\begin{array}{l}\\|\e_{\{1,2\},p}|\vspace{.2cm}\\ \vspace{.2cm}|\e_{\{2\},p}|\\\vspace{.2cm}|\e_{\{1\},p}|\\\vspace{.2cm}|\e_{\phi,p}|\end{array}}}
\end{array}\begin{array}{c}\\\hspace{-.6cm}\cdot\hspace{0cm}\left[\begin{array}{c}W_1\\W_2\end{array}\right].\end{array}
\end{align}
Generally,  $Y_p$ is given by $Y_p=\mathbf{A}_pW$, where  $\mathbf{A}_p$ is the submatrix of the rows of $\mathbf{A}$ that corresponds to $X_{\mathcal{S},p}$, $\mathcal{S}\subseteq \mathcal{I}_1$. Note that $\mathbf{A}_p$ is a zero-structured matrix.
Messages $W_1,W_2$  are decodable at private receiver $p$ if and only if matrix $\mathbf{A}_p$ is full column rank.  From Lemma \ref{lb-CombNet-fullrank}, an assignment of $\mathbf{A}_p$ exists that makes it full column rank  provided that the following inequalities hold:
\begin{align}
&R_2\leq \sum_{\substack{\mathcal{S}\in \Lambda}}\alpha_\mathcal{S}+\sum_{\substack{\mathcal{S}\in \Lambda^c}}\card{\e_{\mathcal{S},p}},\quad\quad \forall \Lambda \subseteq 2^{\mathcal{I}_1}\text{ superset saturated},\label{lb-CombNetp1-lemmaw1w2}\\
&R_1+R_2\leq \sum_{\mathcal{S}\subseteq \mathcal{I}_1}\card{\e_{\mathcal{S},p}}.\label{lb-CombNetp1-lemmaw1w2'}
\end{align}
\end{itemize}

Inequalities \eqref{lb-lemGifull-cond}, \eqref{lb-CombNetp1-lemmaw1w2} and \eqref{lb-CombNetp1-lemmaw1w2'} provide constraints on parameters $\alpha_\mathcal{S}$, $\mathcal{S}\subseteq \mathcal{I}_1$, under which $W_1$ is decodable at the public receivers (i.e.,  matrices $\mathbf{G}^{(i)}$, $i\in\mathcal{I}_1$, could be made full rank), and $W_1,W_2$ are decodable at private receivers (i.e., matrices $\mathbf{A}_p$, $p\in\mathcal{I}_2$, could be made full rank). It remains to argue that there exists a universal assignment of $\mathbf{A}$ such that all receivers can decode their messages of interest. We do this by directly applying the sparse zeros lemma \cite[Lemma 2.3]{kcl07}.
\begin{lemma}
If  $|{\mathbb{F}_q}|>K$, a universal assignment of $\mathbf{A}$ exists such that all $\mathbf{G}^{(i)}$, $i\in\mathcal{I}_1$, and all $\mathbf{A}_p$, $p\in\mathcal{I}_2$, become simultaneously full column rank. 
\end{lemma}

\begin{remark}
{Note that operation over smaller fields is also possible by coding over blocks of larger lengths. Coding over blocks of length $n$ is  effectively done over the field $\mathbb{F}_{q^n}$. Therefore, we require $q^n>K$; i.e., we need $n>\log_{q}K$.}
\end{remark}
The rate-region achievable by this scheme can be posed as a feasibility problem in terms of parameters $\alpha_\mathcal{S}$, $\mathcal{S}\subseteq \mathcal{I}_1$.
We summarize this region in the following proposition.
\begin{proposition}
\label{CombNetp1-innerbound}
The rate pair $(R_1,R_2)$ is achievable if there exists a set of real valued variables $\alpha_{\mathcal{S}}$, $\mathcal{S}\subseteq \mathcal{I}_1$, that satisfies the following inequalities:
\begin{align}
&\text{Structural constraints:}\nonumber\\
&\hspace{1.5cm} \alpha_\mathcal{S}\geq 0,\quad\quad \forall \mathcal{S}\subseteq \mathcal{I}_1 \label{CombNetp1-achpos}\\
 &\hspace{1.5cm} R_2=\sum_{\mathcal{S}\subseteq \mathcal{I}_1}\alpha_\mathcal{S}\label{CombNetp1-achR2or}\\
&\text{Decoding constraints at public receivers:}\nonumber\\
&\hspace{1.5cm} R_1+\sum_{\substack{\mathcal{S}\subseteq \mathcal{I}_1\\\mathcal{S}\ni i}}\alpha_{\mathcal{S}}\leq \sum_{\substack{\mathcal{S}\subseteq \mathcal{I}_1\\\mathcal{S}\ni i}}|\e_\mathcal{S}|,\quad\quad \forall i\in\mathcal{I}_1 \label{CombNetp1-achr1or}\\
&\text{Decoding constraints at private receivers:}\nonumber\\
&\hspace{1.5cm} R_2\leq \sum_{\substack{\mathcal{S}\in \Lambda}}\alpha_\mathcal{S}+\sum_{\substack{\mathcal{S}\in \Lambda^c}}\card{\e_{\mathcal{S},p}},\quad\quad \forall \Lambda \subseteq 2^{\mathcal{I}_1}\text{ superset saturated},\ \forall p\in\mathcal{I}_2\label{CombNetp1-ach1or}\\
&\hspace{1.5cm} R_1+R_2\leq \sum_{\mathcal{S}\subseteq \mathcal{I}_1}|\e_{\mathcal{S},p}|,\quad\quad \forall p\in\mathcal{I}_2.\label{CombNetp1-achR1+R2or}
\end{align}
\end{proposition}

\begin{remark}
\label{lb-CombNet-remarkdecodability}
Note that  inequality \eqref{CombNetp1-achr1or} ensures  decodability of only the common message (and not the superposed messages $W_2^\mathcal{S}$, $i\in \mathcal{S}\subseteq \mathcal{I}_1$) at the public receivers. To have public receiver $i$ decode all the superposed messages $W^{\mathcal{S}}_2$, $i\in \mathcal{S}\subseteq \mathcal{I}_1$, as well, one needs further constraints on $\alpha_\mathcal{S}$, as given below:
\begin{align}
\label{lb-CombNetp1-remark}
\sum_{\substack{ \mathcal{S}\subseteq \mathcal{I}_1\\\mathcal{S}\ni i}}\alpha_\mathcal{S}\leq \sum_{\mathcal{S}\in\Lambda}\alpha_\mathcal{S}+\sum_{\substack{\mathcal{S}\in \Lambda^c\\\mathcal{S}\ni i}}\card{\e_\mathcal{S}},\quad\quad \Lambda\subseteq {\{\{i\}\star\}}\text{ superset saturated}.
\end{align}
{More precisely, define $\tilde{\mathbf{A}}_i$ to be the submatrix of $\mathbf{A}_i$ that does not contain the all-zero columns. One observes that messages $W^{\mathcal{S}}_2$, $i\in \mathcal{S}\subseteq \mathcal{I}_1$, are all decodable if and only if $\tilde{\mathbf{A}}_i$ is full column rank. Since $\tilde{\mathbf{A}}_i$ is zero-structured, Lemma \ref{lb-CombNet-fullrank} gives the required constraints.}
\end{remark}
\begin{remark}
\label{remarkbme}
 In a similar manner, a \textit{general multicast code} could be designed for the scenario where $2^K$ message sets are communicated and each message set is destined for a subset of the $K$ receivers.
 \end{remark}

\begin{remark}
What the zero-structure encoding matrix does, in effect, is implement the standard techniques of rate splitting and linear superposition coding. We prove in Subsection \ref{CombNet-subsect-outer1} that this encoding scheme is rate-optimal for combination networks with two public and any number of private receivers. However, this encoding scheme is not in general optimal. We discuss this sub-optimality next and modify the encoding scheme to attain a strictly larger rate region.
\end{remark}

\subsection{An achievable rate-region from pre-encoding and structured linear codes}
\label{lb-CombNet-MES}
For combination networks with three or more public receivers (and any number of private receivers), the scheme outlined above turns out to be sub-optimal in general. We discussed one such example in Example \ref{lb-CombNet-example2to3} where $(R_1=0,R_2=2)$ was not achievable by  Proposition \ref{CombNetp1-innerbound} (see \eqref{exampleregion1}-\eqref{exampleregionlast}). 
If we were to relax the non-negativity condition on $\alpha_\phi$ in \eqref{exampleregion1}-\eqref{exampleregionlast}, we would get feasibility of $(R_1,R_2)=(0,2)$ for the following set of parameters $\alpha_\mathcal{S}$:  $\alpha_\phi=-1$, $\alpha_{\{1\}}=\alpha_{\{2\}}=\alpha_{\{3\}}=1$,  and $\alpha_{\{1,2\}}=\alpha_{\{1,3\}}=\alpha_{\{2,3\}}=\alpha_{\{1,2,3\}}=0$. 
Obviously, there is no longer a ``structural" meaning to this set of parameters. Nonetheless, it still has a peculiar meaning that we try to investigate in this example. 
As suggested by the positive parameters $\alpha_{\{1\}}$, $\alpha_{\{2\}}$, $\alpha_{\{3\}}$, we would like, in an optimal code design, to reveal a subspace of dimension one of the private message space to each public receiver (and only that public receiver). The subtlety lies in the fact that those pieces of private information are \textit{not} mutually independent, as message  $W_2$ is of rate $2$. The previous scheme does not allow such dependency.

Inspired by Example \ref{lb-CombNet-example2to3}, we modify the encoding scheme, using an appropriate pre-encoder, to obtain a strictly larger achievable region as expressed in Theorem \ref{lb-CombNet-Theoremk2}.
\begin{theorem}
\label{lb-CombNet-Theoremk2}
The rate pair $(R_1,R_2)$ is achievable if there exists a set of real valued variables $\alpha_{S}$, $S\subseteq \mathcal{I}_1$, that satisfy the following inequalities:
\begin{align}
&\text{Structural constraints:}\nonumber\\
&\hspace{1.5cm} \alpha_\mathcal{S}\geq 0,\quad\quad \forall  \mathcal{S}\subseteq \mathcal{I}_1,\ \mathcal{S}\neq \phi \label{CombNetk2-achpos}\\
 &\hspace{1.5cm} R_2=\sum_{\mathcal{S}\subseteq \mathcal{I}_1}\alpha_\mathcal{S},\label{CombNetk2-achR2or}\\
&\text{Decoding constraints at public receivers:}\nonumber\\
&\hspace{1.5cm} R_1+\sum_{\substack{\mathcal{S}\subseteq \mathcal{I}_1\\\mathcal{S}\ni i}}\alpha_{\mathcal{S}}\leq \sum_{\substack{\mathcal{S}\subseteq \mathcal{I}_1\\\mathcal{S}\ni i}}|\e_\mathcal{S}|,\quad\quad\forall i\in\mathcal{I}_1 \label{CombNetk2-achr1or}\\
&\text{Decoding constraints at private receivers:}\nonumber\\
&\hspace{1.5cm} R_2\leq \sum_{\substack{\mathcal{S}\in \Lambda}}\alpha_\mathcal{S}+\sum_{\substack{\mathcal{S}\in \Lambda^c}}\card{\e_{\mathcal{S},p}},\quad\quad \forall \Lambda
\subseteq 2^{\mathcal{I}_1}\text{ superset saturated},\ \forall p\in\mathcal{I}_2\label{CombNetk2-ach1or}\\
&\hspace{1.5cm} R_1+R_2\leq \sum_{\mathcal{S}\subseteq \mathcal{I}_1}|\e_{\mathcal{S},p}|,\quad\quad \forall p\in\mathcal{I}_2.\label{CombNetk2-achR1+R2or}
\end{align}
\end{theorem}
\begin{remark}
Note that  compared to Proposition \ref{CombNetp1-innerbound}, the non-negativity constraint on $\alpha_\phi$ is relaxed in Theorem \ref{lb-CombNet-Theoremk2}.
\end{remark}
\begin{remark}
The inner-bound of Theorem \ref{lb-CombNet-Theoremk2} is tight for combination networks with $m=3$ (or fewer) public and any number of private receivers, see Section \ref{CombNet-subsect-outer2}. 
\end{remark}
\begin{proof}
{Let $(R_1,R_2)$ be in the rate region of Theorem \ref{lb-CombNet-Theoremk2}; i.e., there exist 
parameters $\alpha_\mathcal{S}$, $\mathcal{S}\subseteq \mathcal{I}_1$, that satisfy inequalities \eqref{CombNetk2-achpos}-\eqref{CombNetk2-achR1+R2or}.  In the following, let $(\alpha_\phi)^-=\min(0,\alpha_\phi)$ and $(\alpha_\phi)^+=\max(0,\alpha_\phi)$. Furthermore, without loss of generality, we assume  that parameters $R_1$, $R_2$, $\alpha_\mathcal{S}$, $\mathcal{S}\subseteq \mathcal{I}_1$, are integer valued (see Remark \ref{lb-CombNet-alphapos}).}

{First of all, pre-encode message $W_2$ into a message vector $W^\prime_2$ of dimension $R_2-(\alpha_\phi)^-$ through a pre-encoding matrix $\mathbf{P}\in {\mathbb{F}_q}^{[R_2-(\alpha_\phi)^-]\times R_2}$; i.e., we have
\begin{align}
\label{lb-CombNet-PreEnc}
W^\prime_2=\mathbf{P}W_2.
\end{align}
Then, encode messages $W_1$ and $W^\prime_2$ using a  zero-structured 
matrix with rate split parameters $\alpha_\mathcal{S}$, $S\subseteq \mathcal{I}_1$, omitting the columns corresponding to $S=\phi$ if $\alpha_\phi<0$. The encoding matrix is, therefore, given as follows:}
\vspace{-.5cm}

{\begin{align}
\label{lb:CombNet:structureP}
 \begin{array}{c}\\\mathbf{A}=\end{array}
 \begin{array}{l}   
\begin{array}{ccccc}\hspace{0.1cm}\stackrel{\ \hspace{.25cm}R_1\hspace{.25cm}\ }{\longleftrightarrow}&\hspace{-.3cm}\stackrel{\alpha_{\{1,2\}}}{\longleftrightarrow}&\hspace{-.3cm}\stackrel{\alpha_{\{1\}}}{\leftrightarrow}&\hspace{-.35cm}\stackrel{\alpha_{\{2\}}}{\leftrightarrow}&\hspace{-0.5cm}\stackrel{\hspace{.1cm}(\alpha_{\phi})^+}{\leftrightarrow}\end{array}\\
\left[\begin{array}{c|c|c|c|c}
                   \hspace{.8cm}\ &\ \hspace{.4cm}\ &0&0&0\\\hline
		    \hspace{.8cm}\ &\ \hspace{.4cm}\ &\hspace{.2cm}&0&0\\	\hline	   
		    \hspace{.8cm}\ &\ \hspace{.4cm}\ &0&\hspace{.2cm}&0\\\hline
		    \hspace{.8cm}\ &\ \hspace{.4cm}\ &\hspace{.2cm}&\hspace{.2cm}&\hspace{.2cm}
                  \end{array}\right]
	    \begin{array}{c}\updownarrow\\\updownarrow\\\updownarrow\\\updownarrow\end{array}
	   \hspace{-.3cm} {\tiny{\begin{array}{l}|\e_{\{1,2\}}|\\\\ |\e_{\{1\}}|\\\\|\e_{\{2\}}|\\\\|\e_{\phi}|\end{array}}}\cdot\ \mathbf{P}^\prime
\end{array}
\end{align}
where
\begin{align}
\mathbf{P}^\prime=\begin{array}{l}\\\left[\begin{array}{c|c}\substack{\\\\\\\\\\\\\\}\hspace{.1cm}\mathbf{I}_{R_1\times R_1}\hspace{.1cm}&0\\\hline\substack{\\\\\\\\\\\\\\}0&\hspace{.1cm}\mathbf{P}\hspace{.1cm}\end{array}\right]\end{array},
\end{align}
and all indeterminates are to be assigned from the finite field ${\mathbb{F}_q}$.}

{The conditions for decodability of $W_1$ at each public receiver $i\in \mathcal{I}_1$ are given by \eqref{CombNetk2-achr1or}, and the conditions for decodability of $W_1,W_2$ at each private receiver $p\in \mathcal{I}_2$ is given by \eqref{CombNetk2-ach1or}, \eqref{CombNetk2-achR1+R2or}. The proof is similar to the basic zero structured encoding scheme and is sketched in the following.}

{Let $\mathbf{A}_i$ be the submatrix of $\mathbf{A}$ that constitutes $Y_i$ at receiver $i$. We have 
\begin{align}
\mathbf{A}_i&=\left[\mathbf{B}^{(i)}|\mathbf{T}^{(i)}\right]\mathbf{P}^\prime\\
&=\left[\mathbf{B}^{(i)}|\mathbf{T}^{(i)}\mathbf{P}\right].
\end{align} 
As before, define $\mathbf{L}^{(i)}_{1}$ to be a largest submatrix of the columns of $\mathbf{T}^{(i)}\mathbf{P}$ that could be made full column rank (over all possible assignments).
Define $\mathbf{G}^{(i)}$ to be 
\begin{align}
\mathbf{G}^{(i)}=\left[\mathbf{B}^{(i)}|\mathbf{L}^{(i)}_{1}\right].
\end{align}
From Lemma \ref{lem-Lbuilt}, each receiver $i$ can decode $W_1$ from \eqref{lb-CombNet-structureAi} 
if $\mathbf{G}^{(i)}$, as defined above, is full column rank. This is possible provided that (see Lemma \ref{lb-lemGifull})
\begin{align}
R_1+\sum_{\substack{\mathcal{S}\subseteq \mathcal{I}_1\\ \mathcal{S}\ni i}}\alpha_{\mathcal{S}}\leq \sum_{\substack{\mathcal{S}\subseteq \mathcal{I}_1\\ \mathcal{S}\ni i}}|\mathcal{E}_{\mathcal{S}}|.
\end{align}}

{Decodability of $W_1$ at  private receiver $p\in\mathcal{I}_2$ is similarly guaranteed by \eqref{CombNetk2-achR1+R2or}. Finally, $\mathbf{T}^{(p)}\mathbf{P}$ could itself be full rank under \eqref{CombNetk2-ach1or}, see Appendix \ref{ap-CombNet-lemmaW1W2}, and therefore $W_2$ is decodable at private receivers (in addition to $W_1$). Since the variables involved in $\mathbf{B}^{(p)}$ and $\mathbf{T}^{(p)}\mathbf{P}$ are independent,  $\mathbf{G}^{(p)}$ could be made full column rank.}

{It remains to argue that for $|{\mathbb{F}_q}|>K$, and under \eqref{CombNetk2-achr1or}-\eqref{CombNetk2-achR1+R2or}, all matrices $\mathbf{G}^{(i)}$, $i\in\mathcal{I}$, could be made full rank simultaneously. The proof is based on the sparse zeros lemma and is deferred to Appendix \ref{sparselemmamix}. }

%

\end{proof}

We close this section with
 an example that shows the inner-bound of Theorem \ref{lb-CombNet-Theoremk2} is not in general tight when the number of public receivers exceeds $3$.
\begin{example}
\label{lb-CombNet-example3to4}
\begin{figure}
\centering
\begin{tikzpicture}[scale=2]

\tikzstyle{every node}=[draw,shape=circle,font=\small\itshape]; 

\path (3,3) node (v0) {$S$}; 
\path (3,3.75) node[rectangle,draw=none] () {$\hspace{-1.75cm}W_1=[w_{1,1}]$}; 
\path (3,3.375) node[rectangle,draw=none] () {$\hspace{-.25cm}W_2=[w_{2,1},w_{2,2},w_{2,3}]$}; 

\path (0,1.5) node (v1) {};
\path (6/5,1.5) node (v2) {};
\path (12/5,1.5) node (v3) {};
\path (18/5,1.5) node (v4) {};
\path (24/5,1.5) node (v5) {};
\path (6,1.5)  node (v6) {};

\path (0,0) node (d1) {$D_1$};
\path (1,0) node (d2) {$D_2$};
\path (2,0) node  (d3) {$D_3$};
\path (3,0) node (d4)  {$D_4$};
\path (4,0) node (d5)[shade]  {$D_5$};
\path (5,0) node (d6)[shade]  {$D_6$};
\path (6,0) node (d7)[shade]  {$D_7$};

%

\draw[->] (v0) --node[draw=none,fill=none,left] {$X_{\{1,2\}}$}  (v1);

\draw[->] (v0) --node[draw=none,fill=none,left]{$X_{\{1,3\}}$}   (v2);
\draw[->] (v0) --node[draw=none,fill=none,left] {$\color{black}{X_{\{1,4\}}}$} (v3);
\draw[->] (v0) --node[draw=none,fill=none,right] {$X_{\{2,3\}}$} (v4);

\draw[->] (v0) --node[draw=none,fill=none,right] {$\color{black}{X_{\{2,4\}}}$}  (v5);
\draw[->] (v0) --node[draw=none,fill=none,right]{$\color{black}{X_{\{3,4\}}}$}  (v6);

%
\draw[->] (v1) -- (d1);
\draw[->] (v1) -- (d2);
\draw[->] (v2) -- (d1);
\draw[->] (v2) -- (d3);
\draw[->] (v3) -- (d1);
\draw[->] (v3) -- (d4);
\draw[->] (v4) -- (d2);
\draw[->] (v4) -- (d3);
\draw[->] (v5) -- (d2);
\draw[->] (v5) -- (d4);

\draw[->] (v6) -- (d3);
\draw[->] (v6) -- (d4);

\draw[->] (v1) -- (d5);
\draw[->] (v2) -- (d5);
\draw[->] (v3) -- (d5);
\draw[->] (v5) -- (d5);

\draw[->] (v1) -- (d6);
\draw[->] (v4) -- (d6);
\draw[->] (v5) -- (d6);
\draw[->] (v6) -- (d6);

\draw[->] (v2) -- (d7);
\draw[->] (v4) -- (d7);
\draw[->] (v6) -- (d7);
\draw[->] (v3) -- (d7);

\end{tikzpicture}
\caption{The innerbound of Theorem \ref{lb-CombNet-Theoremk2} is not tight for $m>3$ public receivers: while the rate pair $(1,3)$ is not in the rate-region of Theorem \ref{lb-CombNet-Theoremk2}, the code  in \eqref{codedesign4} is a construction that allows communication at rates $R_1=1$, $R_2=3$.}
\label{CombNet-example3to4}
\end{figure}
{Consider the combination network of Fig.~\ref{CombNet-example3to4}. Destinations $1,2,3, 4$ are public and destinations $5,6,7$ are private receivers. The rate pair $(1,3)$ is achievable over this network by the following code design:
\begin{align}
\begin{array}{l}X_{\{1,2\}}={w_{1,1}}+w_{2,1}\\
X_{\{1,3\}}={w_{1,1}}+w_{2,2}\\
X_{\{1,4\}}={w_{1,1}}+w_{2,1}+w_{2,2}\\
X_{\{2,3\}}={w_{1,1}}+w_{2,3}\\
X_{\{2,4\}}={w_{1,1}}+w_{2,1}+w_{2,3}\\
X_{\{3,4\}}={w_{1,1}}+w_{2,2}-w_{2,3}.
\end{array}
\label{codedesign4}
\end{align}
However, for this rate pair,  the problem in \eqref{CombNetk2-achpos}-\eqref{CombNetk2-achR1+R2or}  is infeasible (this may be verified by Fourier-Motzkin elimination or a LP feasibility solver).}
\end{example}

Let us look at this example (see Fig.~\ref{CombNet-example3to4}) more closely. Each public receiver has three resources available and needs to decode (only) the common message which is of rate $R_1=1$. So no more than two dimensions of the private message space could be revealed to any public receiver. For example, look at the resources that are available to public receiver $4$. These three resources mimic the  same structure of Fig.~\ref{CombNet-example2to3} and demand a certain dependency among their (superposed) private information symbols. More precisely, the superposed private information symbols carried over these resources come from a message space of dimension two (a condition imposed by public receiver $4$), and every two out of three of these resources should carry mutually independent private information symbols (a condition imposed by  the private receivers). A similar dependency structure is needed among the information symbols on $X_{\{1,2\}}$, $X_{\{1,3\}}$, $X_{\{1,4\}}$, and also among $X_{
\{1,2\}}$, $X_{\{2,3\}}$, $X_{\{2,4\}}$ and $X_{\{1,3\}}$, $X_{\{2,3\}}$, $X_{\{3,4\}}$. This shows a more involved dependency structure among the revealed partial private information symbols, and explains 
 why our modified encoding scheme of Theorem \ref{lb-CombNet-Theoremk2} cannot be optimal for this example. 
 
 Now consider the coding scheme of \eqref{codedesign4} that achieves $(1,3)$ (we assume $|{\mathbb{F}_q}|>2$).
This code ensures decodability of $W_1,w_{2,1},w_{2,2}$ at public receiver $1$, decodability of $W_1,w_{2,1}$, $w_{2,3}$ at public receiver $2$, decodability of $W_1,w_{2,2}$, $w_{2,3}$ at public receiver $3$, decodability of $W_1,w_{2,1}+w_{2,2},w_{2,1}+w_{2,3}$ at public receiver $4$ and decodability of $W_1,W_2$ at private receivers $5,6,7$. The (partial) private information that is revealed to different subsets of the public receivers is also as follows: no private information is revealed to subsets ${\{1,2,3,4\}}$, ${\{1,2,3\}}$, ${\{1,2,4\}}$, ${\{1,3,4\}}$, ${\{2,3,4\}}$, the private information revealed to ${\{1,2\}}$ is $w_{2,1}$, to ${\{1,3\}}$ is $w_{2,2}$, to ${\{2,3\}}$ is $w_{2,3}$, to ${\{1,4\}}$ is $w_{2,1}+w_{2,2}$,  to ${\{2,4\}}$ is $w_{2,1}+w_{2,3}$, to ${\{3,4\}}$ is $w_{2,2}-w_{2,3}$, and finally no private information is revealed to ${\{1\}}$, ${\{2\}}$, ${\{3\}}$, ${\{4\}}$. Now, it becomes clearer that the dependency structure that is needed among the partial private
information may, in general, be more involved than is allowed by our simple pre-encoding technique.

In the next section, we develop a simple block Markov encoding scheme to tackle Example \ref{lb-CombNet-example3to4}, and derive a new achievable scheme for the general problem.

\section{A Block Markov Encoding Scheme}
\label{CombNet-ach-blockMarkov}
\subsection{Main idea}
We saw that the difficulty in the code design stems mainly from the following tradeoff:
On one hand, private receivers require all the public and private information symbols and therefore prefer to receive mutually independent information over their resources. On the other hand,  each public receiver requires that its received encoded symbols do not depend on ``too many" private information symbols. This imposes certain dependencies among the encoded symbols of the public receivers' resources .

The main idea in this section is to capture these dependencies over sequential blocks, rather than capturing it through the structure of the {one-time} (one-block) code.
For this, we use a block Markov encoding scheme.
We start with an example where both previous schemes were sub-optimal and we show the optimality of the block Markov encoding scheme for it.

\begin{example}
\label{CombNet-BME-example22-3}
Consider the combination network in Fig.~\ref{CombNet-example3to4}. We saw in Example \ref{lb-CombNet-example3to4} that the rate pair $(R_1=1,R_2=3)$ was not achievable by the zero-structured linear code, even after employing a random pre-encoder in the the modified scheme. In this example, we achieve the rate pair $(1,3)$ through a block Markov encoding scheme, and hence, show that block Markov encoding  could perform strictly better.
Let us first extend the combination network by adding one extra resource to the set $\e_{\{4\}}$, and connecting it to all the private receivers (see Fig.~\ref{CombNet-example3to4bmecode-final}). 
{Over this extended combination network, the larger rate pair $(R_1=1,R^\prime_2=4)$ is achievable  using a  zero-structured linear code characterized in \eqref{CombNetp1-achpos}-\eqref{CombNetp1-achR1+R2or}, \eqref{lb-CombNetp1-remark}:
\begin{align}
&\alpha_{\{1,2,3\}}=\alpha_{\{1,2,4\}}=\alpha_{\{1,3,4\}}=\alpha_{\{2,3,4\}}=1\\
&\alpha_\phi=\alpha_{\{1\}}=\alpha_{\{2\}}=\alpha_{\{3\}}= \alpha_{\{4\}}=\alpha_{\{1,2\}}=\alpha_{\{1,3\}}=\alpha_{\{1,2,3,4\}}=0.
\end{align}}
 One such code design is given in the following.
Let message $W^\prime_2=[w^\prime_{2,1},w^\prime_{2,2},w^\prime_{2,3},w^\prime_{2,4}]$ be the private message of the larger rate ($R^\prime_2=4$) which is to be communicated over the extended combination network, and let $X_{\{1,2\}}$, $X_{\{1,3\}}$, $X_{\{1,4\}}$, $X_{\{2,3\}}$, $X_{\{2,4\}}$, $X_{\{3,4\}}$, $X_{\{4\}}$ be  symbols that are sent over the extended combination network. The encoding is described below (we assume that $\mathbb{F}_q$ has a characteristic larger than $2$):
\begin{align}
\label{lb-CombNet-bme-example2}
\begin{array}{ll}
X_{\{1,2\}}=w_{1,1}+w^\prime_{2,3}&X_{\{2,3\}}=w_{1,1}+2w^\prime_{2,3}\\
X_{\{1,3\}}=2w_{1,1}+w^\prime_{2,3}&X_{\{2,4\}}=w_{1,1}+w^\prime_{2,2}\\
X_{\{1,4\}}=w_{1,1}+w^\prime_{2,1}&X_{\{3,4\}}=w_{1,1}+w^\prime_{2,4}\\
X_{\{4\}}=w_{1,1}+w^\prime_{2,1}+w^\prime_{2,2}+w^\prime_{2,4}.&
\end{array}
\end{align}
Since the resource edge in $\e_{\{4\}}$ is a virtual resource, we aim to emulate it through a block Markov encoding scheme. Using the code design of \eqref{lb-CombNet-bme-example2}, Receiver $4$ may decode, besides the common message, three private information symbols ($w^\prime_{2,1}$, $w^\prime_{2,2}$,  $w^\prime_{2,4}$), see also \eqref{lb-CombNetp1-remark} and Remark \ref{lb-CombNet-remarkdecodability}. Since all these three symbols are decodable at Receiver $4$ and all the private receivers, any of them could be used to emulate the virtual resource in $\e_{\{4\}}$. 

{More precisely,
consider communication over $n$ transmission blocks, and let $(W_1[t],W^{\prime}_2[t])$ be the message pair that is  encoded in block $t\in\{1,\ldots,n\}$. 
In the $t^{\text{th}}$ block, encoding is done as suggested by the code in \eqref{lb-CombNet-bme-example2}. 
To provide Receiver $4$ and the private receivers with the information of $X_{\{4\}}[t]$ (as promised by the virtual resource in $\e_{\{4\}}$), we use $w^{\prime}_{2,4}[t+1]$ in the next block, to convey $X_{\{4\}}[t]$. 
Since this symbol is ensured to be decoded at Receiver $4$ and the private receivers, it indeed emulates $\e_{\{4\}}$.
In the $n^{\text{th}}$ block, we simply encode $X_{\{4\}}[n-1]$ and directly communicate it with receiver~$4$ and the private receivers. 
Upon receiving all the $n$ blocks, the receivers perform backward decoding \cite{WillemsMeulen85}. }

So in $n$ transmissions, we may  send $n-1$ symbols of $W_1$ and $3(n-1)+1$ new symbols of $W_2$ over the original combination 
network;  i.e., for $n\to \infty$, we can achieve the rate-pair $(1,3)$. 

Note that public receivers each have four resources available and therefore rate pair $(1,3)$  is an optimal sum-rate point. 
\begin{figure}
\centering
\begin{tikzpicture}[scale=2]

\tikzstyle{every node}=[draw,shape=circle,font=\small\itshape]; 

\path (3,3) node (v0) {$S$};

\path  (3,3.75) node[rectangle,draw=none]() {$W_1[t+1]=[w_{1,1}[t+1]]$}; 
\path  (3,3.375) node[rectangle,draw=none](message) {$W_2^{\prime}[t+1]=[w^{\prime}_{2,1}[t+1],w^{\prime}_{2,2}[t+1],w^{\prime }_{2,3}[t+1],\textcolor{blue}{w^{\prime t+1}_{2,4}[t+1]}]$}; 

\path (0,1.5) node (v1) {};
\path (6/5,1.5) node (v2) {};
\path (12/5,1.5) node (v3) {};
\path (18/5,1.5) node (v4) {};
\path (24/5,1.5) node (v5) {};
\path (6,1.5)  node (v6) {};
\path (7.5,1.5)  node[color=blue,fill] (vrt) {};

\path (0,0) node (d1) {$D_1$};
\path (1,0) node (d2) {$D_2$};
\path (2,0) node  (d3) {$D_3$};
\path (3,0) node (d4)  {$D_4$};
\path (4,0) node (d5)[shade]  {$D_5$};
\path (5,0) node (d6)[shade]  {$D_6$};
\path (6,0) node (d7)[shade]  {$D_7$};

%

\draw[->] (v0) --node[draw=none,fill=none,yshift=-.4cm] {$X_{\!\{1,2\}}\![t]$}  (v1);

\draw[->] (v0) --node[draw=none,fill=none,yshift=-.4cm,xshift=.2cm] {$X_{\!\{1,3\}}\![t]$}  (v2);
\draw[->] (v0) --node[draw=none,fill=none,yshift=-.4cm,xshift=.4cm] {$X_{\!\{1,4\}}\![t]$}  (v3);
\draw[->] (v0) --node[draw=none,fill=none,yshift=-.4cm,xshift=.75cm] {$X_{\!\{2,3\}}\![t]$}  (v4);

\draw[->] (v0) --node[draw=none,fill=none,yshift=-.4cm,xshift=1.15cm] {$X_{\!\{2,4\}}\![t]$}  (v5);
\draw[->] (v0) --node[draw=none,fill=none,yshift=-.4cm,xshift=1.6cm] {$X_{\!\{3,4\}}\![t]$}  (v6);
\draw[->,very thick,blue] (v0) --node[draw=none,fill=none,yshift=-.4cm,xshift=2.1cm] {$\textcolor{blue}{X_{\{4\}}}\![t]$}  (vrt);
\draw[->,blue,dashed,very thick ](6.7,2.1) to [out=0,in=-90] (7.5,3.1)to [out=90,in=45] (4.5,3.5);

\draw[->,gray!75] (v1) -- (d1);
\draw[->,gray!75] (v1) -- (d2);
\draw[->,gray!75] (v2) -- (d1);
\draw[->,gray!75] (v2) -- (d3);
\draw[->,gray!75] (v3) -- (d1);
\draw[->,gray!75] (v3) -- (d4);
\draw[->,gray!75] (v4) -- (d2);
\draw[->,gray!75] (v4) -- (d3);
\draw[->,gray!75] (v5) -- (d2);
\draw[->,gray!75] (v5) -- (d4);

\draw[->,gray!75] (v6) -- (d3);
\draw[->,gray!75] (v6) -- (d4);

\draw[->] (v1) -- (d5);
\draw[->] (v2) -- (d5);
\draw[->] (v3) -- (d5);
\draw[->] (v5) -- (d5);

\draw[->] (v1) -- (d6);
\draw[->] (v4) -- (d6);
\draw[->] (v5) -- (d6);
\draw[->] (v6) -- (d6);

\draw[->] (v2) -- (d7);
\draw[->] (v4) -- (d7);
\draw[->] (v6) -- (d7);
\draw[->] (v3) -- (d7);

\draw[->,very thick,blue] (vrt) -- (d5);
\draw[->, very thick,blue] (vrt) -- (d6);
\draw[->,very thick,blue] (vrt) -- (d7);
\draw[->,very thick,blue] (vrt) -- (d4);

\end{tikzpicture}
\caption{The extended combination network of Example \ref{CombNet-BME-example22-3}. A block Markov encoding scheme achieves the rate pair $(1,3)$ over the original combination network. At time $t+1$, $w^{\prime}_{2,4}[t+1]$ contains the information of symbol $X_{\{4\}}[t]$. }
\label{CombNet-example3to4bmecode-final}
\end{figure} \end{example}

\subsection{The block Markov encoding scheme: an achievable region}
In both Examples \ref{CombNet-BME-example12-3} and \ref{CombNet-BME-example22-3}, the achievability is through a block Markov encoding scheme, and the construction of it is explained with the help of an \textit{extended combination network}.
Before further explaining this construction, let us clarify what we mean by an extended combination network.

\begin{definition}[Extended combination network]
An extended combination network is formed from the original combination network by adding some extra nodes,  called \textit{virtual nodes}, to the intermediate layer. The source is connected to all of the virtual nodes through edges that we call \textit{virtual resources}. Each virtual resource is connected to a subset of  receivers which we refer to as the end-destinations of that virtual resource. This subset is chosen, depending on the structure of the original combination network and the target rate pair, through an optimization problem that we will address later in this section.
\end{definition}

The idea behind extending the combination network is as follows. The encoding is such that in order to decode the common and private messages in block $t$, each receiver may need the information that it will decode in block $t+1$ (recall that receivers perform backward decoding).
So,  the source wants to design both its outgoing symbols in block $t$ and the side information that the receiver will have in block $t+1$. This is captured by designing a code over an extended combination network, where the virtual resources play, in a sense, the role of the side information. 

Over the extended combination networks, we will design a \textit{general multicast code} (as opposed to one for nested message sets). We will only consider multicast codes based on basic linear superposition coding to emulate the virtual resources (see Remark \ref{remarkbme}). We further elaborate on this in the following.

\begin{definition}[Emulatable virtual resources]
Given an extended combination network and a general multicast code over it,
a virtual resource $v$ is called emulatable if the multicast code allows reliable communication at a rate of at least $1$ to all end-destinations of that virtual resource (over the extended combination network). We call a set of virtual resources  emulatable if they are all simultaneously emulatable.
\end{definition}

We now outline the steps in devising a block Markov encoding scheme for this problem. 
\begin{enumerate}
\item Add a set of virtual resources to the original combination network to form an {extended combination network}. 
\item Design a general (as opposed to one for nested message sets) multicast code over the extended combination network such that all the virtual resources are emulatable. 
\item Use the multicast code to make all the virtual resources {emulatable}. More precisely, use the information symbols in block $t+1$ to also convey the information carried on the virtual resources in block $t$. Use the remaining information symbols to communicate the common and private information symbols.
\end{enumerate}
An achievable rate-region could then be found by optimizing over the virtual resources and the multicast code. 

Formulating this problem in its full generality is not the goal of this section. We instead aim to construct a simple block Markov encoding scheme, show its advantages in  code design, and characterize a region achievable by it. To this end, we confine ourselves to the following two assumptions: 
(i) the virtual resources that we introduce are connected to all private receivers and different subsets of  public receivers, and (ii) the multicast code that we design over the extended combination network is a {basic} linear superposition code along the lines of the codes in Section \ref{lb-CombNet-ZES}.


In order to devise our simple block Markov scheme, we first create an extended combination network by adding for every  $\mathcal{S}\subseteq \mathcal{I}_1$, $\beta_\mathcal{S}$ many virtual resources which are connected to all the private receivers and all the public receivers in $\mathcal{S}\subseteq \mathcal{I}_1$ (and only those). We denote this subset of virtual resources by $\mathcal{V}_\mathcal{S}$.

Over this extended combination network, we then design a (general) multicast code. 
We say that a multicast code achieves rate tuple $(R_1,\alpha_{\{1,\ldots,m\}},\ldots,\alpha_\phi)$ over the extended combination network, if it reliably communicates a message of rate $R_1$ to all receivers, and independent messages of rates $\alpha_\mathcal{S}$, $\mathcal{S}\subseteq \mathcal{I}_1$, to all public receivers in $\mathcal{S}$ and all private receivers. 
To design such a multicast code, we use a basic linear superposition code (i.e.,  a zero-structured encoding matrix). It turns out that the rate tuple \scalebox{.9}[1]{$(R_1,\alpha_{\{1,\ldots,m\}},\ldots,\alpha_\phi)$} is achievable if the following inequalities are satisfied (see Remark \ref{lb-CombNet-remarkdecodability}):
\begin{align}
&\hspace{0cm}\text{Decodability constraints at public receivers:}\nonumber\\
&\hspace{1.5cm}\sum_{\substack{\mathcal{S}\subseteq \mathcal{I}_1\\\mathcal{S}\ni i}}\!\!\alpha_\mathcal{S}\leq\sum_{\substack{\mathcal{S}\in\Lambda}}\!\!\alpha_\mathcal{S}\!+\!\!\sum_{\substack{\mathcal{S}\in\Lambda^c\\\mathcal{S}\ni i}}\!\!\left(\card{\e_\mathcal{S}}+\beta_\mathcal{S}\right),\quad \ {\forall \Lambda\subseteq \{\{i\}\star\}\ \text{ superset saturated}},\ \forall i\in \mathcal{I}_1\label{CombNet-bm-r1}\\
&\hspace{1.5cm}R_1+\sum_{\substack{\mathcal{S}\subseteq \mathcal{I}_1\\\mathcal{S}\ni i}}\!\!\alpha_\mathcal{S}\leq \sum_{\substack{\mathcal{S}\subseteq \mathcal{I}_1\\\mathcal{S}\ni i}}\!\!\left(\card{\e_\mathcal{S}}+\beta_\mathcal{S}\right),\ \forall i\in \mathcal{I}_1\\
&\hspace{0cm}\text{Decodability constraints at private receivers:}\nonumber\\
&\hspace{1.5cm}R^\prime_2\leq \sum_{\mathcal{S}\in\Lambda}\!\!\alpha_\mathcal{S}+\sum_{\mathcal{S}\in\Lambda^c}\!\!\left(\card{\e_{\mathcal{S},p}}+\beta_\mathcal{S}\right),\quad{\forall \Lambda\subseteq 2^{\mathcal{I}_1}\ \text{ superset saturated}},\ \forall p\in \mathcal{I}_2\\
&\hspace{1.5cm}R_1+R^\prime_2\leq\sum_{\mathcal{S}\subseteq \mathcal{I}_1}\!\!\left(\card{\e_{\mathcal{S},p}}+\beta_\mathcal{S}\right),\ \forall p\in \mathcal{I}_2,
\end{align}
where 
\begin{align}
R_2^\prime=\sum_{\mathcal{S}\subseteq \mathcal{I}_1} \alpha_\mathcal{S}.\label{sumrateR2p}
\end{align}

Given such a multicast code, we find conditions for the virtual resources to be emulatable. The proof is relegated to Appendix~\ref{ap-CombNet-bm-lem-match}.
\begin{lemma}
\label{lb-CombNet-bm-lem-match}
Given an extended combination network with $\beta_\mathcal{S}$ virtual resources $\mathcal{V}_\mathcal{S}$, $\mathcal{S}\subseteq \mathcal{I}_1$, and a multicast code design that achieves the rate tuple  $(R_1,\alpha_{\{1,\ldots,m\}},\ldots,\alpha_\phi)$, all virtual resources are emulatable provided that the following set of inequalities hold.
\begin{align}
\label{lb-CombNet-bm-match}
\sum_{\mathcal{S}\in\Lambda}\beta_\mathcal{S}\leq\sum_{\mathcal{S}\in\Lambda}\alpha_\mathcal{S},\quad \forall \Lambda\subseteq 2^{\mathcal{I}_1} \text{superset saturated}.
\end{align}
\end{lemma}

It remains to characterize the common and private rates that our simple block Markov encoding scheme achieves over the original 
combination network. To do so, we disregard 
the information symbols that are used to emulate the virtual resources, for they bring 
redundant information, and characterize the remaining rate of the common and private 
information symbols. In the above scheme, this is simply 
$(R_1,R_2^\prime-\sum_{\mathcal{S}\subseteq \mathcal{I}_1}\beta_\mathcal{S})$, 
where the real valued parameters $\alpha_\mathcal{S}, \beta_\mathcal{S}$  satisfy 
inequalities \eqref{CombNet-bm-r1}-\eqref{sumrateR2p} and the following non-negativity constraints:
\begin{eqnarray}
\label{lb-CombNet-nonnegbeta}
&\alpha_\mathcal{S}\geq 0,\\
&\beta_\mathcal{S}\geq 0.
\end{eqnarray}

To simplify the representation, we define
$\gamma_\mathcal{S}=\alpha_\mathcal{S}-\beta_\mathcal{S},\ \forall \mathcal{S}\subseteq \mathcal{I}_1$,
and then eliminate $\alpha$'s and $\beta$'s from all inequalities involved. We thus have the following theorem.

\begin{theorem}
\label{lb-CombNetbme-Theorem}
The rate pair $(R_1,R_2)$ is achievable if there exist parameters $\gamma_\mathcal{S}$, $\mathcal{S}\subseteq \mathcal{I}_1$, such that they satisfy the following inequalities:
\begin{align}
&\hspace{-.65cm}\sum_{\mathcal{S}\in\Lambda}\gamma_\mathcal{S}\geq0,\quad\quad\forall \Lambda\subseteq 2^{\mathcal{I}_1}\text{ superset saturated}\label{CombNetbme-posor} \\
&\hspace{-.65cm}R_2=\sum_{\mathcal{S}\subseteq \mathcal{I}_1} \gamma_\mathcal{S},\\
&\hspace{-.65cm}\text{Decodability constraints at public receivers:}\nonumber\\
&\hspace{-.65cm}\sum_{\substack{\mathcal{S}\subseteq \mathcal{I}_1\\\mathcal{S}\ni i}}\gamma_\mathcal{S}\leq\sum_{\substack{\mathcal{S}\in\Lambda}}\gamma_\mathcal{S}+\sum_{\substack{\mathcal{S}\in\Lambda^c\\\mathcal{S}\ni i}}\card{\e_\mathcal{S}},\quad { \forall \Lambda\subseteq \{\{i\}\star\}\ \text{ superset saturated}},\ \forall i\in \mathcal{I}_1 \label{CombNetbme-subr1or}\\
&R_1+\sum_{\substack{\mathcal{S}\subseteq \mathcal{I}_1\\\mathcal{S}\ni i}}\gamma_\mathcal{S}\leq \sum_{\substack{\mathcal{S}\subseteq \mathcal{I}_1\\\mathcal{S}\ni i}}\card{\e_\mathcal{S}},\quad \forall i\in \mathcal{I}_1 \\
&\hspace{-.65cm}\text{Decodability constraints at private receivers:}\nonumber\\
&\hspace{-.65cm}R_2\leq \sum_{\mathcal{S}\in\Lambda}\gamma_\mathcal{S}+\sum_{\mathcal{S}\in\Lambda^c}\card{\e_{\mathcal{S},p}},\quad{ \forall \Lambda\subseteq 2^{\mathcal{I}_1}\ \text{ superset saturated}},\ \forall p\in \mathcal{I}_2 \\
&R_1+R_2\leq\sum_{\mathcal{S}\subseteq \mathcal{I}_1}\card{\e_{\mathcal{S},p}},\quad \forall p\in \mathcal{I}_2 .
\end{align}
\end{theorem}
{
Comparing the rate-regions in Theorem \ref{lb-CombNet-Theoremk2} and Theorem \ref{lb-CombNetbme-Theorem}, we see that the former has a more relaxed set of inequalities in \eqref{CombNetbme-posor} while the latter is more relaxed in inequalities  \eqref{CombNetbme-subr1or}. 
Although the two regions are not comparable in general, it turns out that for $m\leq 3$, the two rate-regions coincide and characterize the capacity region (see Theorem \ref{lb-CombNet-Theorem-outer2} and Theorem \ref{lb-CombNet-Theorem-outer3}). Furthermore, the combination network in Fig. \ref{CombNet-example3to4} serves as an instance where the rate-region in Theorem \ref{lb-CombNetbme-Theorem} includes rate pairs that are not included in the region of Theorem \ref{lb-CombNet-Theoremk2} (see Examples \ref{lb-CombNet-example3to4} and \ref{CombNet-BME-example22-3}).

Fig.~\ref{shekl} plots the rate-regions of Theorem \ref{lb-CombNet-Theoremk2} and Theorem \ref{lb-CombNetbme-Theorem} for the network of Fig.~\ref{CombNet-example3to4}. The grey region is  the region of Theorem \ref{lb-CombNet-Theoremk2} (i.e., achievable using pre-encoding, rate-splitting, and linear superposition coding) and the red shaded region is the region of Theorem \ref{lb-CombNetbme-Theorem} (i.e.,  achievable using the proposed block Markov encoding scheme). In this example, the proposed block Markov encoding scheme strictly outperforms our previous schemes.
\begin{figure}
\centering
\includegraphics[width=.7\textwidth]{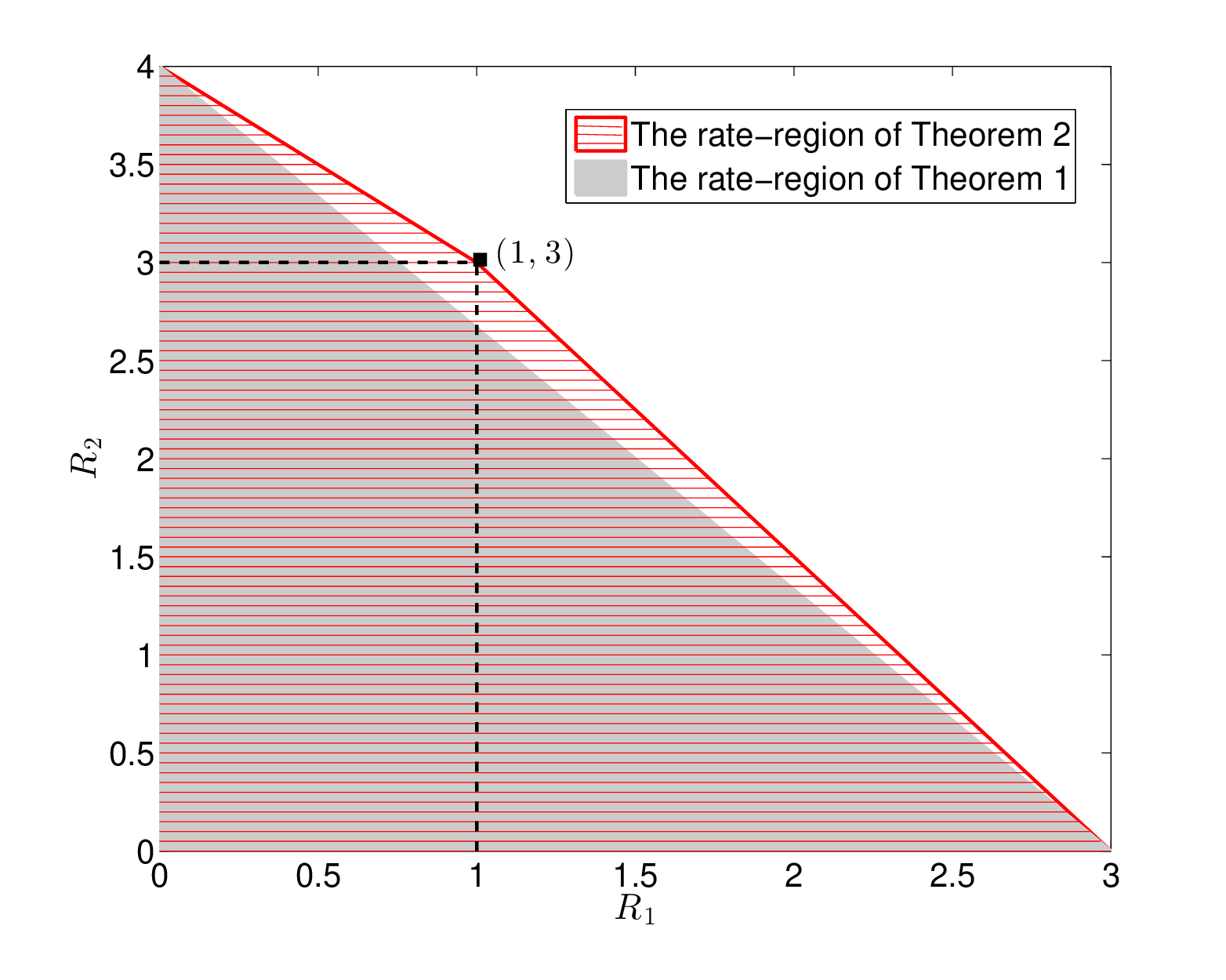}
\caption{The rate-regions of Theorems  \ref{lb-CombNet-Theoremk2} and \ref{lb-CombNetbme-Theorem} for the combination network in Fig.~\ref{CombNet-example3to4}}
\label{shekl}
\end{figure}


\begin{remark}
 It remains open whether the rate-region of Theorem \ref{lb-CombNetbme-Theorem} always includes the rate-region of Theorem \ref{lb-CombNet-Theoremk2}, or not. We conjecture that this is true.
\end{remark}
\section{Optimality Results}
\label{CombNet-outerbound}
In this section, we prove our optimality results. 
More precisely, we prove optimality of the zero-structured encoding scheme of Subsection \ref{lb-CombNet-ZES} when $m=2$, optimality of the structured linear code with pre-encoding discussed in Subsection \ref{lb-CombNet-MES} when $m=3$ (or fewer), and optimality of the block Markov encoding of Section \ref{CombNet-ach-blockMarkov} when $m= 3$ (or fewer). This is summarized in the following theorems.
\begin{theorem}
\label{lb-CombNet-Theorem-outer1}
Over a combination network with two public and any number of private receivers, the rate pair $(R_1,R_2)$ is achievable if and only if it lies in the rate-region of Proposition \ref{CombNetp1-innerbound}.
\end{theorem}

\begin{theorem}
\label{lb-CombNet-Theorem-outer2}
Over a combination network with three (or fewer) public and any number of private receivers, the rate pair $(R_1,R_2)$ is achievable if and only if it lies in the rate-region of Theorem \ref{lb-CombNet-Theoremk2}.
\end{theorem}

\begin{theorem}
\label{lb-CombNet-Theorem-outer3}
Over a combination network with three (or fewer) public and any number of private receivers, the rate pair $(R_1,R_2)$ is achievable if and only if it lies in the rate-region of Theorem \ref{lb-CombNetbme-Theorem}.
\end{theorem}

\subsection{Explicit projection of the polyhedron and the proof of Theorem \ref{lb-CombNet-Theorem-outer1}}
\label{CombNet-subsect-outer1}
The achievability part of  Proposition \ref{CombNetp1-innerbound} was discussed in Section \ref{lb-CombNet-ZES}. We prove the converse here.
Using Fourier-Motzkin elimination method, we first eliminate all parameters $\alpha_\mathcal{S}$, $\mathcal{S}\subseteq \mathcal{I}_1$, in the rate-region of Proposition \ref{CombNetp1-innerbound} and we obtain the following region (recall that $\mathcal{I}_1=\{1,2\}$ and $\mathcal{I}_2=\{3,\ldots,K\}$): 
\begin{align}
&R_1\leq \min\left(|\e_{\{1\}}|+|\e_{\{1,2\}}|,|\e_{\{2\}}|+|\e_{\{1,2\}}|\right)\label{lb-CombNet2-outer1}\\
&R_1+R_2\leq\min_{p\in \mathcal{I}_2}\left\{|\e_{\phi,p}|+|\e_{\{1\},p}|+|\e_{\{2\},p}|+|\e_{\{1,2\},p}|\right\}\label{lb-CombNet2-outer2}\\
&2R_1+R_2\leq \min_{p\in \mathcal{I}_2}\left\{|\e_{\{1\}}|+2|\e_{\{1,2\}}|+|\e_{\{2\}}|+|\e_{\phi,p}|\right\}.\label{lb-CombNet2-outer3}
\end{align}
Note that the right hand side (RHS) of \eqref{lb-CombNet2-outer1} is the minimum of the min-cuts to the two public receivers, and the RHS of \eqref{lb-CombNet2-outer2} is the minimum of the min-cuts to the private receivers.

Our converse proof is  similar to \cite{PrabhakaranDiggaviTse07}.
Inequalities \eqref{lb-CombNet2-outer1} and \eqref{lb-CombNet2-outer2} are immediate (using cut-set bounds) and are easy to derive.  Inequality \eqref{lb-CombNet2-outer3} is, however, not immediate and we prove it in the following.
Assume communication over blocks of length $n$. 
The rate $R_2$ is bounded, for each private receiver $p\in{I}_2$ and any $\epsilon>0$, as follows:
\begin{align*}
nR_2=&H(W_2|W_1)\\
\leq&H(W_2|W_1)- H(W_2|W_1Y^n_p)+ H(W_2|W_1Y^n_p)\\
\stackrel{(a)}{\leq}&I(W_2;Y^n_p|W_1)+n\epsilon\\
=&H(Y^n_p|W_1)+n\epsilon\\
\stackrel{(b)}{\leq}&H(X^n_{\{\phi,\{1\},\{2\}\{1,2\}\},p}X^n_{\{\{1\},\{2\},\{1,2\}\}}|W_1)+n\epsilon\\
\stackrel{(c)}{\leq}&H(X^n_{\{\{1\},\{1,2\}\}}|W_1)+H(X^n_{\{\{2\},\{1,2\}\}}|W_1)+H(X^n_{\{\phi,\{1\},\{2\}\{1,2\}\},p}|X^n_{\{\{1\},\{2\},\{1,2\}\}}W_1)+n\epsilon\\
\stackrel{(d)}{\leq}&H(X^n_{\{\{1\},\{1,2\}\}})+H(X^n_{\{\{2\},\{1,2\}\}})-2nR_1+H(X^n_{\{\phi,\{1\},\{2\}\{1,2\}\},p}|X^n_{\{\{1\},\{2\},\{1,2\}\}}W_1)+3n\epsilon\\
\stackrel{(e)}{\leq}&H(X^n_{\{\{1\},\{1,2\}\}})+H(X^n_{\{\{2\},\{1,2\}\}})-2nR_1+H(X^n_{\phi,p})+3n\epsilon\\
\stackrel{(f)}{\leq}&n(\card{\e_{\{1\}}}+\card{\e_{\{1,2\}}})+n(\card{\e_{\{2\}}}+\card{\e_{\{1,2\}}})-2nR_1+n(\card{\e_{\phi,p}})+3n\epsilon,
\end{align*}
In the above chain of inequalities, step $(a)$ follows from Fano's inequality. Step $(b)$ follows because the received 
sequence $Y^n_p$ is given by all symbols in $X^n_{\{\phi,\{1\},\{2\}\{1,2\}\},p}$ and we further add all 
symbols $X^n_{\{1,2\}}$,$X^n_{\{1\}}$,$X^n_{\{2\}}$. Step $(c)$ follows from sub-modularity of entropy.  Step $(d)$ is a 
result of \eqref{fano0}-\eqref{fano1}, below, where \eqref{fano1} is due to Fano's inequality ($W_1$ should be recoverable with arbitrarily small error probability from $X^n_{\{\{1\},\{1,2\}\}}$).
\begin{align}
H(X^n_{\{\{1\},\{1,2\}\}}|W_1)=&H(X^n_{\{\{1\},\{1,2\}\}}W_1)-nR_1\label{fano0}\\=&H(X^n_{\{\{1\},\{1,2\}\}})+H(W_1|X^n_{\{\{1\},\{1,2\}\}})-nR_1\\\leq& H(X^n_{\{\{1\},\{1,2\}\}})+n\epsilon-nR_1.\label{fano1}
\end{align}
Similarly, we have 
$$H(X^n_{\{\{2\},\{1,2\}\}}|W_1)\leq H(X^n_{\{\{2\},\{1,2\}\}})-nR_1+n\epsilon.$$
 Step $(e)$ follows from the fact that  for any $\mathcal{S}\subseteq \mathcal{I}_1$, $X^n_{\mathcal{S},p}$ is contained in $X^n_\mathcal{S}$ (by definition) and that conditioning reduces the entropy. Finally, step $(f)$ follows because each entropy term $H(X^n_\mathcal{S})$ is bounded by $n|\mathcal{S}|$ (remember that  all rates are written in units of $\log_2|\mathbb{F}_q|$ bits).


%
\begin{figure}
\centering
\begin{tikzpicture}[scale=.7]
\tikzstyle{every node}=[draw,shape=circle,minimum size=.001cm,font=\small\itshape]; 

\node (s) at (4,8) {$S$};

\node (d1) at (1,2) {$D_1$};
\draw [dashed,thick](4,5.5) arc (-90:-180:2cm) ;
\draw[dashed,thick] (6,6.3) arc (-70:-120:3.5cm) ;
\draw[dashed,thick] (6,6.75) arc (-70:-150:3cm) ;
\draw (1.6,7) node [draw=none]  {$C_1$};
\draw (6.3,6.25) node[draw=none]    {$C_2$};
\draw (6.3,7) node [draw=none]   {$C_3$};

\node (d2) at (4,2) {$D_2$};
\node (d3) at (7,2) [shade] {$D_3$};

\node (v1) at (1,5) {};
\node (v2) at (3,5) {};
\node (v3) at (5,5) {};
\node (v4) at (7,5) {};

\draw[->] (s) to node[draw=none,right,below]{$e_1$}(v1);
\draw[->] (s) to node[draw=none,right,below]{$\ \ e_2$}(v2);
\draw[->] (s) to node[draw=none,right,below,xshift=.3cm,yshift=-.1cm]{$e_3$}(v3);
\draw[->] (s) to node[draw=none,right,below,xshift=.6cm,yshift=-.1cm]{$e_4$} (v4);
\draw[->] (v1) --(d1);
\draw[->] (v2) --(d1);
\draw[->] (v2) --(d2);
\draw[->] (v3) --(d2);
\draw[->] (v4) --(d2);
\draw[->] (v1) --(d3);
\draw[->] (v2) --(d3);
\draw[->] (v3) --(d3);
\draw[->] (v4) --(d3);

\end{tikzpicture}
\caption{The rate-region in Theorem \ref{lb-CombNet-Theorem-outer1} evaluates to $R_1\leq 2$, $R_1+R_2\leq 4$ and $2R_1+R_2\leq 5$.}
\label{outerinterpret}
\end{figure}
We discuss an intuitive explanation of this outer-bound via the example in Fig.~\ref{outerinterpret}. Clearly, the common message $W_1$ could be reliably communicated with receiver $1$ (which has a min-cut equal to $2$) only if $R_1\leq 2$. Similarly, $R_1 \leq 3$ (according to the min-cut to receiver $2$) and $R_1 + R_2 \leq 4$ (according to the min-cut to receiver $3$). Now, consider the three cuts $C_1, C_2, C_3$ shown in Fig.~\ref{outerinterpret}. How much information about message $W_2$ could be carried over edges $e_1,e_2,e_3,e_4$, altogether? Edges of the cut $\{e_1,e_2\}$ can carry at most $2-R_1$  units of information about message $W_2$, for they have a total capacity of $2$ and have to also ensure decodability of message $W_1$ (which is of rate $R_1$). Similarly, edges of the cut $\{e_2,e_3,e_4\}$ can carry at most $3 - R_1$ units of information about message $W_2$. So altogether, these edges cannot carry more than $2 - R_1 + 3 - R_1$ bits of information about message $W_2$; i.e., $R_2 \leq 5 - 2R_
1$, or $2R_1 + R_2 \leq 5$.

\subsection{Sub-modularity of the entropy function and the proofs of Theorems \ref{lb-CombNet-Theorem-outer2} and  \ref{lb-CombNet-Theorem-outer3}}
\label{CombNet-subsect-outer2}

While it was not difficult to eliminate all parameters $\alpha_\mathcal{S}$, $\mathcal{S}\subseteq \mathcal{I}_1$, from the rate-region characterization for $m=2$, this becomes a tedious task when the number of public receivers increases. In this section, we prove Theorem \ref{lb-CombNet-Theorem-outer2}  by showing an outer-bound on the rate-region that matches the inner-bound of Theorem \ref{lb-CombNet-Theoremk2}  when $m=3$. We bypass the issue of explicitly eliminating all parameters $\alpha_\mathcal{S}$, $\mathcal{S}\subseteq \mathcal{I}_1$, by first proving an outer-bound which looks \textit{similar} to the inner-bound and then using sub-modularity of entropy to conclude the proof. The same converse technique will be used to prove Theorem \ref{lb-CombNet-Theorem-outer3}. Together, Theorem \ref{lb-CombNet-Theorem-outer2} and Theorem \ref{lb-CombNet-Theorem-outer3} allow us to conclude that the two regions in Theorems \ref{lb-CombNet-Theoremk2} and \ref{lb-CombNetbme-Theorem} coincide for $m=3$ public (and any number of private) receivers and characterize the capacity.
We start with an example.
\begin{example}
\label{CombNet-example-conv2}
Consider the combination network of Fig.~\ref{example-converse} where receivers $1,2,3$ are public and receivers $4,5$ are private receivers.
We ask if the rate pair $(R_1=1,R_2=2)$ is achievable over this network.
To answer this question, let us first see if this rate pair is within the inner-bound of Theorem \ref{lb-CombNet-Theoremk2}. 
By solving the feasibility problem defined in inequalities \eqref{CombNetk2-achpos}-\eqref{CombNetk2-achR1+R2or}, using Fourier-Motzkin elimination method, we obtain the following inner-bound inequality, and conclude that the rate pair $(1,2)$ is not within the inner-bound of Theorem~\ref{lb-CombNet-Theoremk2}.
\begin{align}
\label{lb-CombNet-ex-in-eq}
4R_1+2R_2\leq 7.
\end{align}
Once this is established, we can also answer the following question: what linear combination of inequalities in \eqref{CombNetk2-achpos}-\eqref{CombNetk2-achR1+R2or} gave rise to the inner-bound inequality in \eqref{lb-CombNet-ex-in-eq}? 
The answer is that summing \textit{two} copies of  \eqref{CombNetk2-achr1or} (for $i=1$), \textit{one} copy of  \eqref{CombNetk2-achr1or} (for $i=2$), \textit{one}  copy of  \eqref{CombNetk2-achr1or} (for $i=3$), \textit{one}  copy of  \eqref{CombNetk2-ach1or} (for $\Lambda=\{\{1\}\star,\{2,3\}\star\}$, $p=4$),  \textit{one}  copy of  \eqref{CombNetk2-ach1or} (for $\Lambda=\{\{1\}\star,\{2\}\star,\{3\}\star\}$, $p=5$), and finally one copy of the non-negativity constraint in \eqref{CombNetk2-achpos} (for $\mathcal{S}=\{1,2,3\}$) gives rise to $4R_1+2R_2\leq 7$.

We now write the following upper-bounds on $R_1$ and $R_2$ (which we prove in detail in Section \ref{lb-CombNet-subsec-conv}). Notice the similarity of each outer-bound constraint in \eqref{Combnet-lb-ex-conv1}-\eqref{Combnet-lb-ex-conv6} to an inner-bound constraint that played a role in the derivation of $4R_1+2R_2\leq 7$.
\begin{align}
&R_1+\frac{1}{n}H(X^n_{\{\{1\}\star\}}|W_1)\leq 1\label{Combnet-lb-ex-conv1}\\
&R_1+\frac{1}{n}H(X^n_{\{\{2\}\star\}}|W_1)\leq 2\label{Combnet-lb-ex-conv2}\\
&R_1+\frac{1}{n}H(X^n_{\{\{3\}\star\}}|W_1)\leq 2\\
&R_2 \leq \frac{1}{n}H(X^n_{\{\{1\}\star,\{2,3\}\star\}}|W_1)+1\\
&R_2\leq \frac{1}{n}H(X^n_{\{\{1\}\star,\{2\}\star,\{3\}\star\}}|W_1)\\
&0\leq \frac{1}{n}H(X^n_{\{1,2,3\}}|W_1).\label{Combnet-lb-ex-conv6}
\end{align}
Take two copies of \eqref{Combnet-lb-ex-conv1} and one copy each of \eqref{Combnet-lb-ex-conv2}-\eqref{Combnet-lb-ex-conv6} to yield an outer-bound inequality of the following form. 
\begin{align}
&4R_1+2R_2\nonumber\\\quad\leq& 7-\frac{1}{n}\left(\begin{array}{l}2H(X^n_{\{1\}}|W_1)+H(X^n_{\{2\}}X^n_{\{2,3\}}|W_1)+H(X^n_{\{2,3\}}X^n_{\{3\}}|W_1)\\-H(X^n_{\{1\}}X^n_{\{2,3\}}|W_1)-H(X^n_{\{1\}}X^n_{\{3\}}X^n_{\{2,3\}}X^n_{\{2\}}|W_1)\end{array}\right)\\
\stackrel{(a)}{\leq}& 7\label{lb-Combnet-ex-conv-stepa}
\end{align}
where $(a)$ holds by sub-modularity of entropy.
\begin{figure}
\centering
\begin{tikzpicture}[scale=2]
\tikzstyle{every node}=[draw,shape=circle]; 
\path (2.75,2.5cm) node (v0) {$S$}; 


\path (.75,1.25) node (v2) {};
\path (1.75,1.25) node (v3) {};
\path (2.75,1.25cm)  node (v5) {};
\path (3.75,1.25) node (v1) {};
\path (4.75,1.25)  node (v4) {};

\path (0.75,0) node (d1) {$D_1$};
\path (1.75,0) node (d2) {$D_2$};
\path (2.75,0) node  (d3) {$D_3$};
\path (3.75,0) node (d4)[shade]  {$D_4$};

\path (4.75,0) node (d5)[shade]  {$D_5$};

\draw[->] (v0) --node[right,draw=none,xshift=0cm]{$\e_{\{1\}}$} (v1);
\draw[->] (v0) --node[left,draw=none,xshift=-.2cm]{$\e_{\{2,3\}}$} (v2);
\draw[->] (v0) --node[left,draw=none,xshift=0cm]{$\e_{\{2\}}$} (v3);
\draw[->] (v0) --node[right,draw=none,xshift=.1cm]{$\e_{\{3\}}$} (v4);
\draw[->] (v0) --node[left,draw=none,xshift=.1cm]{$\e_{\phi}$} (v5);

\draw[->] (v1) -- (d1);
\draw[->] (v1) -- (d4);
\draw[->] (v1) -- (d5);
\draw[->] (v5) -- (d4);
\draw[->] (v2) -- (d2);
\draw[->] (v2) -- (d3);
\draw[->] (v2) -- (d4);
\draw[->] (v3) -- (d2);
\draw[->] (v3) -- (d5);
\draw[->] (v4) -- (d3);
\draw[->] (v4) -- (d5);

\end{tikzpicture}
\caption{Is rate pair $(1,2)$ achievable over this combination network?}
\label{example-converse}
\end{figure}
\end{example}

The intuition from Example \ref{CombNet-example-conv2} gives us a method to prove the converse of Theorem \ref{lb-CombNet-Theoremk2} (for $m=3$).

\subsubsection{Outer bound}
\label{lb-CombNet-subsec-conv}
{The starting point in proving the converse is the following lemma which we only state here and prove in Appendix~\ref{ap-CombNetk2-relaxregion}.}

{\begin{lemma}
\label{lb-CombNetk2-relaxregion}
Consider the rate-region characterization in Theorem \ref{lb-CombNet-Theoremk2}  
(where $\mathcal{I}_1=\{1,2,3\}$ and $\mathcal{I}_2=\{4,\ldots, K\}$). The constraints given by inequality 
\eqref{CombNetk2-achpos} in Theorem \ref{lb-CombNet-Theoremk2} can be replaced by  \eqref{lb-CombNet-relax}, below, without affecting the rate-region.
\begin{align}
\label{lb-CombNet-relax}
&\sum_{\mathcal{S}\in \Lambda}\alpha_{\mathcal{S}}\geq 0,\quad \forall \Lambda\subseteq 2^{\mathcal{I}_1}\ \text{superset saturated}.
\end{align}
\end{lemma}
By Lemma \ref{lb-CombNetk2-relaxregion}, the rate-region of Theorem \ref{lb-CombNet-Theoremk2} is equivalently given by constraints
 \eqref{CombNetk2-achR2or}-\eqref{CombNetk2-achR1+R2or}, \eqref{lb-CombNet-relax}. We now find an outer-bound which looks \textit{similar} to this inner-bound.}

{\begin{lemma}
Any achievable rate pair $(R_1,R_2)$ satisfies outer-bound constraints \eqref{CombNetk2-converse-1}-\eqref{CombNetk2-converse-4} for any given $\epsilon>0$.
\label{CombNet-opt-similar}
\begin{align}
&\frac{1}{n}H( X^n_\Lambda|W_1)\geq 0,\quad\quad{\forall \Lambda\subseteq 2^{\mathcal{I}_1}\text{ superset saturated}}\label{CombNetk2-converse-1}\\
&R_1+\frac{1}{n}H(  X^n_{\{\{i\}\star\}}|W_1) \leq\sum_{\mathcal{S}\in\{\{i\}\star\}}\card{\e_\mathcal{S}} + \epsilon,\quad\quad{\forall i\in \mathcal{I}_1}\label{CombNetk2-converse-2}\\
&R_2\leq\frac{1}{n}H( X^n_{\Lambda}|W_1) +  \sum_{\mathcal{S}\in\Lambda^c}   \card{\e_{\mathcal{S},p}} + \epsilon,\quad\quad \forall \Lambda\subseteq 2^{\mathcal{I}_1}\text{ superset saturated},\ \forall p\in \mathcal{I}_2\label{CombNetk2-converse-3}\\
&R_1+R_2\leq\sum_{\mathcal{S}\subseteq \mathcal{I}_1}\card{\e_{\mathcal{S},p}}+\epsilon,\quad\quad{\forall p\in \mathcal{I}_2}.\label{CombNetk2-converse-4}
\end{align} 
\end{lemma}
\begin{remark}
Notice the similarity of inequalities \eqref{CombNetk2-converse-1}, \eqref{CombNetk2-converse-2}, \eqref{CombNetk2-converse-3}, \eqref{CombNetk2-converse-4} with constraints \eqref{lb-CombNet-relax}, \eqref{CombNetk2-achr1or}, \eqref{CombNetk2-ach1or}, \eqref{CombNetk2-achR1+R2or}, respectively. We provide no similar outer-bound for the inner-bound constraint \eqref{CombNetk2-achR2or} because it is redundant\footnote{This inequality is redundant because it is the only inequality that contains the free variable $\alpha_\phi$.}.
\end{remark}
\begin{proof}
Inequalities in \eqref{CombNetk2-converse-1} hold by the positivity of entropy. To show inequalities in \eqref{CombNetk2-converse-2}, we bound $R_1$ for each public receiver $ i\in \mathcal{I}_1$ as follows:
\begin{align}
nR_1=& H(W_1)\\
=&H(W_1)- H(W_1|{Y}^n_i)+ H(W_1|{Y}^n_i)\\
\stackrel{(a)}{\leq}&I(W_1;{Y}^n_i)+n\epsilon\\
=&I(W_1;X^n_{\{\{i\}\star\}})+n\epsilon\\
=&H(X^n_{\{\{i\}\star\}})-H(X^n_{\{\{i\}\star\}}|W_1)+n\epsilon\\
\stackrel{(b)}{\leq}&n\left(\sum_{\mathcal{S}\in{\{\{i\}\star\}}}\card{\e_\mathcal{S}}\right)-H(X^n_{\{\{i\}\star\}}|W_1)+n\epsilon.\label{lb-CombNet-opt-ref-R1}
\end{align}
In the above chain of inequalities, $(a)$ follows from Fano's inequality and $(b)$ follows by bounding the cardinality of the alphabet set of $X^n_{\{\{i\}\star\}}$ and using $H(X)\leq\log|\mathcal{X}|$, where $\mathcal{X}$ is the alphabet set of $X$.
In a similar manner,  we have the following bound on $nR_1+nR_2$ for each private receiver $p$ which proves inequality \eqref{CombNetk2-converse-4}. 
\begin{align}
nR_1+nR_2=& H(W_1W_2)\\
\leq&I(W_1W_2;Y^n_p)+n\epsilon\\
=&I(W_1W_2;X^n_{\{\phi\star\},p})+n\epsilon\\
=&H(X^n_{\{\phi\star\},p})+n\epsilon\\
{\leq}&n\left(\sum_{\mathcal{S}\in{\{\phi\star\}}}\card{\e_{\mathcal{S},p}}\right)+n\epsilon.\label{lb-CombNet-opt-ref-R1+R2}
\end{align}
Finally, we bound $R_2$ to obtain the inequalities in \eqref{CombNetk2-converse-3}. In the following, we have $p\in \mathcal{I}_2$, $\Lambda\subseteq 2^{\mathcal{I}_1}$, and $\epsilon>0$.
\begin{align}
nR_2=& H(W_2|W_1)\\
=&H(W_2|W_1)- H(W_2|W_1Y^n_p)+ H(W_2|W_1Y^n_p)\\
\stackrel{(a)}{\leq}&I(W_2;Y^n_p|W_1)+n\epsilon\\
=&I(W_2;X^n_{\{\phi\star\},p}|W_1)+n\epsilon\\
=&H(X^n_{\{\phi\star\},p}|W_1)+n\epsilon\label{lb-outer-R2}\\
\stackrel{(b)}{\leq}&H(X^n_{\{\phi\star\},p}X^n_\Lambda|W_1)+n\epsilon\\
=&H(X^n_\Lambda|W_1)+H(X^n_{\{\phi\star\},p}|X^n_\Lambda W_1)+n\epsilon\\
\stackrel{(c)}{\leq}&H(X^n_\Lambda|W_1)+H(X^n_{\Lambda^c,p})+n\epsilon\\
\stackrel{(d)}{\leq}&H(X^n_\Lambda|W_1)+n\left(\sum_{\mathcal{S}\in\Lambda^c}\card{\e_{\mathcal{S},p}}\right)+n\epsilon.\label{lb-CombNet-opt-ref-R2}
\end{align}
In the above chain of inequalities, step $(a)$ follows from Fano's inequality. Step $(b)$ holds for any subset $\Lambda\subseteq 2^{\mathcal{I}_1}$ and in particular subsets which are superset saturated. Step $(c)$ follows because  conditioning decreases entropy. Step $(d)$ follows by by bounding the cardinality of the alphabet set of $X^n_{\Lambda^c,p}$ and using $H(X)\leq\log|\mathcal{X}|$, where $\mathcal{X}$ is the alphabet set of $X$.
\end{proof}}

{We shall use sub-modularity of the entropy function to prove that the outer bound of Lemma \ref{CombNet-opt-similar} and the inner bound of Theorem \ref{lb-CombNet-Theoremk2} match. Let us introduce a few techniques, as it may not be clear how sub-modularity could be used in full generality.}

\subsubsection{Sub-modularity lemmas}

We adopt some definitions and results from \cite{BalisterBollobas07} and prove a lemma that takes a central role in the converse proof in Section \ref{lb-CombNet-subsec-proofthm}.

{Let $[\mathbf{F}]$ be a family of multi-sets\footnote{Multi-set is a generalization of the notion of a set in which members are allowed to appear more than once.} of subsets of  $\{s_1,\ldots,s_N\}$. Given a multi-set $\mathbf{\Gamma}=[\Gamma_1,\ldots,\Gamma_l]$ (where $\Gamma_i\subseteq\{s_1,\ldots,s_N\}$, $i=1,\ldots,l$), let $\mathbf{\Gamma}^\prime$ be a multi-set obtained from $\mathbf{\Gamma}$ by replacing $\Gamma_i$ and $\Gamma_j$ by $\Gamma_i\cap\Gamma_j$ and $\Gamma_i\cup\Gamma_j$ for some $i,j\in\{1,\ldots,l\}$, $i\neq j$.  The multi-set $\mathbf{\Gamma}^\prime$ is then said to be an \textit {elementary compression} of $\mathbf{\Gamma}$. The elementary compression is, in particular, \textit{non-trivial} if neither  $\Gamma_i\subseteq \Gamma_j$ nor $\Gamma_j\subseteq \Gamma_i$. 
A sequence of elementary compressions gives a \textit{compression}. A partial order~$\geq$ is defined over $[\mathbf{F}]$ as follows. $\mathbf{\Gamma}\geq\mathbf{\Lambda}$ if $\mathbf{\Lambda}$ is a compression of $\mathbf{\Gamma}$ (equality if and only if the compression is composed of all trivial elementary compressions). 
A simple consequence of the sub-modularity of the entropy function is the following lemma \cite[Theorem 5]{BalisterBollobas07}.
\begin{lemma}{\cite[Theorem 5]{BalisterBollobas07}}
\label{otherslemma}
Let ${X}=(X_{s_i})_{1}^N$ be a sequence of random variables with $H({X})$ finite and let $\mathbf{\Gamma}$ and $\mathbf{\Lambda}$ be finite multi-sets of subsets of $\{s_1,\ldots,s_N\}$ such that $\mathbf{\Gamma}\geq\mathbf{\Lambda}$. Then
$$\sum_{\Gamma\in\mathbf{\Gamma}}H(X_\Gamma)\geq \sum_{\Lambda\in\mathbf{\Lambda}}H(X_\Lambda).$$ 
 \end{lemma}

For our converse, we consider the family of multi-sets of subsets of $2^{\mathcal{I}_1}$ where $\mathcal{I}_1=\{1,2,3\}$. 
We denote multi-sets by bold greek capital  letters (e.g., $\mathbf{\Gamma}$ and $\mathbf{\Lambda}$), subsets of $2^{\mathcal{I}_1}$ by   greek capital letters (e.g., $\Gamma_i$, ${\Sigma}$ and ${\Lambda}$), and elements of $2^{\mathcal{I}_1}$ by calligraphic capital letters (e.g., $\mathcal{S}$ and $\mathcal{T}$).
Over a family of multi-sets of subsets of $2^{\mathcal{I}_1}$, we define  \textit{multi-sets of saturated pattern} and \textit{multi-sets of standard pattern} as follows.}
\begin{definition}[Multi-sets of saturated pattern]
A multi-set (of subsets of $2^{\mathcal{I}_1}$) is said to be of \textit{(superset) saturated pattern} if all its elements are superset saturated.  E.g., for $\mathcal{I}_1=\{1,2,3\}$, we have that multi-sets $[ \{ \{1\},\{1,2\},\{1,3\},\{1,2,3\} \} ]$ and $[\{\{1\},
\{1,2\},\{1,3\},\{1,2,3\}\},\{\{2,3\},\{1,2,3\}\}]$ are both of saturated pattern, but not $[\{\{2\},\{1,2\},\{1,2,3\}\}]$ (since its only element is not superset saturated as $\{2,3\}$ is missing from it) or $[\{\{1\},\{2\},\{1,2\},\{1,3\},\{2,3\},\{1,2,3\}\},\{\{1\},\{1,2\}\}]$ (since $\{\{1\},\{1,2\}\}$ is not superset saturated).
\end{definition}

\begin{definition}[Multi-sets of standard pattern]
A multi-set (of subsets of $2^{\mathcal{I}_1}$) is said to be of \textit{standard pattern} if its elements are all of the form $\{\mathcal{S}\subseteq \mathcal{I}_1: \mathcal{S}\ni i\}$, for some $i\in \mathcal{I}_1$. E.g., both multi-sets $[ \{ \{1\},\{1,2\},\{1,3\},\{1,2,3\} \} ]$ and $[\{\{1\},
\{1,2\},\{1,3\},\{1,2,3\}\},\{\{2\},\{1,2\},\{2,3\},\{1,2,3\}\}]$ are of standard pattern, but not multi-sets $[\{\{1,2\},\{1,2,3\}\}]$ or $[\{\{1\},\{2\},\{1,2\},\{1,3\},\{2,3\},\{1,2,3\}\}]$.
\end{definition}
\begin{definition}[Balanced pairs of multi-sets]
We say that multi-sets $\mathbf{\Gamma}$ and $\mathbf{\Lambda}$ are a \textit{balanced} pair if
$\sum_{\Gamma\in \mathbf{\Gamma}}\mathbf{1}_{\mathcal{S}\in\Gamma}=\sum_{\Lambda\in \mathbf{\Lambda}}\mathbf{1}_{\mathcal{S}\in\Lambda}$, for all sets $\mathcal{S}\in 2^{\mathcal{I}_1}$.
\end{definition}

\begin{remark}
\label{impremark}
One observes that 
(i) multi-sets of standard pattern are also of saturated pattern, 
(ii) the set of all multi-sets of saturated pattern is closed under compression, 
(iii) if a multi-set $\mathbf{\Lambda}$ is a compression of a multi-set $\mathbf{\Gamma}$, then they are balanced, and
(iv) two multi-sets of standard pattern are balanced if and only if they are equal.
\end{remark}
Let us look at step $(a)$ in inequality \eqref{lb-Combnet-ex-conv-stepa} in this formulation. Consider the family of multi-sets of subsets of $2^{\{1,2,3\}}$, and in particular, the multi-set $\mathbf{\Gamma}=[\{\{1\}\},\{\{1\}\},\{\{2\},\{2,3\}\},\{\{3\},\{2,3\}\}]$. After the following non-trivial elementary compressions, the multi-set $\mathbf{\Lambda}=[\{\{1\},\{2,3\}\},\{\{1\},\{2\},\{3\},\{2,3\}\}]$ is obtained:
\begin{align}
\label{lb-CombNet-conv-comp}
\mathbf{\Gamma}=&[\{\{1\}\},\{\{1\}\},\{\{2\},\{2,3\}\},\{\{3\},\{2,3\}\}]\\
\geq& [\{\{1\}\},\{\{1\},\{2\},\{2,3\}\},\{\{3\},\{2,3\}\}]=:\mathbf{\Gamma}^{\prime}\\ 
\geq& [\{\{1\}\},\{\{1\},\{2\},\{3\},\{2,3\}\},\{\{2,3\}\}]=:\mathbf{\Gamma}^{\prime\prime}\\ 
\geq& [\{\{1\},\{2,3\}\},\{\{1\},\{2\},\{3\},\{2,3\}\}]=:\mathbf{\Lambda}.\label{lb-CombNet-conv-comp2}
\end{align}
Therefore, we have $\mathbf{\Gamma}\geq \mathbf{\Lambda}$ and by Lemma \ref{otherslemma}, $\sum_{\Gamma\in\mathbf{\Gamma}}H(X^n_\Gamma|W_1)\geq \sum_{\Lambda\in\mathbf{\Lambda}}H(X^n_\Lambda|W_1)$. Thus, step $(a)$ of inequality \eqref{lb-Combnet-ex-conv-stepa} follows.

Here, we develop an alternative visual tool. Associate a graph $G_{\mathbf{\Gamma}}$ to every multi-set $\mathbf{\Gamma}$. The graph is formed as follows. Each node of the graph represents one set in  $\mathbf{\Gamma}$, and is labeled by it. Two nodes are connected by an edge if and only if neither is a subset of the other. Each time an elementary compression is performed on the multi-set $\mathbf{\Gamma}$, a compressed multi-set $\mathbf{\Gamma}^\prime$ (with a new graph associated with it) is created.
E.g., graphs associated with multi-sets $\mathbf{\Gamma}$, $\mathbf{\Gamma}^\prime$, $\mathbf{\Gamma}^{\prime\prime}$, and $\mathbf{\Lambda}$ (which are all defined in inequalities \eqref{lb-CombNet-conv-comp}-\eqref{lb-CombNet-conv-comp2}) are shown in Fig.~\ref{CombNet-compression}.

\begin{figure}
 \centering
 \begin{subfigure}[b]{0.24\textwidth}
 \begin{tikzpicture}[scale=1.3]
\tikzstyle{every node}=[draw,font=\footnotesize, shape=circle,minimum size=.005cm]; 

\path  (180+60:1cm) node[fill](v1) {};
\path  (180+40:.98cm) node[rectangle,draw=none]() { $\{\{1\}\}$};

\path  (0-60:1cm) node(v11)[fill] {};
\path  (0-40:.98cm) node[rectangle,draw=none]() { $\{\{1\}\}$};
\path  (270-15:1.9cm) node(v2)[fill] {}; 
\path  (270-15:2.2cm) node[rectangle,draw=none]() {$\{\{2\},\{2,3\}\}\ $};

\path  (270+15:1.9cm) node(v3)[fill] {};
\path  (270+15:2.2cm) node[rectangle,draw=none]() {$\ \{\{3\},\{2,3\}\}$};

\draw (v1) -- (v2);
\draw (v11) -- (v2);
\draw (v1) -- (v3);
\draw (v11) -- (v3);
\draw (v2) -- (v3);

\end{tikzpicture}
 \subcaption{Graph $G_\mathbf{\Gamma}$}
 \end{subfigure}
 \begin{subfigure}[b]{0.24\textwidth}
\begin{tikzpicture}[scale=1.3]
\tikzstyle{every node}=[draw,font=\footnotesize, shape=circle,minimum size=.005cm]; 

\path  (180+60:1cm) node(v1)[fill] {};
\path  (180+40:.98cm) node[rectangle,draw=none]() { $\{\{1\}\}$};

\path  (0-60:1cm) node(v11)[fill] {};
\path  (0-45:.9cm) node[rectangle,draw=none]() { $\{\{1\},\{2\},\{2,3\}\}$};

\path  (270-15:1.9cm) node(v2)[fill] {}; 
\path  (270-15:2.2cm) node[rectangle,draw=none]() {$\phi$};

\path  (270+15:1.9cm) node(v3)[fill] {};
\path  (270+15:2.2cm) node[rectangle,draw=none]() {$\{\{3\},\{2,3\}\}$};
\path  (-135:2cm) node()[draw=none] {$>$};

\draw (v1) -- (v3);
\draw (v11) -- (v3);

\end{tikzpicture}
 \subcaption{Graph $G_{\mathbf{\Gamma}^\prime}$}
 \end{subfigure}
 \begin{subfigure}[b]{0.24\textwidth}
\begin{tikzpicture}[scale=1.3]
\tikzstyle{every node}=[draw,font=\footnotesize, shape=circle,minimum size=.005cm]; 

\path  (180+60:1cm) node(v1)[fill] {};
\path  (180+43:.89cm) node[rectangle,draw=none]() { $\{\{1\}\}$};

\path  (0-60:1cm) node(v11)[fill] {};
\path  (0-40:.95cm) node[rectangle,draw=none]() { $\{\{1\},\{2\},\{2,3\},\{3\}\}$};

\path  (270-15:1.9cm) node(v2)[fill] {}; 
\path  (270-15:2.2cm) node[rectangle,draw=none]() {$\phi$};

\path  (270+15:1.9cm) node(v3)[fill] {};
\path  (270+15:2.2cm) node[rectangle,draw=none]() {$\{\{2,3\}\}$};

\draw (v1)--(v3);
\path  (-135:2cm) node()[draw=none] {$>$};

\end{tikzpicture}
 \subcaption{Graph $G_{\mathbf{\Gamma}^{\prime\prime}}$}
 \end{subfigure}
 \begin{subfigure}[b]{0.24\textwidth}
\begin{tikzpicture}[scale=1.3]
\tikzstyle{every node}=[draw,font=\footnotesize, shape=circle,minimum size=.005cm]; 

\path  (180+60:1cm) node(v1)[fill] {};
\path  (180+43:.89cm) node[rectangle,draw=none]() { $\hspace{-.4cm}\{\{1\},\{2,3\}\}$};

\path  (0-60:1cm) node(v11)[fill] {};
\path  (0-40:.95cm) node[rectangle,draw=none]() { $\hspace{.4cm}\{\{1\},\{2\},\{2,3\},\{3\}\}$};

\path  (270-15:1.9cm) node(v2)[fill] {}; 
\path  (270-15:2.2cm) node[rectangle,draw=none]() {$\phi$};

\path  (270+15:1.9cm) node(v3)[fill] {};
\path  (270+15:2.2cm) node[rectangle,draw=none]() {$\phi$};
\path  (-135:2cm) node()[draw=none] {$>$};

\end{tikzpicture}
 \subcaption{Graph $G_\mathbf{\Lambda}$}
 \end{subfigure}
 \caption{Graphs associated with multi-sets $\mathbf{\Gamma}$, $\mathbf{\Gamma}^\prime$, $\mathbf{\Gamma}^{\prime\prime}$, $\mathbf{\Lambda}$ obtained through the compression that is performed in inequalities \eqref{lb-CombNet-conv-comp}-\eqref{lb-CombNet-conv-comp2}. }
 \label{CombNet-compression}
\end{figure}

For such (associated) graphs, we prove in Appendix \ref{ap-CombNet-graph-strict} that compression reduces the total number of edges in  graphs.
\begin{lemma}
\label{lb-CombNet-graph-strict}
Let $G_\mathbf{\Gamma}$ denote the graph associated with a multi-set $\mathbf{\Gamma}$ and $G_{\mathbf{\Lambda}}$ denote the graph associated with a multi-set $\mathbf{\Lambda}$. Provided that $\mathbf{\Lambda}<\mathbf{\Gamma}$, the total number of edges in $G_{\mathbf{\Lambda}}$ is strictly smaller than that of $G_{\mathbf{\Gamma}}$.
\end{lemma}

Define a \textit{(non-trivial) decompression} as the inverse act of a (non-trivial) compression. As opposed to compression, a non-trivial decompression is not always possible using every two elements of a multi-set $\mathbf{\Lambda}$. It is, indeed, not clear whether a multi-set $\mathbf{\Lambda}$ is decompressible at all.
For example, the multi-set $[\{\{2,3\},\{1,2,3\}\},\{\{1\},\{2\},\{1,2\},\{1,3\},$ $\{2,3\},\{1,2,3\}\}]$ cannot be non-trivially decompressed; i.e., there exists no multi-set $\mathbf{\Gamma}$ such that \begin{align}
\mathbf{\Gamma}>[\{\{2,3\},\{1,2,3\}\},\{\{1\},\{2\},\{1,2\},\{1,3\},\{2,3\},\{1,2,3\}\}].
\end{align}
The table in Fig.~\ref{lb-CombNet-opt-table-decompress} gives a list of some non-trivial elementary decompressions for multi-sets of subsets of $2^{\{1,2,3\}}$.
\begin{figure}\centering\begin{tabular}{|l|l|}
   \multicolumn{1}{l}{Multi-set $\mathbf{\Lambda}$}
 & \multicolumn{1}{l}{Multi-set $\mathbf{\Sigma}>\mathbf{\Lambda}$} \\
\cline{1-2}
 $\left[\;\ldots,\;\left\{\{i,j\}\star\right\},\; \left\{\{i\}\star,\{j\}\star\right\},\;\ldots\;\right]$$\substack{\ \\\\\\\ }$ & $\left[\;\ldots,\;\left\{\{i\}\star\right\},\;\left\{\{j\}\star\right\},\;\ldots\;\right]$\\
 \cline{1-2}
 $\left[\;\ldots,\;\left\{\{1,2,3\}\right\},\;\left\{\{i\}\star,\{j,k\}\star\right\},\;\ldots\;\right]$$\substack{\ \\\\\\\ }$ & $\left[\;\ldots,\;\left\{\{i\}\star\right\},\;\left\{\{j,k\}\star\right\},\;\ldots\;\right]$\\
\cline{1-2}
 $\left[\;\ldots,\;\left\{\{1,2,3\}\right\},\;\left\{\{i,j\}\star,\{i,k\}\star\right\},\;\ldots\;\right]$$\substack{\ \\\\\\\ }$ & $\left[\;\ldots,\;\left\{\{i,j\}\star\right\},\;\left\{\{i,k\}\star\right\},\;\ldots\;\right]$\\
\cline{1-2}
 $\left[\;\ldots,\;\left\{\{1,2,3\}\right\},\;\left\{\{i,j\}\star,\{i,k\}\star,\{j,k\}\star\right\},\;\ldots\;\right]$ $\substack{\ \\\\\\\ }$& $\left[\;\ldots,\;\left\{\{i,j\}\star\right\},\;\left\{\{i,k\}\star,\{j,k\}\star\right\},\;\ldots\;\right]$\\
\cline{1-2}
 $\left[\;\ldots,\;\left\{\{1\}\star,\{2\}\star,\{3\}\star\right\},\;\left\{\{i\}\star,\{j,k\}\star\right\},\;\ldots\;\right]$ $\substack{\ \\\\\\\ }$& $\left[\;\ldots,\;\left\{\{i\}\star,\{j\}\star\right\},\;\left\{\{i\}\star,\{k\}\star\right\},\;\ldots\;\right]$\\
\cline{1-2}
 $\left[\;\ldots,\;\left\{\{1\}\star,\{2\}\star,\{3\}\star\right\},\;\left\{\{i,j\}\star,\{i,k\}\star\right\},\;\ldots\;\right]$$\substack{\ \\\\\\\ }$ & $\left[\;\ldots,\;\left\{\{i\}\star\right\},\;\left\{\{j\}\star,\{k\}\star\right\},\;\ldots\;\right]$\\
\cline{1-2}
$\left[\;\ldots,\;\left\{\{1\}\star,\{2\}\star,\{3\}\star\right\},\;\left\{\{i,j\}\star,\{i,k\}\star,\{j,k\}\star\right\},\;\ldots\;\right]$ $\substack{\ \\\\\\\ }$& $\left[\;\ldots,\;\left\{\{i\}\star,\{j\}\star\right\},\;\left\{\{k\}\star,\{i,j\}\star\right\},\;\ldots\;\right]$\\
\cline{1-2}
\end{tabular}
\caption{Non-trivial elementary decompressions for multi-sets of subsets of $2^{\{1,2,3\}}$. Here, $(i,j,k)$ is a permutation of $(1,2,3)$.}
\label{lb-CombNet-opt-table-decompress}
\end{figure}
Although not all multi-sets are decompressible, Lemma \ref{CombNet-count} below identifies a class of multi-sets of subsets of $2^{\{1,2,3\}}$ that are decompressible.

\begin{lemma}
\label{CombNet-count}
Let $\mathbf{\Lambda}$ and $\mathbf{\Gamma}$ be multi-sets of subsets of $2^{\{1,2,3\}}$. Suppose $\mathbf{\Lambda}$ is of saturated pattern and $\mathbf{\Gamma}$ is of standard pattern.
If  $\mathbf{\Lambda}$ and $\mathbf{\Gamma}$ are a pair of balanced multi-sets, then a non-trivial elementary 
decompression could be performed over $\mathbf{\Lambda}$, unless $\mathbf{\Lambda}=\mathbf{\Gamma}$.
\end{lemma}
\begin{proof}
The proof is by showing that for any multi-set $\mathbf{\Lambda}$ with the stated assumptions, at least one of the non-trivial elementary decompressions in Fig.~\ref{lb-CombNet-opt-table-decompress} is doable.
This is done by double counting (once in $\mathbf{\Lambda}$ and once in $\mathbf{\Gamma}$) the number of times each subset $\mathcal{S}\in 2^{\{1,2,3\}}$ appears in the multi-set $\mathbf{\Lambda}$, and showing that no matter what $\mathbf{\Lambda}$ and $\mathbf{\Gamma}$ are, at least one of the cases of Fig.~\ref{lb-CombNet-opt-table-decompress} occurs. We defer details of this proof to Appendix \ref{ap-CombNet-count}.
\end{proof}

Lemma \ref{CombNet-count} shows that a multi-set $\mathbf{\Lambda}$ of saturated pattern, which is balanced with a multi-set 
$\mathbf{\Gamma}$ of standard pattern, can be non-trivially decompressed. Let the result of this non-trivial elementary 
decompression be a multi-set $\mathbf{\Sigma}$.
Since the decompressed multi-set $\mathbf{\Sigma}$ is, itself, of saturated pattern and remains balanced with multi-set 
$\mathbf{\Gamma}$ (see Remark \ref{impremark}), one can continue decompressing it using Lemma \ref{CombNet-count} as long as 
$\mathbf{\Sigma}\neq\mathbf{\Gamma}$.
 This, either ends in an infinite loop, or ends in $\mathbf{\Sigma}=\mathbf{\Gamma}$ (Note that there cannot be two different 
 multi-sets $\mathbf{\Gamma}\geq\mathbf{\Lambda}$ and $\mathbf{\Gamma}^\prime\geq \mathbf{\Lambda}$ such that $\mathbf{\Gamma}$, $\mathbf{\Gamma}^\prime$, and $\mathbf{\Lambda}$ are all balanced);  the former is ensured not to happen, for the total number of edges in the associated graph strictly decreases after each decompression (see Lemma \ref{lb-CombNet-graph-strict}). Thus, we arrive at the following lemma.

\begin{lemma}
\label{lb-CombNet-decompression}
Let $\mathbf{\Lambda}$ and $\mathbf{\Gamma}$ be multi-sets of subsets of $2^{\{1,2,3\}}$ where $\mathbf{\Lambda}$ is of saturated pattern and $\mathbf{\Gamma}$ is of standard pattern. If  $\mathbf{\Lambda}$ and $\mathbf{\Gamma}$ are balanced, then $\mathbf{\Lambda}$ can be decompressed to $\mathbf{\Gamma}$; i.e., $\mathbf{\Gamma}\geq \mathbf{\Lambda}$.
\end{lemma}

We are now ready to prove Theorem \ref{lb-CombNet-Theorem-outer2}.

\subsubsection{proof of Theorem \ref{lb-CombNet-Theorem-outer2}}
\label{lb-CombNet-subsec-proofthm}
The rate-region of Theorem \ref{lb-CombNet-Theoremk2} can be obtained explicitly by applying Fourier-Motzkin elimination  to \eqref{CombNetk2-achr1or}-\eqref{CombNetk2-achR1+R2or} and \eqref{lb-CombNet-relax} to eliminate parameters $\alpha_\mathcal{S}$. This gives a set of inequalities of the form $m_1R_1+m_2R_2\leq E$,
each obtained by summing potentially multiple copies of constraints \eqref{CombNetk2-achr1or}-\eqref{CombNetk2-achR1+R2or}, \eqref{lb-CombNet-relax}, so that all variables $\alpha_{\mathcal{S}}$, $\mathcal{S}\subseteq \mathcal{I}_1$, get eliminated. 
To show a converse for each such inner-bound inequality, $m_1R_1+m_2R_2\leq E$, 
take copies of the corresponding outer-bound constraints \eqref{CombNetk2-converse-1}-\eqref{CombNetk2-converse-4} and sum them up to yield an outer-bound inequality of the form
\begin{align}
&m_1R_1+m_2R_2+\frac{1}{n}\sum_{\Gamma\in\mathbf{\Gamma}}H(X^n_{\Gamma}|W_1)\nonumber\\
&\leq E +\frac{1}{n}\sum_{\Lambda\in\mathbf{\Lambda}} H(X^n_{\Lambda}|W_1)\label{CombNet2-subm}
\end{align}
where $\mathbf{\Gamma}$ is a multi-set of standard pattern and $\mathbf{\Lambda}$ is a multi-set of saturated pattern, both consisting of subsets of $2^{\mathcal{I}_1}$ where $\mathcal{I}_1=\{1,2,3\}$. Notice that $\mathbf{\Gamma}$ and $\mathbf{\Lambda}$ are balanced since Fourier-Motzkin elimination ensures that all the $\alpha_\mathcal{S}$'s are eliminated.
So by Lemma \ref{lb-CombNet-decompression},  we have
\begin{align}\mathbf{\Lambda}\leq \mathbf{\Gamma}\end{align} and therefore,
\begin{align}
\label{entropysubmod}
 \sum_{\Lambda\in\mathbf{\Lambda}} H(X^n_{\Lambda}|W_1)\leq\sum_{\Gamma\in\mathbf{\Gamma}}H(X^n_{\Gamma}|W_1).
\end{align} 
Using  \eqref{entropysubmod} in the outer-bound inequality \eqref{CombNet2-subm}, we conclude the converse to $m_1R_1+m_2R_2\leq E$, for every such inequality that appears in the inner-bound. This concludes the proof.

\begin{remark}
Lemmas \ref{lb-CombNetk2-relaxregion}, \ref{CombNet-count}, and \ref{lb-CombNet-decompression} are valid only for $m\leq3$, and Lemmas \ref{otherslemma}, \ref{CombNet-opt-similar}, and \ref{lb-CombNet-graph-strict} hold in general.
\end{remark}
\begin{remark}
The following example serves as a counter example for Lemma \ref{CombNet-count} when $m>3$.
Consider the following multi-sets of subsets of $2^{\{1,2,3,4\}}$:
\begin{eqnarray}
&&\mathbf{\Lambda}=\left[\{\{1\}\star\{2\}\star\{3\}\star\{4\}\star\},\{\{1,2\}\star\{1,3\}\star\{2,4\}\},\{\{1,4\}\star\{2,3\}\star\{3,4\}\star\},\{\{1,2,3,4\}\}\right],\nonumber\\
&&\mathbf{\Gamma}=[\{\{1\}\star\},\{\{2\}\star\},\{\{3\}\star\},\{\{4\}\star\}].\nonumber
\end{eqnarray}
It is easy to see that $\mathbf{\Lambda}$ is of saturated pattern, $\mathbf{\Gamma}$ is of standard pattern, and they are balanced. Nevertheless, no elementary decompression can be performed on $\mathbf{\Lambda}$.
\end{remark}

One could prove Theorem \ref{lb-CombNet-Theorem-outer3} using the same technique. More precisely, Lemma \ref{lb-CombNetk2-relaxregion} is already implied and Lemma \ref{CombNet-opt-similar} could accommodate an outer-bound constraint \textit{similar} to the inner-bound constraint of \eqref{CombNetbme-subr1or} as follows:
\begin{eqnarray}
\frac{1}{n}H({X}^n_{\{\{i\}\star\}}|W_1)&\leq& \frac{1}{n}H({X}^n_{\Lambda}|W_1)+\frac{1}{n}H(X^n_\Lambda|X^n_{\{\{i\}\star\}},W_1)\\
&\leq&\frac{1}{n}H(X^n_{\Lambda}|W_1)+\frac{1}{n}H(X^n_{(\{\{i\}\star\}\backslash\Lambda)})\\
&\leq&\frac{1}{n}H(X^n_{\Lambda}|W_1)+\sum_{\substack{\mathcal{S}\in\Lambda^c\\\mathcal{S}\ni i}}\card{\e_{\mathcal{S}}}.
\end{eqnarray}
The rest of the converse proof follows, as before, by sub-modularity of entropy.

\section{A block Markov encoding scheme for broadcasting two nested message sets over broadcast channels}
\label{lb-sec-generalbc}
We extend our coding technique to general broadcast channels. 
Consider a broadcast channel $p(y_1,\ldots,y_K|x)$ with input signal $X$, outputs $Y_1,\ldots,Y_K$ where $Y_i$, $i\in \mathcal{I}_1$, is  available to public receiver $i$ and $Y_p$, $p\in \mathcal{I}_2$, is available to  private receiver~$p$. 

In many cases where the optimal rates of communication are known for broadcasting nested messages, the classical techniques of rate splitting and superposition coding have proved optimal, and this motivates us, also, to start with such encoding schemes.
In particular, in the context of two message broadcast, we split the private message into different pieces $W^\mathcal{S}_2$ of rates $\alpha_\mathcal{S}$, $\mathcal{S}\subseteq \mathcal{I}_1$, where $W^\mathcal{S}_2$ is revealed to all public receivers in $\mathcal{S}$ (as well as the private receivers). $X^n$ is then formed by superposition coding. For $\mathcal{I}_1=\{1,2\}$, for instance, $W^{\{1\}}_2$ and $W^{\{2\}}_2$ are each independently superposed on $(W_1,W^{\{1,2\}}_2)$, and $W^{\phi}_2$ is superposed on all of them to form the input sequence$X^n$.
The rate-region achievable by superposition coding is given by a feasibility problem (a straightforward generalization of \cite[Theorem 8.1]{ElGamalKim}). 
The rate pair $(R_1,R_2)$ is achievable if there exist parameters $\alpha_\mathcal{S}$, $\mathcal{S}\subseteq \mathcal{I}_1$, and auxiliary random variables $U_\mathcal{S}$, $\phi\neq\mathcal{S}\subseteq \mathcal{I}_1$, such that inequalities in \eqref{Generalbme-posor}-\eqref{General-R1+R2or} hold for a  probability mass function (pmf) $\prod_{k=1}^K\prod_{\substack{\mathcal{S}\subseteq \mathcal{I}_1\\|\mathcal{S}|=k}}p(u_{\mathcal{S}}|\{u_{\mathcal{T}}\}_{\substack{\mathcal{T}\in\{\mathcal{S}\star\}\\\mathcal{T}\neq\mathcal{S}}})p(x|\{u_{\mathcal{S}}\}_{\substack{\mathcal{S}\subseteq \mathcal{I}_1\\\mathcal{S}\neq \phi}})$.
\allowdisplaybreaks
\begin{align}
&\text{Structural constraints:}\nonumber\\
&\hspace{1.5cm}\alpha_\mathcal{S}\geq0,\quad\quad\forall \mathcal{S}\subseteq \mathcal{I}_1\label{Generalbme-posor}\\
&\hspace{1.5cm}R_2=\sum_{\mathcal{S}\subseteq \mathcal{I}_1} \alpha_\mathcal{S},\label{General-R2}\\
&\text{Decodability constraints at public receivers:}\nonumber\\
&\hspace{1.5cm}\text{ {$\sum_{\substack{\mathcal{S}\subseteq \mathcal{I}_1\\\mathcal{S}\ni i}}\!\!\alpha_\mathcal{S}\!\leq\!\sum_{\substack{\mathcal{S}\in\Lambda}}\!\alpha_\mathcal{S}\!+\!\hspace{-.05cm}I(\cup_{\!\substack{\mathcal{S}\subseteq \mathcal{I}_1\\ \mathcal{S}\ni i}}\!U_{\mathcal{S}};Y_i|\!\cup_{\!\substack{\mathcal{S}\in\Lambda}}\!U_{\mathcal{S}}\!)$}},\quad { \forall \Lambda\subseteq \{\{i\}\star\}\ \text{ superset saturated}},\ \forall i\in \mathcal{I}_1\label{General-r1or}\\
&\hspace{1.5cm}R_1+\sum_{\substack{\mathcal{S}\subseteq \mathcal{I}_1\\\mathcal{S}\ni i}}\alpha_\mathcal{S}\leq I(\cup_{\substack{\mathcal{S}\subseteq \mathcal{I}_1,\\ \mathcal{S}\ni i}}U_{\mathcal{S}};Y_i),\quad \forall i\in \mathcal{I}_1\\
&\text{Decodability constraints at private receivers:}\nonumber\\
&\hspace{1.5cm}R_2\leq \sum_{\mathcal{S}\in\Lambda}\alpha_\mathcal{S}+I\left(X;Y_p|\cup_{\mathcal{S}\in\Lambda}U_{\mathcal{S}}\right),\quad{ \forall \Lambda\subseteq 2^{\mathcal{I}_1}\ \text{ superset saturated}},\ \forall p\in \mathcal{I}_2\label{General-R1+R2or}\\
&\hspace{1.5cm}R_1+R_2\leq I(X;Y_p),\quad \forall p\in \mathcal{I}_2.\label{endeq}
\end{align}

We shall prove that one can achieve rate pairs which satisfy a relaxed version of \eqref{Generalbme-posor}-\eqref{endeq} with a simple block Markov encoding scheme. More specifically, the constraints in \eqref{Generalbme-posor} are relaxed to the following set of constraints:
\begin{eqnarray}
\begin{array}{l}
\sum_{\mathcal{S}\in\Lambda}\alpha_\mathcal{S}\geq 0\quad  \forall \Lambda\subseteq 2^{\mathcal{I}_1} \text{ superset saturated}.\label{General-relax}
\end{array}
\end{eqnarray}
We briefly outline this block Markov encoding scheme for the case where we have two public and one private receiver (the same  line of argument goes through for the general case). We devise our block Markov encoding scheme in three steps:

1)  Form an extended broadcast channel with input/outputs $X^\prime,Y^\prime_1,Y^\prime_2,Y^\prime_3$, as shown in Fig.~\ref{Chap-GeneralChannel-virtual}. We have 
\begin{align}
&X^\prime=(X,V_{\{1,2\}},V_{\{2\}},V_{\{1\}},V_\phi)\\
&Y^\prime_1=(Y_1,V_{\{1,2\}},V_{\{1\}})\\
&Y^\prime_2=(Y_2,V_{\{1,2\}},V_{\{2\}})\\
&Y^\prime_3=(X,V_{\{1,2\}},V_{\{2\}},V_{\{1\}},V_\phi),
\end{align}
where $V_\mathcal{S}$, $\mathcal{S}\subseteq\{1,2\}$, takes its value in an alphabet set $\mathcal{V}_\mathcal{S}$ of size $2^{\beta_\mathcal{S}}$. 
We call variables $V_\mathcal{S}$  the \textit{virtual signals}. 
\begin{figure}[ht!]
\centering
\begin{tikzpicture}[scale=1]
\tikzstyle{every node}=[draw,shape=circle,font=\footnotesize\itshape]; 

\draw(3cm,2cm) node[rectangle,minimum height=3cm] {$p(y_1,y_2,y_3|x)$};
\draw(3.2cm,0.2cm) node[rectangle,minimum height=7cm,minimum width=4.5cm]{};

\draw (-1cm,0cm) node[draw=none] (x'){};
\draw (-1cm,-.2cm) node[draw=none] (x'v12){};
\draw (-1cm,-.4cm) node[draw=none] (x'v1){};
\draw (-1cm,-.6cm) node[draw=none] (x'v2){};
\draw (-1cm,-.8cm) node[draw=none] (x'v0){};

\draw (3.9cm,3cm) node[draw=none] (y1){};
\draw (3.9cm,2cm) node[draw=none] (y2){};
\draw (3.9cm,1cm) node[draw=none] (y3){};

\draw [->] (x') -- (1.5cm,0cm) -- (1.5cm,2cm) -- (2cm,2cm);
\draw [-] (x'v12) -- (1.8cm,-.2cm) -- (1.8cm,-.2cm) -- (4.5cm,-.2cm);
\draw [-] (x'v1) -- (1.6cm,-.4cm) -- (1.6cm,-1cm) -- (4.7cm,-1cm);
\draw [-] (x'v2) -- (1.4cm,-.6cm) -- (1.4cm,-1.8cm) -- (4.9cm,-1.8cm);
\draw [-] (x'v0) -- (1.2cm,-.8cm) -- (1.2cm,-2.6cm) -- (5.1cm,-2.6cm);

\draw [->] (4.5cm,-.2cm)--(4.5cm,2.9cm)--(7cm,2.9cm);
\draw [->] (4.5cm,-.2cm)--(4.5cm,1.9cm)--(7cm,1.9cm);
\draw [->] (4.7cm,-1cm)--(4.7cm,2.8cm)--(7cm,2.8cm);
\draw [->] (4.9cm,-1.8cm)--(4.9cm,1.8cm)--(7cm,1.8cm);
\draw [->]  (5.1cm,-2.6cm)-- (5.1cm,.6cm)-- (7cm,.6cm);
\draw [->] (4.5cm,-.2cm)--(4.5cm,.9cm)--(7cm,.9cm);
\draw [->] (4.7cm,-1cm)--(4.7cm,.8cm)--(7cm,.8cm);
\draw [->] (4.9cm,-1.8cm)--(4.9cm,.7cm)--(7cm,.7cm);

\draw [->] (y1) -- (7cm,3cm);
\draw [->] (y2) -- (7cm,2cm);
\draw [->] (y3) -- (7cm,1cm);


\draw(-.5cm,.4cm) node[draw=none] {$X^\prime$};
\draw(1.5cm,2.4cm) node[draw=none] {$X$};
\draw(3cm,0cm) node[draw=none] {$V_{\{1,2\}}$};
\draw(3cm,-.8cm) node[draw=none] {$V_{\{1\}}$};
\draw(3cm,-1.6cm) node[draw=none] {$V_{\{2\}}$};
\draw(3cm,-2.4cm) node[draw=none] {$V_{\phi}$};
\draw(6.5cm,1.4cm) node[draw=none] {$Y^\prime_3$};
\draw(6.5cm,2.4cm) node[draw=none] {$Y^\prime_2$};
\draw(6.5cm,3.4cm) node[draw=none] {$Y^\prime_1$};

\end{tikzpicture}
\caption{The extended broadcast channel obtained from $p(y_1,y_2,y_3|x)$ and the virtual signals $V_{\{1,2\}}$, $V_{\{2\}}$, $V_{\{1\}}$, $V_{\phi}$.}
\label{Chap-GeneralChannel-virtual}
\end{figure}

2) Design a general {multicast code} over the extended channel.
We say that a multicast code achieves the rate tuple $(R_1,\alpha^\prime_{\{1,2\}},\alpha^\prime_{\{2\}},\alpha^\prime_{\{1\}},\alpha^\prime_\phi)$, if it communicates a message of rate $R_1$ to all receivers and independent messages of rates $\alpha^\prime_\mathcal{S}$, $\mathcal{S}\subseteq \{1,2\}$, to public receivers in $\mathcal{S}$ and all private receivers. 
We design
 such a multicast code using superposition coding.
Conditions under which this encoding scheme achieves
 a rate tuple 
$(R_1,\alpha^\prime_{\{1,2\}},\alpha^\prime_{\{2\}},\alpha^\prime_{\{1\}},\alpha^\prime_\phi)$ over the extended broadcast channel 
are readily given by inequalities in \eqref{General-r1or}-\eqref{General-R1+R2or}  (for parameters $\alpha^\prime_\mathcal{S}$,  auxiliary random variables $U^\prime_\mathcal{S}$, $\phi\hspace{-.05cm}\neq\hspace{-.05cm}\mathcal{S}\hspace{-.05cm}\subseteq\hspace{-.05cm} {\mathcal{I}_1}$, and input/outputs $X^\prime\hspace{-.05cm}, Y^\prime_1\hspace{-.05cm},Y^\prime_2\hspace{-.025cm},\hspace{-.025cm}Y^\prime_3$).

3) Emulate the virtual signals. An extension to Lemma \ref{lb-CombNet-bm-lem-match} provides us with sufficient conditions. 

We now
use the information bits that are to be encoded in block $t+1$, to also
convey (the content of) the virtual signals in block $t$. We use the remaining information bits, 
not assigned to the virtual signals, to communicate the 
common and private messages.
Putting together the constraints needed in the above three steps (as in Section \ref{CombNet-ach-blockMarkov}), we obtain an achievable rate region for each joint probability distribution of the form
 {$p(u^\prime_{\{1,2\}})p(u^\prime_{\{1\}}|u^\prime_{\{1,2\}})p(u^\prime_{\{2\}}|u^\prime_{\{1,2\}})p(x^\prime|u^\prime_{\{2\}},u^\prime_{\{1\}},u^\prime_{\{1,2\}})$}.
In particular, by a proper choice for the auxiliary random variables, we show that  the rate region defined in \eqref{General-R2}-\eqref{General-relax} is achievable. 
More precisely, we have the following theorem (details of the proof are available in \cite[Theorem 4.3]{Saeedi12}).
\begin{theorem}
\label{Gen-bme-theorem}
The rate pair $(R_1,R_2)$ is achievable if there exist parameters $\alpha_\mathcal{S}$, $\mathcal{S}\subseteq \mathcal{I}_1$, and auxiliary random variables $U_\mathcal{S}$, $\phi\neq\mathcal{S}\subseteq {\mathcal{I}_1}$, such that they satisfy   inequalities in \eqref{General-R2}-\eqref{General-relax} for a joint pmf of the form $$\prod_{k=1}^K\prod_{\substack{\mathcal{S}\subseteq \mathcal{I}_1\\|\mathcal{S}|=k}}p(u_{\mathcal{S}}|\{u_{\mathcal{T}}\}_{\substack{\mathcal{T}\in\{\mathcal{S}\star\}\\\mathcal{T}\neq\mathcal{S}}})p(x|\{u_{\mathcal{S}}\}_{\substack{\mathcal{S}\subseteq \mathcal{I}_1\\\mathcal{S}\neq \phi}}).$$
\end{theorem}

\begin{remark}
\label{remarkstrict}
Note that the rate-region in Theorem \ref{Gen-bme-theorem} looks similar to that of superposition coding, see \eqref{Generalbme-posor}-\eqref{General-R1+R2or}. Clearly, the former rate-region has a less constrained set of non-negativity constraints on  $\alpha_\mathcal{S}$ and includes the latter.
 It is interesting to ask if this inclusion is strict, and it is non-trivial to answer this question because of the union that is taken over all proper pmfs. For a fixed pmf, the inclusion is strict for $m\geq 3$ public (and any number of private) receivers. 
So it is possible that the proposed block Markov  scheme strictly enlarges the rate-region of superposition coding. However, this needs further investigation.
\end{remark}

\begin{remark}
One may design a more general block Markov  scheme by using Marton's coding in the second step (when devising a  multicast code for the extended broadcast channel). Therefore, following similar steps as above, non-negativity constraints on rate-split parameters can be (partially) relaxed from the rate-region that is achievable by rate-splitting, superposition coding, and Marton's coding, see \cite[Theorem 4.3]{Saeedi12}. 
\end{remark}

\section{Conclusion}
\label{conclusion}
In this paper, we studied the problem of multicasting two nested message sets over combination networks and gave a full characterization of the capacity region for combination networks with three (or fewer) public and any number of private receivers. 

More generally, we discussed three encoding schemes which are based on linear superposition schemes. We showed that the standard linear superposition encoding scheme is optimal for networks with two public and any number of private receivers. For networks with more than two public receivers, however, this scheme is sub-optimal. We discussed two improved schemes. The first scheme uses an appropriate pre-encoding at the source, followed by a linear superposition scheme. We characterized the achievable rate-region in terms of a feasibility problem (Theorem \ref{lb-CombNet-Theoremk2}) and showed its tightness for networks with three (or fewer) public and any number of private receivers (Theorem \ref{lb-CombNet-Theorem-outer2}). We illustrated an example (Example \ref{CombNet-BME-example22-3}) where this scheme performs sub-optimally, and proposed an optimal block Markov encoding scheme. Motivated by this example, we proposed a block Markov encoding scheme and characterized its achievable rate-region (Theorem \ref{lb-CombNetbme-Theorem}). This scheme, also, is capacity achieving for network with three (or fewer) public and any number of receivers (Theorem \ref{lb-CombNet-Theorem-outer3}). While the rate-regions of Theorem \ref{lb-CombNet-Theoremk2} and Theorem \ref{lb-CombNetbme-Theorem} are not comparable in general, we showed an example where Theorem \ref{lb-CombNetbme-Theorem} strictly includes Theorem \ref{lb-CombNet-Theoremk2}. We conjecture that this inclusion always holds.

We discussed that combination networks are an interesting class of networks to study especially since they also form a class of resource-based broadcast channels. To illustrate the implications of our study over broadcast channels, we generalized the block Markov encoding scheme that we proposed in the context of combination networks and proposed a new achievable scheme for broadcast channels with two nested message sets. The rate-region achieved by this scheme includes the previous rate-regions. It remains open whether this inclusion is strict.


%
\appendices
\section{Proof of Lemma \ref{lb-CombNet2-lemmaW1}}
\label{ap-CombNet2-lemmaW1}
We start with proving the ``only if'' statement. Since $W_1$ is recoverable from $Y$,  for any $W_1,W_2,W^\prime_1,W^\prime_2$ it holds
that if $Y=Y^\prime$ (or equivalently $\mathbf{T}_1(W_1-W^\prime_1)+\mathbf{T}_2(W_2-W^\prime_2)=0$), then $W_1=W^\prime_1$. 
In particular, for $W_2=W_2^\prime$ one finds that for any $W_1,W^\prime_1$ equation $\mathbf{T}_1(W_1-W_1^\prime)=0$ 
results in $W_1=W^\prime_1$. Therefore, $\mathbf{T}_1$ is  full column rank; i.e., $\rank(\mathbf{T}_1)=R_1$. Furthermore, 
for all vectors $W_1, W^\prime_1,W_2,W^\prime_2$ such that $W_1\neq W^\prime_1$, one obtains 
$\mathbf{T}_1(W_1-W^\prime_1)\neq \mathbf{T}_2(W^\prime_2-W_2)$; i.e., the column space of matrix $\mathbf{T}_2$ is disjoint 
from the column space of matrix $\mathbf{T}_1$.

To prove the ``if" statement, we prove that if it holds that $\rank(\mathbf{T}_1)=R_1$ and column spaces of 
$\mathbf{T}_1$ and $\mathbf{T}_2$ are disjoint, then equation $Y=Y^\prime$ (or equivalently  $\mathbf{T}_1(W_1-W^\prime_1)+\mathbf{T}_2(W_2-W^\prime_2)=0$) 
results in $W_1=W^\prime_1$ for all vectors $W_1,W_2,W^\prime_1,W^\prime_2$. We show this by contradiction. 
Let $\mathbf{T}_1(W_1-W^\prime_1)+\mathbf{T}_2(W_2-W^\prime_2)=0$ and $W_1\neq W^\prime_1$. For the cases where 
$\mathbf{T}_2(W_2- W^\prime_2)=0$, we get $\mathbf{T}_1(W_1-W^\prime_1)=0$ for $W_1\neq W^\prime_1$, which contradicts the 
original assumption of $\rank(\mathbf{T}_1)=R_1$. If $\mathbf{T}_2(W_2-W^\prime_2)\neq 0$, then $\mathbf{T}_1(W_1-W^\prime_1)+\mathbf{T}_2(W_2-W^\prime_2)=0$ 
implies that there exists at least one non-zero vector in the intersection of the column spaces of $\mathbf{T}_1$ and 
$\mathbf{T}_2$ which contradicts the assumption that column space of $\mathbf{T}_1$ and $\mathbf{T}_2$ are disjoint.

\section{Proof of Lemma \ref{lb-CombNet-declem-2}}
\label{ap-CombNet-declem-2}
We start with proving that $(i)$ implies $(ii)$. Assume that matrix $\mathbf{T}$ is assigned  full column rank.
Over the equivalent unicast network, we propose a network code that constitutes a transfer matrix (from node $A$ to node $B$) exactly equal to the full-rank matrix $\mathbf{T}$: First, let each outgoing edge of the source carry one uncoded symbol of the message. Then, let the first $r_{\{1,2\}}$ rows of matrix $\mathbf{T}$ specify the local encoding matrix at node $n^\prime_{\{1,2\}}$. Similarly, let  the second $r_{\{2\}}$, 
the third $r_{\{1\}}$, and the last $r_\phi$ set  of rows of matrix $\mathbf{T}$ specify the local encoding matrices at nodes $n^\prime_{\{2\}}$, $n^\prime_{\{1\}}$, and $n^\prime_{\phi}$, respectively. Note that the zero structure of matrix $\mathbf{T}$, which is given in equation \eqref{lb-CombNet-structureB}, ensures that this is a well defined construction. It follows that $\mathbf{T}$ is the transfer matrix from node $A$ to node $B$ and, since it is full column rank, the encoded message can be decoded at sink node $B$.

Statement $(ii)$ also implies $(i)$ and the proof is by construction.  If a message of rate $c$ could be sent over 
the equivalent unicast network (from node $A$ to node $B$), then there are $c$ edge disjoint paths from the 
source $A$ to the sink  $B$. Each such path matches one of the outgoing edges of source node $A$ to one of the 
incoming edges of sink node $B$. We call this matching between the outgoing edges of the source and the incoming 
edges of the sink matching $M$ and we note that its size is $c$. We use this matching to fill the indeterminates of matrix $\mathbf{T}$ with $0-1$. 
First of all, note that each column $j$ of matrix $\mathbf{T}$ corresponds to an outgoing edge of the source, say $e_j$, 
and each row $i$ of matrix $\mathbf{T}$ corresponds to an incoming edge of the sink, say $e^\prime_i$. So each entry 
$(i,j)$ of matrix $\mathbf{T}$ is priorly set to zero if and only  if edges $e_j$ and $e^\prime_j$ cannot be matched.
Now, put a $1$ in entry $(i,j)$ of matrix $\mathbf{T}$ if edge $e_j$ is matched to edge $e^\prime_j$ over the matching $M$. 
Since matching $M$ has a size equal to $c$, matrix $\mathbf{T}$ is made full column rank.

\section{Proof of Lemma \ref{lb-CombNet-mincut}}
\label{ap-CombNet-mincut}
Consider all (edge-) cuts separating the source node $A$ from the sink node $B$. Since the intermediate edges all have 
infinite capacity, the minimum cut does not contain any 
edges from them. 
One can verify the following fact over Fig.~\ref{FigCombinationNetworkBasicLemma2}: 
If an edge $(n^\prime_{\mathcal{S}^\prime},B)$, $\mathcal{S}^\prime\subseteq \{1,\ldots,t\}$, does not belong to the cut, 
then all edges $(A,n_\mathcal{S})$ where $\mathcal{S}\supseteq \mathcal{S}^\prime$ belong to that cut.
So each (finite-valued) cut corresponds to a subset of the nodes $n^\prime_\mathcal{S}$ in the third layer. Denoting
this subset of nodes by $\{n^\prime_\mathcal{S}\}_{\mathcal{S}\in{\Gamma}}$, for some $\Gamma\subseteq2^{\{1,\ldots,t\}}$, we
have the value of the cut given by
\begin{align}
\label{app-CombNetk2-mincut-1}
\sum_{\mathcal{S}\in \Gamma}{r_{\mathcal{S}}}+\sum_{\mathcal{S}\supseteq \mathcal{S}^\prime,\ \mathcal{S}^\prime\in\Gamma^c} c_{\mathcal{S}}.
\end{align}
It is not difficult to verify that the (inclusion-wise) minimal cuts are derived for sets $\Gamma^c$ that are superset 
saturated. Renaming $\Gamma^c$ as $\Lambda$ concludes the proof.
\section{Proof of Lemma \ref{lb-lemGifull}}
\label{ap-lemGifull}
Let $r_i$ denote the number of rows in $\mathbf{G}^{(i)}$; i.e., let $r_i= \sum_{\substack{\mathcal{S}\subseteq \mathcal{I}_1\ \mathcal{S}\ni i}}|\mathcal{E}_{\mathcal{S}}|$. In the matrix $\mathbf{G}^{(i)}=[\mathbf{B}^{(i)}|\mathbf{L}^{(i)}_1]$, first select an assignment  for the columns of  $\mathbf{L}^{(i)}_1$ that makes them linearly independent (such an assignment exists from construction). 
  Since the element variables in $\mathbf{L}^{(i)}_1$ are independent of those in $\mathbf{B}^{(i)}$, $\mathbf{G}^{(i)}$ can be  made full-rank just by picking $R_1$  vectors from $\mathbb{F}^{r_i}$, linearly independent from the columns of $\mathbf{L}^{(i)}_1$, for the columns of $\mathbf{B}^{(i)}$. This is possible since $\rank(\mathbf{B}^{(i)})$ is at most $\sum_{\substack{\mathcal{S}\subseteq \mathcal{I}_1\ \mathcal{S}\ni i}}\alpha_{\mathcal{S}}$ which is bounded by $r_i-R_1$ (by assumption).

\section{Decodability constraints for private receivers in Theorem \ref{lb-CombNet-Theoremk2}}
\label{ap-CombNet-lemmaW1W2}
For simplicity of notation, we give the proof for the case where $m=2$.
We prove that $\mathbf{T}^{(p)}\mathbf{P}$ could be made full column rank if and only if message $W_2$ could be unicast over the  network of Fig.~\ref{FigCombinationNetworkLemmaNegativea0}. In this network the outgoing edges of the source and the incoming edges of the sink are all of unit
capacity and the bold edges in the middle are of infinite capacity. Define $\mathbf{T}^{(p)\prime}:=\mathbf{T}^{(p)}\mathbf{P}$ which is given by 
\begin{eqnarray}
\label{lb-CombNet-structureTp}
\begin{array}{c}\\\mathbf{T}^{(p)\prime}=\end{array}\begin{array}{l}\begin{array}{cccc}\stackrel{\alpha_{\{1,2\}}}{\longleftrightarrow}&\hspace{-.3cm}\stackrel{\alpha_{\{1\}}}{\leftrightarrow}&\hspace{-.35cm}\stackrel{\alpha_{\{2\}}}{\leftrightarrow}&\hspace{-0.25cm}\stackrel{\alpha_{\phi}}{\leftrightarrow}\end{array}\\
\left[\begin{array}{c|c|c|c}
\hspace{.4cm}\ &0&0&0\\\hline
\hspace{.4cm}\ &\hspace{.2cm}&0&0\\\hline
\hspace{.4cm}\ &0&\hspace{.2cm}&0\\\hline
\hspace{.4cm}\ &\hspace{.2cm}&\hspace{.2cm}&\hspace{.2cm}
                  \end{array}\right] \begin{array}{c}\updownarrow\\\updownarrow\\\updownarrow\\\updownarrow\end{array}
	   \hspace{-.3cm} {\tiny{\begin{array}{l}\\|\e_{\{1,2\},p}|\vspace{.3cm}\\ \vspace{.3cm}|\e_{\{1\},p}|\\\vspace{.3cm}|\e_{\{2\},p}|\\\vspace{.3cm}|\e_{\phi,p}|\end{array}}}
\end{array}\begin{array}{c}\\\hspace{-.3cm}\cdot\hspace{.5cm}\mathbf{P}.\end{array}
\end{eqnarray}
Think of matrix $\mathbf{P}$ as the local transfer matrix at source node $A$. Also, think of the matrix formed by the 
first $\card{\e_{\{1,2\},p}}$ rows of $\mathbf{T}^{(p)\prime}$ as the local transfer matrix at intermediate node 
$n_{\{1,2\},p}$, the matrix formed by the second $\card{\e_{\{1\},p}}$ rows of $\mathbf{T}^{(p)\prime}$ as the local transfer 
matrix at intermediate node $n_{\{1\}}$ and so on. Notice the equivalence between matrix $\mathbf{T}^{(p)\prime}$ and 
the transfer matrix imposed by a  linear network code from node $A$ to node $B$ over the  network of 
Fig.~\ref{FigCombinationNetworkLemmaNegativea0}, and conclude that message $W_2$ is decodable  if and only if message $W_2$ 
could be unicast from $A$ to $B$.
\begin{figure}
\centering
\begin{tikzpicture}[scale=2.5]
\tikzstyle{every node}=[draw,shape=circle,minimum size=1cm,font=\small\itshape]; 
\path (0:0cm) node (v0) {$A$}; 
\path (20:.5cm) node [draw=none] {source}; 

\path(180+30:1.8cm)  node (v12) {$n_{\{1,2\}}$}; 
\path(180+70:1cm)  node (v2) {$n_{\{2\}}$}; 
\path (180+110:1cm) node (v1) {$n_{\{1\}}$}; 
\path (180+150:1.8cm) node (vemp) {$n_\phi$};

\path (180+45:2.5cm)  node (u12) {$n^\prime_{\{1,2\},p}$};
\path(180+70:1.87cm)  node (u2) {$n^\prime_{\{2\},p}$}; 
\path(180+110:1.87cm)  node (u1) {$n^\prime_{\{1\},p}$}; 
\path (180+135:2.5cm) node (uemp) {$n^\prime_{\phi,p}$};

\path (180+90:2.6cm) node (d) {$B$};
\path (-80:2.8cm) node[draw=none] {sink};

\draw[->] (v0) -- node[left,draw=none,fill=none] {$\alpha_{\{1,2\}}$}    (v12) ;

\draw[->] (v0) -- node[left,draw=none,fill=none] {$\alpha_{\{2\}}$}    (v2) ;

\draw[->] (v0) -- node[right,draw=none,fill=none] {$\alpha_{\{1\}}$}    (v1) ;

\draw[->] (v0) -- node[right,draw=none,fill=none] {$(\alpha_\phi)^+$}    (vemp) ;


\draw[->,line width=.07cm] (v12) -- (u12);
\draw[->,line width=.07cm] (v12) -- (u2);
\draw[->,line width=.07cm] (v12) -- (u1);
\draw[->,line width=.07cm] (v12) -- (uemp);
\draw[->,line width=.07cm] (v2) -- (u2);
\draw[->,line width=.07cm] (v2) -- (uemp);
\draw[->,line width=.07cm] (v1) -- (u1);
\draw[->,line width=.07cm] (v1) -- (uemp);
\draw[->,line width=.07cm] (vemp) -- (uemp);

\draw[->] (u12) -- node[left,draw=none,fill=none,yshift=-.1cm] {$\card{\e_{\{1,2\},p}}$}    (d) ;

\draw[->] (u2) -- node[left,draw=none,fill=none,yshift=.1cm] {$\card{\e_{\{2\},p}}$}    (d) ;

\draw[->] (u1) -- node[right,draw=none,fill=none,yshift=.1cm] {$\card{\e_{\{1\},p}}$}    (d) ;

\draw[->] (uemp) -- node[right,draw=none,fill=none,yshift=-.1cm] {$\card{\e_{\phi,p}}$}    (d) ;

\end{tikzpicture}
\caption{Source node $A$ communicates a message of rate $R_2$ to the sink node, $B$.}
\label{FigCombinationNetworkLemmaNegativea0}
\end{figure}
Decodability conditions at receiver $p$ can, therefore, be inferred from the min-cut separating nodes $A$ and $B$ over the  network of Fig.~\ref{FigCombinationNetworkLemmaNegativea0}. Lemma \ref{lb-CombNet-mincut} gives this min-cut by the following expression.
\begin{eqnarray}
\label{lb-CombNet-mincutap}
\min\left\{\min_{\substack{\Lambda \subset 2^{\mathcal{I}_1}\\ \Lambda\text{superset saturated}}}\sum_{\mathcal{S}\in \Lambda}\alpha_\mathcal{S}+\sum_{\mathcal{S}\in \Gamma^c}\card{\e_{\mathcal{S},p}},\sum_ {\mathcal{S}\in 2^{\mathcal{I}_1},\ \mathcal{S}\neq \phi}\alpha_\mathcal{S}+(\alpha_\phi)^+\right\}.
\end{eqnarray}
One can verify that $R_2$ is smaller than the expression in \eqref{lb-CombNet-mincutap}, provided that inequalities in \eqref{CombNetk2-ach1or} hold.

\section{Application of the Sparse Zeros Lemma to the proof of Theorem \ref{lb-CombNet-Theoremk2}}
\label{sparselemmamix}
We constructed matrices $\mathbf{G}^{(i)}$, $i\in\mathcal{I}$ and reduced the problem to making all these matrices simultaneously full-rank. 
Matrices $\mathbf{G}^{(i)}$ have their entries defined by the variables in $\mathbf{A}$ and $\mathbf{P}$. We also discussed that each of these matrices could be made full column rank.
This implies that there exists a  square submatrix of each $\mathbf{G}^{(i)}$, say ${\mathbf{G}}_s^{(i)}$, that could be made full-rank.  Let   $\mathcal{P}^{(i)}$ be the polynomial corresponding to the determinant of   $\mathbf{G}_s^{(i)}$, and $\mathcal{P}=\prod_i\mathcal{P}^{(i)}$. 
 Given that there  exists an assignment for the variables such that each individual polynomial $\mathcal{P}^{(i)}$ is 
 non-zero, we can  conclude from the sparse zeros lemma that there exists an assignment such that all polynomials are 
 simultaneously non-zero. 
We can furthermore provide an upper bound on the required size for $\mathbb{F}$. This is done next by finding the degree of each polynomial $\mathcal{P}^{(i)}$ in each variable. 

For $i\in\mathcal{I}_1$, since all variables in $\mathbf{G}_s^{(i)}$ are independent of each other, the desired degree is at most $1$. For $\mathcal{P}^{(p)}$, $p\in\mathcal{I}_2$, also,  we can show that the maximum degree in each variable is $1$. To see this, we proceed as follows. Recall that for $p\in\mathcal{I}_2$, $\mathbf{G}^{(p)}=\mathbf{T}^{(p)}\mathbf{P}$. Let us denote the $(i,j)^{\text{th}}$ entry of $\mathbf{G}^{(p)}$ by $g_{i,j}$. So each $g_{i,j}=\sum_l t_{i,l}p_{l,j}$, where $t_{i,l}$ refers to the $(i,l)^{\text{th}}$ entry of $\mathbf{T}^{(p)}$ and $p_{l,j}$ refers to the $(l,j)^{\text{th}}$ entry of $\mathbf{P}$. Using  Laplace expansion, we have
\begin{eqnarray}
\det{\mathbf{G}_s^{(p)}}= \sum_i (-1)^{i+j}g_{i,j}\det{\mathbf{G}^{(p)}_{i,j}},
\end{eqnarray}
where $\det{\mathbf{G}_{i,j}}$ is the matrix obtained from $\mathbf{G}^{(p)}_s$ after removing the $i^{\text{th}}$ row and the $j^{\text{th}}$ column. Now, note that $\det{\mathbf{G}_{i,j}}$ is not a function of variables $\{t_{i,l}\}_l$ (which are indeed variables of $\mathbf{A}$), nor is it a function of variables $\{p_{l,j}\}_l$.  Thus, the degree of $\mathcal{P}^{(p)}$  is at most $1$ in each of the variables of $\mathbf{A}$ and $\mathbf{P}$.

Let us construct the polynomial $\mathcal{P}=\prod_{i\in\mathcal{I}}\mathcal{P}^{(i)}$. Each $\mathcal{P}^{(i)}$ is of degree 
at most $1$ in each variable. So, the degree of $\mathcal{P}$ is at most $K$ in each variable.
From the sparse zeros lemma \cite[Lemma 2.3]{kcl07}, $\mathbb{F}$ need only be such that $|\mathbb{F}|\geq K$.

\section{Proof of Lemma \ref{lb-CombNet-bm-lem-match}}
\label{ap-CombNet-bm-lem-match}

Each virtual resource $v\in\mathcal{V}_\mathcal{S}$ can be emulated by one information symbol from any of the messages that 
are destined to all end-destinations of $v$; i.e., all messages $W^{\prime \mathcal{S}^\prime}_2$  
where $\mathcal{S}^\prime\supseteq \mathcal{S}$. One can form a bipartite graph with the virtual resources as one set of 
nodes and the information symbols as the other set of nodes (see Fig.~\ref{FigCombinationNetworkMatching}). Each virtual resource is connected to those information 
symbols that can emulate it. These edges are shown in light colour in Fig.~\ref{FigCombinationNetworkMatching}. Therefore, all resources are emulatable if there exists a matching of size 
$\sum_{\mathcal{S}\subseteq \mathcal{I}_1} \beta_\mathcal{S}$ over the bipartite graph of 
Fig.~\ref{FigCombinationNetworkMatching}. The matching is shown via bold edges.  It is well-known that there exists such a matching if and only if a flow of 
amount $\sum_{\mathcal{S}\subseteq \mathcal{I}_1} \beta_\mathcal{S}$ could be sent over the network 
of Fig.~\ref{FigCombinationNetworkflow} from node $A$ to node $B$. The min-cut separating nodes $A$ and $B$ is given by the 
following expression (see the proof of Lemma \ref{lb-CombNet-mincut}):
\begin{figure}
 \centering
\begin{subfigure}[b]{0.45\textwidth}
\begin{tikzpicture}[scale=2.2]
\tikzstyle{every node}=[draw,shape=circle,minimum size=.01cm,font=\small\itshape]; 
\node[draw=none] () at (-1.2,0.3) {$\beta_{\phi}$};
\node[draw=none] () at (-1.2,0.15) {$\overbrace{\hspace{1cm}}$};
\node[rectangle,draw=none]() at (0,.5) {Virtual resources};
\node[rectangle,draw=none]() at (0,-2.5) {Information symbols};

\node[] (v12) at (-1.2,0) {};
\node[] (v12') at (-1.35,0) {};
\node[] (v12'') at (-1.05,0) {};

\node[draw=none] () at (-.375,0.3) {$\beta_{\{1\}}$};
\node[draw=none] () at (-.375,0.15) {$\overbrace{\hspace{.5cm}}$};

\node[] (v2) at (-.3,0) {};
\node[] (v2') at (-.45,0) {};

\node[draw=none] () at (.525,0.3) {$\beta_{\{2\}}$};
\node[draw=none] () at (.525,0.15) {$\overbrace{\hspace{1.2cm}}$};

\node[] (v1) at (.3,0) {};
\node[] (v1') at (.45,0) {};
\node[] (v1'') at (.6,0) {};
\node[] (v1''') at (.75,0) {};

\node[draw=none] () at (1.275,0.3) {${\beta_{\{1,2\}}}$};
\node[draw=none] () at (1.275,0.15) {$\overbrace{\hspace{.5cm}}$};

\node[] (vemp) at (1.2,0) {};
\node[] (vemp') at (1.35,0) {};

\node[] (u12) at (-.9-.15,-2) {};
\node[] (u12') at (-1.05-.15,-2) {};
\node[draw=none] () at (-.975-.15,-2.3) {$\alpha_{\phi}$};
\node[draw=none] () at (-.975-.15,-2.15) {$\underbrace{\hspace{.5cm}}$};

\node[] (u2) at (-.3-.15,-2) {};
\node[] (u2') at (-.45-.15,-2) {};
\node[] (u2'') at (-.15-.15,-2) {};
\node[draw=none] () at (-.3-.15,-2.3) {${\alpha_{\{1\}}}$};
\node[draw=none] () at (-.3-.15,-2.15) {$\underbrace{\hspace{1cm}}$};

\node[] (u1) at (.3-.15,-2) {};
\node[draw=none] () at (.3-.15,-2.3) {${\alpha_{\{2\}}}$};
\node[draw=none] () at (.3-.15,-2.15) {$\underbrace{\hspace{.1cm}}$};

\node[] (uemp) at (1.2-.15,-2) {};
\node[] (uemp') at (1.35-.15,-2) {};
\node[] (uemp'') at (1.5-.15,-2) {};
\node[] (uemp''') at (1.05-.15,-2) {};
\node[] (uemp'''') at (.9-.15,-2) {};
\node[draw=none] () at (1.05,-2.3) {${\alpha_{\{1,2\}}}$};
\node[draw=none] () at (1.05,-2.15) {$\underbrace{\hspace{1.5cm}}$};

\draw[->,gray!50,very thin] (v12') -- (u12);
\draw[->,gray!50,very thin] (v12'') -- (u12);

\draw[->,gray!50,very thin] (v12) -- (u2);
\draw[->,gray!50,very thin] (v12') -- (u2);
\draw[->,gray!50,very thin] (v12'') -- (u2);

\draw[->,gray!50,very thin] (v12) -- (u1);
\draw[->,gray!50,very thin] (v12') -- (u1);
\draw[->,gray!50,very thin] (v12'') -- (u1);

\draw[->,gray!50,very thin] (v12) -- (uemp);
\draw[->,gray!50,very thin] (v12') -- (uemp);
\draw[->,gray!50,very thin] (v12'') -- (uemp);

\draw[->,gray!50,very thin] (v12) -- (u12');
\draw[->,black, line width=.05cm] (v12') -- (u12');
\draw[->,gray!50,very thin] (v12'') -- (u12');

\draw[->,gray!50,very thin] (v12) -- (u2');
\draw[->,gray!50,very thin] (v12') -- (u2');

\draw[->,gray!50,very thin] (v12) -- (uemp');
\draw[->,gray!50,very thin] (v12') -- (uemp');
\draw[->,gray!50,very thin] (v12'') -- (uemp');

\draw[->,gray!50,very thin] (v12) -- (uemp'');
\draw[->,gray!50,very thin] (v12') -- (uemp'');
\draw[->,gray!50,very thin] (v12'') -- (uemp'');

\draw[->,gray!50,very thin] (v12) -- (uemp''');
\draw[->,gray!50,very thin] (v12') -- (uemp''');
\draw[->,gray!50,very thin] (v12'') -- (uemp''');

\draw[->,gray!50,very thin] (v12) -- (uemp'''');
\draw[->,gray!50,very thin] (v12') -- (uemp'''');
\draw[->,gray!50,very thin] (v12'') -- (uemp'''');

\draw[->, gray!50,very thin] (v2) -- (u2);
\draw[->,gray!50,very thin] (v2) -- (u2');
\draw[->,gray!50,very thin] (v2') -- (u2');
\draw[->, gray!50,very thin] (v2') -- (u2'');

\draw[->,gray!50,very thin] (v2) -- (uemp);
\draw[->,gray!50,very thin] (v2) -- (uemp');
\draw[->,gray!50,very thin] (v2) -- (uemp'');
\draw[->,gray!50,very thin] (v2) -- (uemp''');
\draw[->,gray!50,very thin] (v2) -- (uemp'''');

\draw[->,gray!50,very thin] (v2') -- (uemp);
\draw[->,gray!50,very thin] (v2') -- (uemp');
\draw[->,gray!50,very thin] (v2') -- (uemp'');
\draw[->,gray!50,very thin] (v2') -- (uemp''');
\draw[->,gray!50,very thin] (v2') -- (uemp'''');

\draw[->,gray!50,very thin] (v1') -- (u1);
\draw[->,gray!50,very thin] (v1'') -- (u1);
\draw[->,gray!50,very thin] (v1''') -- (u1);

\draw[->,gray!50,very thin] (v1) -- (uemp);
\draw[->,gray!50,very thin] (v1) -- (uemp');
\draw[->,gray!50,very thin] (v1) -- (uemp'');
\draw[->,gray!50,very thin] (v1) -- (uemp''');
\draw[->,gray!50,very thin] (v1) -- (uemp'''');

\draw[->,gray!50,very thin] (v1') -- (uemp);
\draw[->,gray!50,very thin] (v1') -- (uemp');
\draw[->,gray!50,very thin] (v1') -- (uemp'');
\draw[->, gray!50,very thin] (v1') -- (uemp''');

\draw[->,gray!50,very thin] (v1'') -- (uemp);
\draw[->,gray!50,very thin] (v1'') -- (uemp');
\draw[->,gray!50,very thin] (v1'') -- (uemp'');
\draw[->, gray!50,very thin] (v1'') -- (uemp'''');

\draw[->,gray!50,very thin] (v1''') -- (uemp');
\draw[->,gray!50,very thin] (v1''') -- (uemp'');
\draw[->,gray!50,very thin] (v1''') -- (uemp''');
\draw[->,gray!50,very thin] (v1''') -- (uemp'''');

\draw[->,black, line width=.05cm] (vemp) -- (uemp');
\draw[->,gray!50,very thin] (vemp) -- (uemp);
\draw[->,gray!50,very thin] (vemp) -- (uemp'');
\draw[->,gray!50,very thin] (vemp) -- (uemp''');
\draw[->,gray!50,very thin] (vemp) -- (uemp'''');
\draw[->,gray!50,very thin] (vemp') -- (uemp);
\draw[->,gray!50,very thin] (vemp') -- (uemp');
\draw[->, black, line width=.05cm] (vemp') -- (uemp'');
\draw[->,gray!50,very thin] (vemp') -- (uemp''');
\draw[->,gray!50,very thin] (vemp') -- (uemp'''');

\draw[->, black, line width=.05cm] (v2) -- (u2'');
\draw[->, black, line width=.05cm] (v2') -- (u2);
\draw[->, black, line width=.05cm] (v1''') -- (uemp);
\draw[->,black, line width=.05cm] (v1'') -- (uemp''');
\draw[->, black, line width=.05cm] (v1) -- (u1);
\draw[->,black, line width=.05cm] (v1') -- (uemp'''');
\draw[->,black, line width=.05cm] (v12'') -- (u2');
\draw[->,black, line width=.05cm] (v12) -- (u12);

\end{tikzpicture}  \subcaption{The bi-partite graph between the virtual resources and the decodable information symbols.}\label{FigCombinationNetworkMatching}\end{subfigure}
\begin{subfigure}[b]{0.45\textwidth}
{\begin{tikzpicture}[scale=2.2]
\tikzstyle{every node}=[draw,shape=circle,minimum size=.001cm,font=\small\itshape]; 

\node (s) at (0,0.5) {$A$};
\node (d) at (0,-2.5) {$B$};

\node[] (v12) at (-1.2,0) {};
\node[] (v12') at (-1.35,0) {};
\node[] (v12'') at (-1.05,0) {};


\node[] (v2) at (-.3,0) {};
\node[] (v2') at (-.45,0) {};


\node[] (v1) at (.3,0) {};
\node[] (v1') at (.45,0) {};
\node[] (v1'') at (.6,0) {};
\node[] (v1''') at (.75,0) {};


\node[] (vemp) at (1.2,0) {};
\node[] (vemp') at (1.35,0) {};

\node[] (u12) at (-.9-.15,-2) {};
\node[] (u12') at (-1.05-.15,-2) {};

\node[] (u2) at (-.3-.15,-2) {};
\node[] (u2') at (-.45-.15,-2) {};
\node[] (u2'') at (-.15-.15,-2) {};

\node[] (u1) at (.3-.15,-2) {};

\node[] (uemp) at (1.2-.15,-2) {};
\node[] (uemp') at (1.35-.15,-2) {};
\node[] (uemp'') at (1.5-.15,-2) {};
\node[] (uemp''') at (1.05-.15,-2) {};
\node[] (uemp'''') at (.9-.15,-2) {};

\draw[->,gray!50,very thin] (v12') -- (u12);
\draw[->,gray!50,very thin] (v12'') -- (u12);

\draw[->,gray!50,very thin] (v12) -- (u2);
\draw[->,gray!50,very thin] (v12') -- (u2);
\draw[->,gray!50,very thin] (v12'') -- (u2);

\draw[->,gray!50,very thin] (v12) -- (u1);
\draw[->,gray!50,very thin] (v12') -- (u1);
\draw[->,gray!50,very thin] (v12'') -- (u1);

\draw[->,gray!50,very thin] (v12) -- (uemp);
\draw[->,gray!50,very thin] (v12') -- (uemp);
\draw[->,gray!50,very thin] (v12'') -- (uemp);

\draw[->,gray!50,very thin] (v12) -- (u12');
\draw[->,black, line width=.05cm] (v12') -- (u12');
\draw[->,gray!50,very thin] (v12'') -- (u12');

\draw[->,gray!50,very thin] (v12) -- (u2');
\draw[->,gray!50,very thin] (v12') -- (u2');

\draw[->,gray!50,very thin] (v12) -- (uemp');
\draw[->,gray!50,very thin] (v12') -- (uemp');
\draw[->,gray!50,very thin] (v12'') -- (uemp');

\draw[->,gray!50,very thin] (v12) -- (uemp'');
\draw[->,gray!50,very thin] (v12') -- (uemp'');
\draw[->,gray!50,very thin] (v12'') -- (uemp'');

\draw[->,gray!50,very thin] (v12) -- (uemp''');
\draw[->,gray!50,very thin] (v12') -- (uemp''');
\draw[->,gray!50,very thin] (v12'') -- (uemp''');

\draw[->,gray!50,very thin] (v12) -- (uemp'''');
\draw[->,gray!50,very thin] (v12') -- (uemp'''');
\draw[->,gray!50,very thin] (v12'') -- (uemp'''');

\draw[->, gray!50,very thin] (v2) -- (u2);
\draw[->,gray!50,very thin] (v2) -- (u2');
\draw[->,gray!50,very thin] (v2') -- (u2');
\draw[->, gray!50,very thin] (v2') -- (u2'');

\draw[->,gray!50,very thin] (v2) -- (uemp);
\draw[->,gray!50,very thin] (v2) -- (uemp');
\draw[->,gray!50,very thin] (v2) -- (uemp'');
\draw[->,gray!50,very thin] (v2) -- (uemp''');
\draw[->,gray!50,very thin] (v2) -- (uemp'''');

\draw[->,gray!50,very thin] (v2') -- (uemp);
\draw[->,gray!50,very thin] (v2') -- (uemp');
\draw[->,gray!50,very thin] (v2') -- (uemp'');
\draw[->,gray!50,very thin] (v2') -- (uemp''');
\draw[->,gray!50,very thin] (v2') -- (uemp'''');

\draw[->,gray!50,very thin] (v1') -- (u1);
\draw[->,gray!50,very thin] (v1'') -- (u1);
\draw[->,gray!50,very thin] (v1''') -- (u1);

\draw[->,gray!50,very thin] (v1) -- (uemp);
\draw[->,gray!50,very thin] (v1) -- (uemp');
\draw[->,gray!50,very thin] (v1) -- (uemp'');
\draw[->,gray!50,very thin] (v1) -- (uemp''');
\draw[->,gray!50,very thin] (v1) -- (uemp'''');

\draw[->,gray!50,very thin] (v1') -- (uemp);
\draw[->,gray!50,very thin] (v1') -- (uemp');
\draw[->,gray!50,very thin] (v1') -- (uemp'');
\draw[->, gray!50,very thin] (v1') -- (uemp''');

\draw[->,gray!50,very thin] (v1'') -- (uemp);
\draw[->,gray!50,very thin] (v1'') -- (uemp');
\draw[->,gray!50,very thin] (v1'') -- (uemp'');
\draw[->, gray!50,very thin] (v1'') -- (uemp'''');

\draw[->,gray!50,very thin] (v1''') -- (uemp');
\draw[->,gray!50,very thin] (v1''') -- (uemp'');
\draw[->,gray!50,very thin] (v1''') -- (uemp''');
\draw[->,gray!50,very thin] (v1''') -- (uemp'''');

\draw[->,black, line width=.05cm] (vemp) -- (uemp');
\draw[->,gray!50,very thin] (vemp) -- (uemp);
\draw[->,gray!50,very thin] (vemp) -- (uemp'');
\draw[->,gray!50,very thin] (vemp) -- (uemp''');
\draw[->,gray!50,very thin] (vemp) -- (uemp'''');
\draw[->,gray!50,very thin] (vemp') -- (uemp);
\draw[->,gray!50,very thin] (vemp') -- (uemp');
\draw[->, black, line width=.05cm] (vemp') -- (uemp'');
\draw[->,gray!50,very thin] (vemp') -- (uemp''');
\draw[->,gray!50,very thin] (vemp') -- (uemp'''');

\draw[->, black, line width=.05cm] (v2) -- (u2'');
\draw[->, black, line width=.05cm] (v2') -- (u2);
\draw[->, black, line width=.05cm] (v1''') -- (uemp);
\draw[->,black, line width=.05cm] (v1'') -- (uemp''');
\draw[->, black, line width=.05cm] (v1) -- (u1);
\draw[->,black, line width=.05cm] (v1') -- (uemp'''');
\draw[->,black, line width=.05cm] (v12'') -- (u2');
\draw[->,black, line width=.05cm] (v12) -- (u12);

\draw[->,black] (s) -- (vemp);
\draw[->,black] (s) -- (vemp');
\draw[->,black] (s) -- (v1);
\draw[->,black] (s) -- (v1');
\draw[->,black] (s) -- (v1'');
\draw[->,black] (s) -- (v1''');
\draw[->,black] (s) -- (v2);
\draw[->,black] (s) -- (v2');
\draw[->,black] (s) -- (v12);
\draw[->,black] (s) -- (v12');
\draw[->,black] (s) -- (v12'');

\draw[->,black] (uemp) -- (d);
\draw[->,black] (uemp') -- (d);
\draw[->,black] (uemp'') -- (d);
\draw[->,black] (uemp''') -- (d);
\draw[->,black] (uemp'''') -- (d);
\draw[->,black] (u1) -- (d);
\draw[->,black] (u2) -- (d);
\draw[->,black] (u2') -- (d);
\draw[->,black] (u2'') -- (d);
\draw[->,black] (u12) -- (d);
\draw[->,black] (u12') -- (d);

\end{tikzpicture}\subcaption{The equivalent unicast information flow problem.\\ \ }\label{FigCombinationNetworkflow}}\end{subfigure}
\caption{The virtual resources are emulatable if there exists a matching of size $\sum_{S\subseteq \mathcal{I}_1} \beta_S$ between the virtual resources and the information symbols that can emulate them. The light edges  connect the virtual resources to the information symbols that can emulate them, and the bold edges show a matching between them.}
\end{figure}
\begin{eqnarray}
\min_{ \substack{\forall \mathcal{S}\subseteq \mathcal{I}_1\\\Gamma\text{ superset saturated}}}\sum_{\mathcal{S}\in\Gamma}\alpha_\mathcal{S}+\sum_{\mathcal{S}\in\Gamma^c}\beta_\mathcal{S}.
\end{eqnarray}
It is easy to see that the flow value, $\sum_{\mathcal{S}\subseteq \mathcal{I}_1} \beta_\mathcal{S}$, is no more than this term provided that inequalities in \eqref{lb-CombNet-bm-match} hold and, therefore, there is an assignment of information symbols to virtual resources so that all virtual resources are emulatable.

\section{Proof of Lemma \ref{lb-CombNet-graph-strict}}
\label{ap-CombNet-graph-strict}
We prove that a non-trivial elementary compression over  $\mathbf{\Gamma}$ strictly reduces the total number of edges in its associated graph. Assume that a non-trivial elementary compression over  $\mathbf{\Gamma}$ yields a compressed multi-set $\mathbf{\Gamma}^\prime$, and the compression is performed using two sets $\Gamma_i$ and $\Gamma_j$. Consider the nodes associated with these two sets and track, throughout the compression, all edges that connect them to the other nodes of the associated graph. Let $\Gamma_k$ ($\neq \Gamma_j,\Gamma_j$) be an arbitrary node of the associated graph. 
We first show that for any such node, the total number of edges connecting it to $\Gamma_i$ and $\Gamma_j$ does not increase after the compression. This is summarized in the following.
\begin{itemize}
\item There is an edge $(\Gamma_i,\Gamma_k)$ and an edge $(\Gamma_j,\Gamma_k)$: In this case, no matter what the resulting graph $G_{\mathbf{\Gamma}^\prime}$ is after the compression, there cannot be more than two edges connecting $\Gamma_k$ to $\Gamma_i$ and $\Gamma_j$.
\item There is an edge $(\Gamma_i,\Gamma_k)$ but there is no edge $(\Gamma_j,\Gamma_k)$: Since there is no edge between $\Gamma_j$ and $\Gamma_k$, one of them is a subset of the other.
\begin{enumerate}
\item If $\Gamma_j\subseteq\Gamma_k$, then $\Gamma_i\cap\Gamma_j\subseteq\Gamma_k$ and there is, therefore, no edge between $\Gamma_k$ and $\Gamma_i\cap\Gamma_j$ after the compression. 
\item If otherwise $\Gamma_j\supseteq\Gamma_k$, then $\Gamma_i\cup\Gamma_j\supseteq\Gamma_k$ and there is, therefore, no edge between $\Gamma_k$ and $\Gamma_i\cup\Gamma_j$ after the compression. 
\end{enumerate}
\item There is no edge $(\Gamma_i,\Gamma_k)$ but an edge $(\Gamma_j,\Gamma_k)$: This case is similar to the previous case.
\item There is neither an edge $(\Gamma_i,\Gamma_k)$ nor an edge $(\Gamma_j,\Gamma_k)$: In this case, we have either of the following possibilities.
\begin{enumerate}
\item If $\Gamma_i\subseteq\Gamma_k$ and $\Gamma_j\subseteq\Gamma_k$, then both $\Gamma_i\cup\Gamma_j$ and $\Gamma_i\cap\Gamma_j$ are subsets of $\Gamma_k$ and there is no edge connecting $\Gamma_k$ to $\Gamma_i\cap\Gamma_j$ or $\Gamma_i\cup\Gamma_j$ over $G_{\mathbf{\Gamma}^\prime}$.
\item If $\Gamma_i\subseteq\Gamma_k$ and $\Gamma_j\supseteq\Gamma_k$, then $\Gamma_i\cup\Gamma_j\supseteq\Gamma_k$ and $\Gamma_i\cap\Gamma_j\subseteq\Gamma_k$ and there is, therefore, no edge connecting $\Gamma_k$ to $\Gamma_i\cap\Gamma_j$ or $\Gamma_i\cup\Gamma_j$ over  $G_{\mathbf{\Gamma}^\prime}$.
\item If $\Gamma_i\supseteq\Gamma_k$ and $\Gamma_j\subseteq\Gamma_k$, then similar to the previous case one concludes that there is no edge connecting $\Gamma_k$ to  $\Gamma_i\cap\Gamma_j$ or $\Gamma_i\cup\Gamma_j$ over  $G_{\mathbf{\Gamma}^\prime}$.
\item If $\Gamma_i\supseteq\Gamma_k$ and $\Gamma_j\supseteq\Gamma_k$, then both $\Gamma_i\cup\Gamma_j$ and $\Gamma_i\cap\Gamma_j$ are supersets of $\Gamma_k$ and there is, therefore, no edge connecting $\Gamma_k$ to their replacements $\Gamma_i\cap\Gamma_j$ or $\Gamma_i\cup\Gamma_j$ over  $G_{\mathbf{\Gamma}^\prime}$.
\end{enumerate}
\end{itemize} 
Besides, edges between $\Gamma_k$ and $\Gamma_{k^\prime}$, where $k,k^\prime\notin\{i,j\}$, remain unaffected.
Since the compression is non-trivial, nodes $\Gamma_i$ and $\Gamma_j$ have been connected over  $G_{\mathbf{\Gamma}}$ and are no longer connected over  $G_{\mathbf{\Gamma}^\prime}$. So, the total number of edges in $G_{\mathbf{\Gamma}}$ is strictly smaller than $G_{\mathbf{\Gamma}}$, and this concludes the proof.

\section{Proof of Lemma \ref{CombNet-count}}
\label{ap-CombNet-count}
Let $\mathbf{\Lambda}$ be a multiset of $2^{\{1,2,3\}}$ with saturated pattern, $\mathbf{\Gamma}$ be a multi-set of $2^{\{1,2,3\}}$ with standard pattern, and $\mathbf{\Gamma}$ and $\mathbf{\Lambda}$ be balanced and such that $\mathbf{\Lambda}\neq\mathbf{\Gamma}$.
We prove that no matter what $\mathbf{\Gamma}$ and $\mathbf{\Lambda}$ are, at least one of the cases of the table in Fig.~\ref{lb-CombNet-opt-table-decompress} must hold for $\mathbf{\Lambda}$, and therefore a non-trivial elementary decompression of $\mathbf{\Lambda}$ is feasible. 

Let us first count, in two different ways (once in $\mathbf{\Lambda}$ and once in $\mathbf{\Gamma}$), the number of times a set $\mathcal{S}\subseteq \mathcal{I}_1$ appears in the sets of multi-sets $\mathbf{\Lambda}$ and $\mathbf{\Gamma}$. First of all, define $n_\mathcal{S}$, $\mathcal{S}\subseteq \mathcal{I}_1$, to be the number of all sets $\Gamma\in \mathbf{\Gamma}$ that contain $\mathcal{S}$ as an element. One observes (from the standard pattern of multi-set  $\mathbf{\Gamma}$) that $n_{S}=\sum_{i\in S}n_{\{i\}}$. Similarly, define $m_\Lambda$, $\Lambda\in\mathbf{\Lambda}$, to be the number of times the set $\Lambda$ appears in the multi-set $\mathbf{\Lambda}$.
For simplicity of notation, we use $m_{\Lambda\cup}$ to denote the number of all sets $\Lambda^\prime$ in $\mathbf{\Lambda}$ which are of the form $\Lambda^\prime=\Lambda\cup\Sigma$ where $\Sigma\subseteq 2^{\mathcal{I}_1}$ is superset saturated, and $\Lambda\not\supseteq\Sigma$. For example, $m_{\{\{1,2\}\star\}\cup}$ counts the number of all sets such as $\{\{1,2\}\star\}$, $\{\{1,3\}\star,\{1,2\}\star\}$, $\{\{2,3\}\star,\{1,2\}\star\}$, $\{\{2,3\}\star,\{1,3\}\star,\{1,2\}\star\}$, and $\{\{3\}\star\{1,2\}\star\}$ in $\mathbf{\Lambda}$, but not $\{\{1\}\star\}$ or $\{\{1\}\star,\{3\}\star\}$.

Since multi-sets $\mathbf{\Gamma}$ and $\mathbf{\Lambda}$ are balanced, the number of sets in $\mathbf{\Gamma}$ which contain a set $\mathcal{S}$ is equal to the number of the sets in $\mathbf{\Lambda}$ which contain it. Thus, counting the number of sets in $\mathbf{\Gamma}$ and $\mathbf{\Lambda}$ which contain $\{i\}$ and $\{i,j\}$ as elements, we obtain the following relationship. In the following, we assume $(i,j,k)$ to be a permutation of $(1,2,3)$.
\begin{eqnarray}
 && n_{\{i\}}=m_{\{\{i\}\star\}\cup}\label{CombNet-ni}\\
&& n_{\{i,j\}}=m_{\{\{i\}\star\}\cup}+m_{\{\{j\}\star\}\cup}-m_{\{\{i\}\star,\{j\}\star\}\cup}+m_{\{\{i,j\}\star\}\cup}\label{CombNet-nij}
\end{eqnarray}
Since $n_{\{i,j\}}=n_{\{i\}}+n_{\{j\}}$, we conclude from \eqref{CombNet-ni} and \eqref{CombNet-nij} the following equation.
\begin{eqnarray}
\label{CombNet-mij}
 m_{\{\{i\}\star,\{j\}\star\}\cup}=m_{\{\{i,j\}\star\}\cup}
\end{eqnarray}
Similarly, counting the number of sets in $\mathbf{\Gamma}$ and $\mathbf{\Lambda}$ which contain $\{1,2,3\}$ as an element, we arrive at the following equation.
\begin{eqnarray}
 n_{\{1,2,3\}}&=&m_{\{\{1\}\star\}\cup}+m_{\{\{2\}\star\}\cup}+m_{\{\{3\}\star\}\cup}+m_{\{\{1,2\}\star\}\cup}+m_{\{\{1,3\}\star\}\cup}+m_{\{\{2,3\}\star\}\cup}+\nonumber\\
&&+m_{\{\{1,2,3\}\star\}\cup}-m_{\{\{1\}\star,\{2\}\star\}\cup}-m_{\{\{1\}\star,\{3\}\star\}\cup}-m_{\{\{2\}\star,\{3\}\star\}\cup}-m_{\{\{1\}\star,\{2,3\}\star\}\cup}+\nonumber\\
&&-m_{\{\{2\}\star,\{1,3\}\star\}\cup}-m_{\{\{3\}\star,\{1,2\}\star\}\cup}-m_{\{\{1,2\}\star,\{1,3\}\star\}\cup}-m_{\{\{1,2\}\star,\{2,3\}\star\}\cup}+\nonumber\\
&&-m_{\{\{1,3\}\star,\{2,3\}\star\}\cup}+m_{\{\{1,2\}\star,\{1,3\}\star,\{2,3\}\star\}\cup}+m_{\{\{1\}\star,\{2\}\star,\{3\}\star\}\cup}\label{CombNet-miju}
\end{eqnarray}
Using $n_{\{1,2,3\}}=n_{\{1\}}+n_{\{2\}}+n_{\{3\}}$, equation \eqref{CombNet-ni} and inserting  \eqref{CombNet-mij} into \eqref{CombNet-miju}, we obtain
\begin{eqnarray}
m_{\{\{1,2,3\}\star\}\cup}+m_{\{\{1\}\star,\{2\}\star,\{3\}\star\}\cup}
&=&m_{\{\{1\}\star,\{2,3\}\star\}\cup}+m_{\{\{2\}\star,\{1,3\}\star\}\cup}+m_{\{\{3\}\star,\{1,2\}\star\}\cup}+\nonumber\\
&&+m_{\{\{1,2\}\star,\{1,3\}\star\}\cup}+m_{\{\{1,2\}\star,\{2,3\}\star\}\cup}+m_{\{\{1,3\}\star,\{2,3\}\star\}\cup}\nonumber\\&&-m_{\{\{1,2\}\star,\{1,3\}\star,\{2,3\}\star\}\cup}.
\end{eqnarray}
Now we write each $m_{\Lambda\cup}$ in terms of $m_\Lambda$'s to derive the equation of our interest.
\begin{eqnarray}
m_{\{\{1,2,3\}\star\}}+m_{\{\{1\}\star,\{2\}\star,\{3\}\star\}}
&=&m_{\{\{1\}\star,\{2,3\}\star\}}+m_{\{\{2\}\star,\{1,3\}\star\}}+m_{\{\{3\}\star,\{1,2\}\star\}}+m_{\{\{1,2\}\star,\{1,3\}\star\}}+\nonumber\\
&&+m_{\{\{1,2\}\star,\{1,3\}\star,\{2,3\}\star\}}+m_{\{\{1,2\}\star,\{2,3\}\star\}}+m_{\{\{1,2\}\star,\{1,3\}\star,\{2,3\}\star\}}+\nonumber\\
&&+m_{\{\{1,3\}\star,\{2,3\}\star\}}+m_{\{\{1,2\}\star,\{1,3\}\star,\{2,3\}\star\}}-m_{\{\{1,2\}\star,\{1,3\}\star,\{2,3\}\star\}}\nonumber\\
&=&m_{\{\{1\}\star,\{2,3\}\star\}}+m_{\{\{2\}\star,\{1,3\}\star\}}+m_{\{\{3\}\star,\{1,2\}\star\}}+m_{\{\{1,2\}\star,\{1,3\}\star\}}+\nonumber\\
\label{CombNet-mijk}&&+m_{\{\{1,2\}\star,\{2,3\}\star\}}+m_{\{\{1,3\}\star,\{2,3\}\star\}}+2m_{\{\{1,2\}\star,\{1,3\}\star,\{2,3\}\star\}}
\end{eqnarray}
Observe from equality \eqref{CombNet-mijk} that if there is a non-zero term, $m_{\Lambda_1}$, on the left hand, there is at least one other non-zero term, $m_{\Lambda_2}$, on the right hand of the equality. No matter what $\Lambda_1$ and $\Lambda_2$ are, see that we are in one of the decompression cases in the table in Fig.~\ref{lb-CombNet-opt-table-decompress}. If both sides of equality \eqref{CombNet-mijk} are zero, then one concludes that  
$m_{\{\{i\}\star,\{j\}\star\}\cup}=m_{\{\{i\}\star,\{j\}\star\}}$ and $m_{\{\{i,j\}\star\}\cup}=m_{\{\{i,j\}\star\}}$ and therefore, by equation \eqref{CombNet-mij}, we have another equation of interest.
\begin{eqnarray}
\label{CombNet-mij2}
m_{\{\{i\}\star,\{j\}\star\}}=m_{\{\{i,j\}\star\}}
\end{eqnarray}
Again, if $m_{\{\{i\}\star,\{j\}\star\}}$ is non-zero so is $m_{\{\{i,j\}\star\}}$, and we have the first case described in the table of Fig.~\ref{lb-CombNet-opt-table-decompress}.

We have proved that a non-trivial elementary decomposition is possible unless all terms in \eqref{CombNet-mijk} and \eqref{CombNet-mij2} are zero, and all terms in \eqref{CombNet-mijk} and \eqref{CombNet-mij2} are zero only if $\mathbf{\Lambda}=\mathbf{\Gamma}$ which contradicts the hypothesis.


\section{Proof of Lemma \ref{lb-CombNetk2-relaxregion}}
\label{ap-CombNetk2-relaxregion}

Let us call the rate-region characterized in Theorem \ref{lb-CombNet-Theoremk2} (when $\mathcal{I}_1=\{1,2,3\}$) 
region $\mathcal{R}_1$ and the rate-region obtained from relaxing inequality 
\eqref{CombNetk2-achpos} to inequality \eqref{lb-CombNet-relax} (when $\mathcal{I}_1=\{1,2,3\}$) region $\mathcal{R}_2$. 
Clearly, $\mathcal{R}_1\subseteq \mathcal{R}_2$. It is, therefore, sufficient to show that 
$\mathcal{R}_2\subseteq\mathcal{R}_1$. Both rate-regions $\mathcal{R}_1$ and $\mathcal{R}_2$ are in terms of feasibility 
problems. In this sense, rate pair $(R_1,R_2)$ belongs to $\mathcal{R}_1$ if and only if feasibility problem 
$1$ (characterized by inequalities \eqref{CombNetk2-achpos}-\eqref{CombNetk2-achR1+R2or}) is feasible. Similarly, 
rate pair $(R_1,R_2)$ belongs to $\mathcal{R}_2$ if and only if feasibility problem $2$ (characterized by inequalities \eqref{lb-CombNet-relax}, \eqref{CombNetk2-achR2or}-\eqref{CombNetk2-achR1+R2or}) is feasible. 

In order to show that $\mathcal{R}_2\subseteq\mathcal{R}_1$, 
we show that if $(R_1,R_2)$ is such that there exists a solution, $\alpha_\mathcal{S}$, $\mathcal{S}\subseteq \mathcal{I}_1$, to feasibility problem $2$, then there also exists a solution $\alpha^\prime_\mathcal{S}$, $\mathcal{S}\subseteq \mathcal{I}_1$, to feasibility problem $1$. Note that problem $1$ varies from problem $2$ only in the non-negativity constraints on parameters  $\alpha^\prime_\mathcal{S}$, $\phi\neq \mathcal{S}\subseteq \mathcal{I}_1$. The goal is to construct parameters $\alpha^\prime_\mathcal{S}$ (from parameters $\alpha_\mathcal{S}$) such that besides satisfying constraints \eqref{CombNetk2-achR2or}-\eqref{CombNetk2-achR1+R2or}, they all become non-negative except for $\alpha_\phi$. 

We prove the existence of a solution $\alpha_\mathcal{S}^\prime$, $\mathcal{S}\subseteq 2^{\mathcal{I}_1}$, by construction. This construction is done recursively.
We start from a solution to feasibility problem $2$, $\alpha_\mathcal{S}$, $\mathcal{S}\subseteq \mathcal{I}_1$, and, at each step, we propose a solution $\alpha^\prime_\mathcal{S}$, $\mathcal{S}\subseteq \mathcal{I}_1$, that is still a solution to feasibility problem $2$ but is ``strictly less negative" (excluding $\alpha_\phi$). So after enough number of steps, we end up with a set of parameters $\alpha^\prime_\mathcal{S}$, $\mathcal{S}\subseteq \mathcal{I}_1$, that satisfies 
\eqref{CombNetk2-achR2or}-\eqref{CombNetk2-achR1+R2or} and also satisfies the non-negativity constraints in \eqref{CombNetk2-achpos}. 

\begin{itemize}
\item [1.]
First, we set $\alpha^\prime_\mathcal{S}=\alpha_\mathcal{S}$ for all $\mathcal{S}\subseteq\{1,2,3\}$. 
All $\alpha^\prime_\mathcal{S}$, $\mathcal{S}\subseteq \{1,2,3\}$, satisfy \eqref{CombNetk2-achR2or}-\eqref{CombNetk2-achR1+R2or}, \eqref{lb-CombNet-relax}. $\alpha_{\{1,2,3\}}$ is ensured to be non-negative by \eqref{lb-CombNet-relax} (for $\Lambda=\{\{1,2,3\}\}$).
\item [2.] We then choose non-negative parameters $\alpha^\prime_{\{i,j\}}$, $i,j\in\{1,2,3\}$. Without loss of generality, take the following three cases, and set $\alpha^\prime_\mathcal{S}$'s as suggested (other cases are dealt with similarly).
 \begin{enumerate}
 \item[(a)] If $\alpha_{\{1,2\}}<0$ and $\alpha_{\{1,3\}}<0$ and $\alpha_{\{2,3\}}<0$:\\
set $\alpha^\prime_{\{1,2,3\}}=\alpha_{\{1,2,3\}}+\alpha_{\{1,2\}}+\alpha_{\{1,3\}}+\alpha_{\{2,3\}}$, $\alpha^\prime_{\{1,2\}}=\alpha^\prime_{\{1,3\}}=\alpha^\prime_{\{2,3\}}=0$,   $\alpha^\prime_{\{i\}}=\alpha_{\{i\}}$, for $i=1,2,3$, and $\alpha_\phi^\prime=\alpha_\phi$.
One can verify that all $\alpha^\prime_\mathcal{S}$, $\mathcal{S}\subseteq \{1,2,3\}$, satisfy \eqref{CombNetk2-achR2or}-\eqref{CombNetk2-achR1+R2or}, \eqref{lb-CombNet-relax}. We outline this verification in the following. \eqref{CombNetk2-achR2or}, \eqref{CombNetk2-achr1or},  \eqref{CombNetk2-achR1+R2or}, and \eqref{lb-CombNet-relax} clearly hold. We show how to verify \eqref{CombNetk2-ach1or} through an example, say, $\Lambda=\{\{1,2\},\{1,2,3\}\}$ and some $p\in\mathcal{I}_2$. From feasibility of problem $2$, we know that $\alpha_\mathcal{S}$ is such that it satisfies  \eqref{CombNetk2-ach1or} for all superset saturated subsets $\Lambda$, and in particular for $\Lambda=\{\{1,2\}\star,\{1,3\}\star,\{2,3\}\star\}$. So, we have
\begin{align}
R_2&\leq \alpha_{\{1,2,3\}}+\alpha_{\{1,2\}}+\alpha_{\{1,3\}}+\alpha_{\{2,3\}}+\sum_{\mathcal{S}\in \{\phi,\{1\},\{2\},\{3\}\}}|\e_{\mathcal{S},p}|\\
&= \alpha^\prime_{\{1,2,3\}}+\sum_{\mathcal{S}\in \{\phi,\{1\},\{2\},\{3\}\}}|\e_{\mathcal{S},p}|\\
&\leq \alpha^\prime_{\{1,2,3\}}+\alpha^\prime_{\{1,2\}}+\sum_{\mathcal{S}\in \{\{1,3\},\{2,3\},\phi,\{1\},\{2\},\{3\}\}}|\e_{\mathcal{S},p}|.
\end{align}
 \item[(b)] If $\alpha_{\{1,2\}}<0$ and $\alpha_{\{1,3\}}<0$:\\
 set $\alpha^\prime_{\{1,2,3\}}=\alpha_{\{1,2,3\}}+\alpha_{\{1,2\}}+\alpha_{\{1,3\}}$, $\alpha^\prime_{\{1,2\}}=\alpha^\prime_{\{1,3\}}=0$, $\alpha_{\{2,3\}}^\prime=\alpha_{\{2,3\}}$, $\alpha^\prime_{\{i\}}=\alpha_{\{i\}}$ for $i=1,2,3$, and $\alpha_\phi^\prime=\alpha_\phi$.
Verify that all $\alpha^\prime_\mathcal{S}$, $\mathcal{S}\subseteq \{1,2,3\}$, satisfy \eqref{CombNetk2-achR2or}-\eqref{CombNetk2-achR1+R2or}, \eqref{lb-CombNet-relax}.
 \item[(c)] If $\alpha_{\{1,2\}}<0$:\\
set $\alpha^\prime_{\{1,2,3\}}=\alpha_{\{1,2,3\}}+\alpha_{\{1,2\}}$, $\alpha^\prime_{\{1,2\}}=0$, $\alpha_{\{1,3\}}^\prime=\alpha_{\{1,3\}}$, $\alpha_{\{2,3\}}^\prime=\alpha_{\{2,3\}}$, $\alpha^\prime_{\{i\}}=\alpha_{\{i\}}$ for $i=1,2,3$, and $\alpha_\phi^\prime=\alpha_\phi$.
Verify that all $\alpha^\prime_\mathcal{S}$, $\mathcal{S}\subseteq \{1,2,3\}$, satisfy \eqref{CombNetk2-achR2or}-\eqref{CombNetk2-achR1+R2or}, \eqref{lb-CombNet-relax}.
 \end{enumerate}
\item [3.] Finally, we choose non-negative parameters $\alpha^\prime_{\{i\}}$, $i\in\{1,2,3\}$.
This is done recursively, following the procedure below, for each $\alpha^\prime_i<0$, until all $\alpha^\prime_{\{i\}}$, $i=1,2,3$, are non-negative. $\delta$ is assumed a small enough positive number.
\begin{itemize}
\item[(a)] If $\alpha^\prime_{\{i,j\}},\alpha^\prime_{\{i,k\}}>0$:\\
set $\alpha^\prime_{\{i\}}=\alpha^\prime_{\{i\}}+\delta$, $\alpha^\prime_{\{i,j\}}=\alpha^\prime_{\{i,j\}}-\delta$, $\alpha^\prime_{\{i,k\}}=\alpha^\prime_{\{i,k\}}-\delta$, $\alpha^\prime_{\{1,2,3\}}=\alpha^\prime_{\{1,2,3\}}+\delta$, and keep the rest of $\alpha^\prime_\mathcal{S}$'s unchanged.
Verify that all $\alpha^\prime_\mathcal{S}$, $\mathcal{S}\subseteq \{1,2,3\}$, satisfy \eqref{CombNetk2-achR2or}-\eqref{CombNetk2-achR1+R2or}, \eqref{lb-CombNet-relax}. We show \eqref{CombNetk2-ach1or} for one example: $\Lambda=\{\{1,2\}\star,\{1,3\}\star\}$ and some $p\in\mathcal{I}_2$. Since the $\alpha_\mathcal{S}$'s we start with satisfy \eqref{CombNetk2-ach1or} for any superset saturated $\Lambda$, and in particular for $\Lambda=\{\{1\}\star\}$, we have
\begin{align}
R_2&\leq \alpha_{\{1\}}+\alpha_{\{1,2\}}+\alpha_{\{1,3\}}+\alpha_{\{1,2,3\}}+\sum_{\mathcal{S}\in \{\phi,\{2\},\{3\},\{2,3\}\}}|\e_{\mathcal{S},p}|\\
&\leq \alpha^\prime_{\{1\}}+\alpha^\prime_{\{1,2\}}+\alpha^\prime_{\{1,3\}}+\alpha^\prime_{\{1,2,3\}}+\sum_{\mathcal{S}\in \{\phi,\{2\},\{3\},\{2,3\}\}}|\e_{\mathcal{S},p}|\\
&\leq \alpha^\prime_{\{1,2\}}+\alpha^\prime_{\{1,3\}}+\alpha^\prime_{\{1,2,3\}}+\sum_{\mathcal{S}\in \{\phi,\{2\},\{3\},\{2,3\}\}}|\e_{\mathcal{S},p}|,
\end{align}
where the last inequality is because $\alpha^\prime_{\{1\}}$ is either still negative, or has just become zero after adding the small $\delta$.
\item[(b)] If $\alpha^\prime_{\{i,j\}}=0$, $\alpha^\prime_{\{i,k\}}>0$:\\
set $\alpha^\prime_{\{i\}}=\alpha^\prime_{\{i\}}+\delta$, $\alpha^\prime_{\{i,k\}}=\alpha^\prime_{\{i,k\}}-\delta$, and keep the rest of $\alpha^\prime_\mathcal{S}$'s unchanged.
Verify that all $\alpha^\prime_\mathcal{S}$, $\mathcal{S}\subseteq \{1,2,3\}$, satisfy \eqref{CombNetk2-achR2or}-\eqref{CombNetk2-achR1+R2or}, \eqref{lb-CombNet-relax}.
\item[(c)] If $\alpha^\prime_{\{i,j\}}>0$, $\alpha^\prime_{\{i,k\}}=0$:\\
set $\alpha^\prime_{\{i\}}=\alpha^\prime_{\{i\}}+\delta$, $\alpha^\prime_{\{i,j\}}=\alpha^\prime_{\{i,j\}}-\delta$, and keep the rest of $\alpha^\prime_\mathcal{S}$'s unchanged.
Verify that all $\alpha^\prime_\mathcal{S}$, $\mathcal{S}\subseteq \{1,2,3\}$, satisfy \eqref{CombNetk2-achR2or}-\eqref{CombNetk2-achR1+R2or}, \eqref{lb-CombNet-relax}.
\item[(d)] If $\alpha^\prime_{\{i,j\}}=0$, $\alpha^\prime_{\{i,k\}}=0$:\\
set $\alpha^\prime_{\{i\}}=\alpha^\prime_{\{i\}}+\delta$, $\alpha^\prime_{\{1,2,3\}}=\alpha^\prime_{\{1,2,3\}}-\delta$, and keep the rest of $\alpha^\prime_\mathcal{S}$'s unchanged.
Verify that all $\alpha^\prime_\mathcal{S}$, $\mathcal{S}\subseteq \{1,2,3\}$, satisfy \eqref{CombNetk2-achR2or}-\eqref{CombNetk2-achR1+R2or}, \eqref{lb-CombNet-relax}.
\end{itemize}
\end{itemize}
Note that in step $1$, we obtain a solution to feasibility problem $2$ with $\alpha_{\{1,2,3\}}\geq 0$. After step $2$, the solution is such that  $\alpha_{\{1,2\}}, \alpha_{\{1,3\}}, \alpha_{\{2,3\}}, \alpha_{\{1,2,3\}}$ are all non-negative. In step $3$, after each iteration,  $\alpha_{\{1,2\}},\alpha_{\{1,3\}},\alpha_{\{2,3\}},\alpha_{\{1,2,3\}}$ all remain non-negative and at the same time one negative $\alpha_{\{i\}}$ is increased. So after step $3$, all parameters $\alpha_\mathcal{S}$, $\phi\neq S\subseteq \{1,2,3\},$ become non-negative. This is the solution to feasibility problem $1$ that we were looking for.




\ifCLASSOPTIONcaptionsoff
  \newpage
\fi



\bibliographystyle{IEEEtran}
\bibliography{bibliography}

\begin{thebibliography}{10}
\providecommand{\url}[1]{#1}
\csname url@samestyle\endcsname
\providecommand{\newblock}{\relax}
\providecommand{\bibinfo}[2]{#2}
\providecommand{\BIBentrySTDinterwordspacing}{\spaceskip=0pt\relax}
\providecommand{\BIBentryALTinterwordstretchfactor}{4}
\providecommand{\BIBentryALTinterwordspacing}{\spaceskip=\fontdimen2\font plus
\BIBentryALTinterwordstretchfactor\fontdimen3\font minus
  \fontdimen4\font\relax}
\providecommand{\BIBforeignlanguage}[2]{{%
\expandafter\ifx\csname l@#1\endcsname\relax
\typeout{** WARNING: IEEEtran.bst: No hyphenation pattern has been}%
\typeout{** loaded for the language `#1'. Using the pattern for}%
\typeout{** the default language instead.}%
\else
\language=\csname l@#1\endcsname
\fi
#2}}
\providecommand{\BIBdecl}{\relax}
\BIBdecl

\bibitem{Cover95}
T.~Cover, ``Comments on broadcast channels,'' \emph{IEEE Trans. Inf. Theory},
  vol.~44, no.~6, pp. 2524--2530, Oct. 1998.

\bibitem{AhlswedeCaiLiYeung00}
R.~Ahlswede, N.~Cai, S.-Y. Li, and R.~Yeung, ``Network information flow,''
  \emph{IEEE Trans. Inf. Theory}, vol.~46, no.~4, pp. 1204 --1216, July 2000.

\bibitem{WeingartenSteinbergShamai06}
H.~Weingarten, Y.~Steinberg, and S.~Shamai, ``The capacity region of the
  gaussian multiple-input multiple-output broadcast channel,'' \emph{IEEE
  Trans. Inf. Theory}, vol.~52, no.~9, pp. 3936 --3964, Sept. 2006.

\bibitem{AvestimehrDiggaviTse11}
A.~Avestimehr, S.~Diggavi, and D.~Tse, ``Wireless network information flow: A
  deterministic approach,'' \emph{IEEE Trans. Inf. Theory}, vol.~57, no.~4, pp.
  1872 --1905, Apr. 2011.

\bibitem{GengNair14}
Y.~Geng and C.~Nair, ``The capacity region of the two-receiver gaussian vector
  broadcast channel with private and common messages,'' \emph{IEEE Trans. Inf.
  Theory}, vol.~60, no.~4, pp. 2087--2104, Apr. 2014.

\bibitem{Marton79}
K.~Marton, ``A coding theorem for the discrete memoryless broadcast channel,''
  \emph{IEEE Trans. Inf. Theory}, vol.~25, no.~3, pp. 306 -- 311, May 1979.

\bibitem{GohariElGamalAnantheram10}
A.~Gohari, A.~El~Gamal, and V.~Anantharam, ``On an outer bound and an inner
  bound for the general broadcast channel,'' in \emph{IEEE Int. Symp. Inf.
  Theory}, June 2010, pp. 540 --544.

\bibitem{NairElGamal07}
C.~Nair and A.~El~Gamal, ``An outer bound to the capacity region of the
  broadcast channel,'' \emph{IEEE Trans. Inf. Theory}, vol.~53, no.~1, pp. 350
  --355, Jan. 2007.

\bibitem{LiangKramerPoor11}
Y.~Liang, G.~Kramer, and H.~Poor, ``On the equivalence of two achievable
  regions for the broadcast channel,'' \emph{IEEE Trans. Inf. Theory}, vol.~57,
  no.~1, pp. 95 --100, Jan. 2011.

\bibitem{ChanGrant07}
T.~Chan and A.~Grant, ``Entropy vectors and network codes,'' in \emph{IEEE Int.
  Symp. Inf. Theory}, June 2007, pp. 1586--1590.

\bibitem{Yeung08}
R.~W. Yeung, \emph{Information Theory and Network Coding}, 1st~ed.\hskip 1em
  plus 0.5em minus 0.4em\relax Springer Publishing Company, 2008.

\bibitem{SongYeungCai03}
L.~Song, R.~Yeung, and N.~Cai, ``Zero-error network coding for acyclic
  networks,'' \emph{IEEE Trans. Inf. Theory}, vol.~49, no.~12, pp. 3129--3139,
  Dec. 2003.

\bibitem{Cover72}
T.~Cover, ``Broadcast channels,'' \emph{IEEE Trans. Inf. Theory}, vol.~18,
  no.~1, pp. 2 -- 14, Jan. 1972.

\bibitem{Bergmans73}
P.~Bergmans, ``Random coding theorem for broadcast channels with degraded
  components,'' \emph{IEEE Trans. Inf. Theory}, vol.~19, no.~2, pp. 197 -- 207,
  Mar. 1973.

\bibitem{Marton77}
K.~Marton, ``The capacity region of deterministic broadcast channels,'' in
  \emph{IEEE Int. Symp. Inf. Theory}, 1977.

\bibitem{Pinsker78}
M.~S. Pinsker, ``The capacity region of noiseless broadcast channels,''
  \emph{Probl. Inf. Transm}, vol.~14, pp. 97 -- 102, 1974.

\bibitem{KornerMarton77}
J.~Korner and K.~Marton, ``General broadcast channels with degraded message
  sets,'' \emph{IEEE Trans. Inf. Theory}, vol.~23, no.~1, pp. 60 -- 64, Jan.
  1977.

\bibitem{LiYeungCai03}
S.-Y. Li, R.~Yeung, and N.~Cai, ``Linear network coding,'' \emph{IEEE Trans.
  Inf. Theory}, vol.~49, no.~2, pp. 371 --381, Feb. 2003.

\bibitem{KoetterMedard03}
R.~Koetter and M.~Medard, ``An algebraic approach to network coding,''
  \emph{IEEE/ACM Trans. Networking}, vol.~11, no.~5, pp. 782 -- 795, Oct. 2003.

\bibitem{HanKobayashi81}
T.~Han and K.~Kobayashi, ``A new achievable rate region for the interference
  channel,'' \emph{IEEE Trans. Inf. Theory}, vol.~27, no.~1, pp. 49--60, Jan.
  1981.

\bibitem{ChongMotaniGargElGamal08}
H.~F. Chong, M.~Motani, H.~K. Garg, and H.~{El Gamal}, ``On the {Han-Kobayashi}
  region for the interference channel,'' \emph{IEEE Trans. Inf. Theory},
  vol.~57, no.~7, pp. 3188--3195, July 2008.

\bibitem{NairElGamal09}
C.~Nair and A.~{El Gamal}, ``The capacity region of a class of 3-receiver
  broadcast channels with degraded message sets,'' \emph{IEEE Trans. Inf.
  Theory}, vol.~55, no.~10, pp. 4479--4493, Oct. 2009.

\bibitem{SaeediPrabhakaranDiggavi12}
S.~Saeedi~Bidokhti and V.~Prabhakaran, ``Is non-unique decoding necessary?''
  \emph{IEEE Trans. Inf. Theory}, vol.~60, no.~5, pp. 2594--2610, May 2014.

\bibitem{RamamoorthyWesel09}
A.~Ramamoorthy and R.~D. Wesel, ``The single source two terminal network with
  network coding,'' \emph{ArXiv e-prints}, 2009, available on
  http://arxiv.org/abs/0908.2847.

\bibitem{NgaiYeung04a}
C.~Ngai and R.~Yeung, ``Multisource network coding with two sinks,'' in
  \emph{Int. Conf. Commun., Circuits and Sys}, vol.~1, June 2004, pp. 34 -- 37
  Vol.1.

\bibitem{ErezFeder03}
E.~Erez and M.~Feder, ``Capacity region and network codes for two receivers
  multicast with private and common data,'' in \emph{Workshop on Coding,
  Cryptography and Combinatorics}, 2003.

\bibitem{DiggaviTse06}
S.~Diggavi and D.~Tse, ``On opportunistic codes and broadcast codes with
  degraded message sets,'' in \emph{IEEE Inf. Theory Workshop}, Mar. 2006.

\bibitem{BoradeZhengTrott07}
S.~Borade, L.~Zheng, and M.~Trott, ``Multilevel broadcast networks,'' in
  \emph{IEEE Int. Symp. Inf. Theory}, June 2007, pp. 1151 --1155.

\bibitem{PrabhakaranDiggaviTse07}
V.~Prabhakaran, S.~Diggavi, and D.~Tse, ``Broadcasting with degraded message
  sets: A deterministic approach,'' in \emph{Annual Allerton Conf. on Commun.,
  Control and Computing}, 2007.

\bibitem{SaeediDiggaviFragouliPrabhakaran09}
S.~Saeedi~Bidokhti, S.~Diggavi, C.~Fragouli, and V.~Prabhakaran, ``On degraded
  two message set broadcasting,'' in \emph{IEEE Inf. Theory Workshop}, Oct.
  2009, pp. 406 --410.

\bibitem{GheorghiuSaeediFragouliToledo11}
S.~Gheorghiu, S.~Saeedi~Bidokhti, C.~Fragouli, and A.~Toledo, ``Degraded
  multicasting with network coding over the combination network,'' in
  \emph{IEEE Int. Symp. Network Coding}, 2011, available on
  http://infoscience.epfl.ch/record/174452.

\bibitem{Saeedi12}
S.~{Saeedi Bidokhti}, \emph{Broadcasting and Multicasting Nested Message
  Sets}.\hskip 1em plus 0.5em minus 0.4em\relax {PhD dissertation, Ecole
  Polytechnique F\'{e}d\'{e}ral de Lausanne}, 2012.

\bibitem{GLT}
X.~Xu and Y.~L. Guan, ``Joint routing and random network coding for
  multi-session networks,'' in \emph{IEEE Int. Conf. Networks}, Dec. 2013, pp.
  1--5.

\bibitem{ChanGrant12}
T.~H. Chan and A.~J. Grant, ``Network coding capacity regions via entropy
  functions,'' \emph{ArXiv e-prints}, 2012, available on
  http://arxiv.org/abs/1201.1062.

\bibitem{NgaiYeung04}
C.~K. Ngai and R.~Yeung, ``Network coding gain of combination networks,'' in
  \emph{IEEE Inf. Theory Workshop}, Oct. 2004, pp. 283 -- 287.

\bibitem{BalisterBollobas07}
\BIBentryALTinterwordspacing
P.~{Balister} and B.~{Bollob{\'a}s},
  ``\BIBforeignlanguage{English}{Projections, entropy and sumsets},''
  \emph{\BIBforeignlanguage{English}{Combinatorica}}, vol.~32, no.~2, pp.
  125--141, 2012. [Online]. Available:
  \url{http://dx.doi.org/10.1007/s00493-012-2453-1}
\BIBentrySTDinterwordspacing

\bibitem{WillemsMeulen85}
F.~Willems and E.~van~der Meulen, ``The discrete memoryless multiple-access
  channel with cribbing encoders,'' \emph{IEEE Trans. Inf. Theory}, vol.~31,
  no.~3, pp. 313 -- 327, May 1985.

\bibitem{kcl07}
\BIBentryALTinterwordspacing
C.~Fragouli and E.~Soljanin, ``Network coding fundamentals,'' \emph{Foundations
  and Trends in Networking}, vol.~2, no.~1, pp. 1--133, 2007. [Online].
  Available: \url{http://dx.doi.org/10.1561/1300000003}
\BIBentrySTDinterwordspacing

\bibitem{ElGamalKim}
A.~{El Gamal} and Y.~H. Kim, \emph{Network Information Theory}.\hskip 1em plus
  0.5em minus 0.4em\relax Cambridge University Press, 2011.

\end{thebibliography}

\end{document}